\documentclass[a4paper,UKenglish,cleveref, autoref, thm-restate]{lipics-v2021}

\usepackage{thmtools} 
\usepackage{thm-restate}
\usepackage{amsfonts,amssymb,thmtools}
\usepackage{bbm}
\usepackage{stmaryrd}
\usepackage{bbding}
\usepackage{pifont}
\usepackage{dsfont}

\usepackage{xcolor, soul}

\usepackage{multicol}

\usepackage{graphicx}

\usepackage{paralist}
\usepackage{enumitem}
\usepackage{wrapfig}
\usepackage{xspace}

\usepackage{tikz}
\usetikzlibrary{arrows,positioning,shapes,patterns,decorations.pathreplacing,automata,backgrounds,petri,fit,calc,shapes.multipart}

\usepackage{hyperref}
\usepackage{cleveref}

\nolinenumbers
\numberwithin{theorem}{section}

\counterwithin{theorem}{section}
\counterwithin{lemma}{section}
\counterwithin{corollary}{section}
\counterwithin{proposition}{section}
\counterwithin{exercise}{section}
\counterwithin{definition}{section}
\counterwithin{conjecture}{section}
\counterwithin{remark}{section}
\counterwithin{note}{section}
\counterwithin{claim}{section}
\counterwithin{figure}{section}
\counterwithin{example}{section}

\theoremstyle{plain}

\theoremstyle{definition}
\newtheorem{requirement}{Requirement}
\newtheorem{mydefinition}[theorem]{Definition}

\theoremstyle{remark} 
\newtheorem{myexample}[theorem]{Example}



  \DeclareFontFamily{U}{dutchcal}{\skewchar \font =45}
  \DeclareFontShape{U}{dutchcal}{m}{n}{
    <-> dutchcal-r}{}
  \DeclareFontShape{U}{dutchcal}{b}{n}{
    <-> dutchcal-b}{}
  \DeclareMathAlphabet{\mdutchcal}{U}{dutchcal}{m}{n}
  \SetMathAlphabet{\mdutchcal}{bold}{U}{dutchcal}{b}{n}
  \DeclareMathAlphabet{\mdutchbcal} {U}{dutchcal}{b}{n}

\newcommand{\commentout}[1]{}

\newcommand{\e}{\varepsilon}
\newcommand{\norm}[1]{\lVert{#1}\rVert}
\newcommand{\bb}[1]{\mathbb{#1}}

\newcommand{\cost}{\mathrm{cost}}

\newcommand{\defeq}{\stackrel{\text{def}}{=}}
\newcommand{\nonnegR}{\mathbb{R}_{\ge 0}}
\newcommand{\powerset}{\ensuremath{\mdutchcal{P}}}

\newcommand{\diam}{\mathrm{diam}}

\renewcommand{\vec}[1]{\mathbf{#1}}
\newcommand{\edgraph}[1]{\ensuremath{{G}_{{#1}}}}
\newcommand{\cyclicedgraph}[1]{\ensuremath{{G}^{\circlearrowleft}_{#1}}}

\newcommand{\dir}{\mathrm{dir}}
\newcommand{\LP}{\mathrm{LP}}
\newcommand{\OPT}{\mathrm{OPT}}
\newcommand{\Cone}{\mathrm{Cone}}

\newcommand{\setX}{\textsc{x}}
\newcommand{\setY}{\textsc{y}}
\newcommand{\ep}{p}
\newcommand{\indicator}{\mathbbm{1}}

\newcommand{\Parikh}{\textsl{Parikh}}

\newcommand{\ve}[1]{\textbf{#1}}
\newcommand{\vecInd}[2]{\textbf{#1}_{#2}}
\newcommand{\vecIndEnt}[3]{\textbf{#1}_{#2}[{#3}]}

\newcommand{\langop}[1]{\mdutchcal{#1}}

\newcommand{\essence}{\langop{E}}

\newcommand{\essenceFa}{\essence_{\textsc{fa}}}
\newcommand{\essenceRegex}{\essence_{\textsc{reg}}}

\newcommand{\equivD}{\equiv_{\metric{D}}}

\newcommand{\elementmetric}[1]{\textsf{#1}}

\newcommand{\ned}{\elementmetric{ned}}
\newcommand{\ed}{\elementmetric{ed}}
\newcommand{\ged}{\elementmetric{ged}}
\newcommand{\ced}{\elementmetric{ced}}

\newcommand{\prf}{\elementmetric{prf}}

\newcommand{\hamming}{\elementmetric{hamming}}
\newcommand{\lcs}{\elementmetric{lcs}}

\newcommand{\wordlang}[1]{\textbf{\textsf{#1}}} 
\newcommand{\metric}[1]{\mathcal{#1}} 
\newcommand{\lift}[1]{\mathbb{#1}} 
\newcommand{\rlift}[1]{\mathbb{#1}^{\,\vartriangleright}} 

\newcommand{\sema}[1]{\llbracket #1 \rrbracket}

\newcommand{\dr}{\ensuremath{\rlift{H}_{\elementmetric{d}}}}
\newcommand{\dar}{\ensuremath{\rlift{AH}_{\elementmetric{d}}}}
\newcommand{\dAH}{\ensuremath{\lift{AH}_{\elementmetric{d}}}}

\newcommand{\edH}{\ensuremath{\lift{H}_{\ed}}}

\newcommand{\nedar}{\ensuremath{\rlift{AH}_{\ned}}}
\newcommand{\nedH}{\ensuremath{\lift{H}_{\ned}}}
\newcommand{\nedAH}{\ensuremath{\lift{AH}_{\ned}}}

\newcommand{\gedAH}{\ensuremath{\lift{AH}_{\ged}}}
\newcommand{\cedAH}{\ensuremath{\lift{AH}_{\ced}}}

\newcommand{\iined}{\ensuremath{\lift{I}_{\ned}}}
\newcommand{\iid}{\ensuremath{\lift{I}_{\elementmetric{d}}}}
\newcommand{\ssned}{\ensuremath{\lift{S}_{\ned}}}
\newcommand{\ssd}{\ensuremath{\lift{S}_{\elementmetric{d}}}}

\newcommand{\hned}{\ensuremath{\lift{H}_{\ned}}}
\newcommand{\hed}{\ensuremath{\lift{H}_{\ed}}}
\newcommand{\ahned}{\ensuremath{\lift{AH}_{\ned}}}

\newcommand{\hprf}{\ensuremath{\lift{H}_{\prf}}}

\newcommand{\jac}{\ensuremath{\metric{J}}}
\newcommand{\jaccard}{{\jac}}
\newcommand{\jaccardind}{{\bar{\jac}}}
\newcommand{\acost}{\ensuremath{\metric{AC}^{\vartriangleright}}}
\newcommand{\acostsym}{\ensuremath{\metric{AC}}}

\newcommand{\cesjac}{\ensuremath{\metric{J}_\mathcal{C}}}
\newcommand{\cj}{\cesjac}
\newcommand{\dsj}{\ensuremath{\metric{J}_\mathcal{DS}}}
\newcommand{\ir}{\ensuremath{\metric{I}}}

\newcommand{\predicate}{\varphi}

\newcommand{\dgamma}{d}
\newcommand{\swap}[2]{\ensuremath{\left[\begin{smallmatrix} #1 \\ #2 \end{smallmatrix}\right]}}
\newcommand{\wgt}{\textsf{wgt}}
\newcommand{\len}{\textsf{len}}

\newcommand{\pspace}{\textsc{Pspace}\xspace}
\newcommand{\coNExp}{\textsc{coNExp}\xspace}
\newcommand{\Exp}{\textsc{Exp}\xspace}

\newcommand{\gedar}{\ensuremath{\rlift{AH}_{\ged}}}%
\newcommand{\cedar}{\ensuremath{\rlift{AH}_{\ced}}}%

\newcommand{\xonecol}{black}
\newcommand{\xtwocol}{black}
\newcommand{\yonecol}{black}
\newcommand{\ytwocol}{black}
\newcommand{\ythreecol}{black}

\newcommand{\wordSlice}[2]{#1_{\textsc{#2}}}
\newcommand{\wPref}[1]{\wordSlice{#1}{pref}}
\newcommand{\wSuff}[1]{\wordSlice{#1}{suff}}

\title{Asymptotic Hausdorff and Language Similarity}

\titlerunning{Asymptotic Hausdorff and Language Similarity} 

\author{Dana Fisman}{Institute for the Theory of Computing, Stein Faculty of Computer and Information Science, Ben Gurion University, Israel}{dana@bgu.ac.il}{https://orcid.org/0000-0002-6015-4170}{}

\author{Gal Meirom}{Institute for the Theory of Computing, Stein Faculty of Computer and Information Science, Ben Gurion University, Israel}{galmeirom@gmail.com}{https://orcid.org/0009-0009-4984-4179}{Supported by ISF grant 2507/21 and Frankel Center for Computer Science, BGU}

\authorrunning{Dana Fisman and Gal Meirom}



\funding{}

\acknowledgements{We thank Dror Fried, Guy Ofek, Omer Shachar and Gera Weiss for helpful comments on an early draft of this paper.}

\Copyright{Dana Fisman, Gal Meirom}

\EventEditors{Sayan Bhattacharya, Danupon Nanongkai, Michael Benedikt, and Gabriele Puppis}
\EventNoEds{4}
\EventLongTitle{53rd International Colloquium on Automata, Languages, and Programming (ICALP 2026)}
\EventShortTitle{ICALP 2026}
\EventAcronym{ICALP}
\EventYear{2026}
\EventDate{July 7--10, 2026}
\EventLocation{Royal Holloway, University of London, Egham, United Kingdom}
\EventLogo{}
\SeriesVolume{374}
\ArticleNo{188}

\begin{document}

\maketitle

\begin{abstract}
We introduce the \emph{Asymptotic Hausdorff} lifting, denoted $\lift{AH}_{d}$, a general method for lifting an element-level metric $d$ to a (pseudo-) metric on sets, that captures asymptotic similarity in infinite domains equipped with a notion of size.
The construction is designed to be insensitive to finite deviations and to avoid the limitations of classical Hausdorff-based approaches, which are often overly sensitive to outliers and fail to reflect asymptotic behavior.

Formal languages provide a central motivating instance of this framework, where elements are words and sets are languages.
When applied to normalized edit distances, the Asymptotic Hausdorff lifting yields metric-valued distances between languages that reflect asymptotic edit behavior while preserving metric structure.
We study the equivalence classes of regular languages induced by $\lift{AH}_{d}$ for normalized edit distances $d$, and characterize their asymptotic essence.
Focusing in particular on the normalized edit distance of Marzal and Vidal, $\ned$, we investigate the computation of $\lift{AH}_{\ned}$ for regular languages and for bounded context-free languages.
\end{abstract}

\section{Introduction}
Various applications in formal methods call for a notion of similarity between languages. Such a need arises, for example, in applications of program repair, robustness quantification, and grammatical inference.
In all these settings, similarity between two languages $X$ and $Y$ is grounded in a notion of edit operations required to transform words $x\in X$ into words $y\in Y$. More generally, edit distance itself has been extensively studied in areas such as error-correcting codes, parsing theory, speech recognition, and molecular biology, highlighting its broad relevance.

In the context of robustness, Filliot et al.~\cite{FiliotMRST20} study the computation of 
$\inf_{x\in X}\inf_{y\in Y}\elementmetric{d}(x,y)$, where $\elementmetric{d}(x,y)$ is a cost function for editing $x$ into $y$, implemented by a given weighted transducer.
Similarity functions of this form are also studied by Mohri~\cite{Mohri02}, Samanta et al.~\cite{Samanta13}, and Henzinger et al.~\cite{HenzingerOS14}. The latter also considers the dual quantity 
$\sup_{x\in X}\sup_{y\in Y}\elementmetric{d}(x,y)$.

In repair applications, Benedikt et al.~\cite{BenediktPR11-lics,BenediktPR14} consider the standard edit operations of insertion, substitution, and deletion (with uniform costs) required to transform a string in $X$ into a string in $Y$.
In~\cite{BenediktPR11-lics}, they ask whether one can transform any word in $X$ into a word in $Y$ using a bounded number of edits.
For example, if $X=a^*c^*$ and $Y=a^*bc^*$, 
then at most one edit operation (inserting a $b$) is required to transform any string in $X$ into a string in $Y$.
In contrast, if $X=a^*$ and $Y=(ab)^*$, then there is no bound on the number of edits required to transform a word in $X$ into a word in $Y$.
They note that requiring a uniform bound on the number of edits is a strong restriction, and in subsequent work~\cite{BenediktPR14} they therefore study the \emph{percentage} of letters that need to be edited.
In the latter example, the expected value is $\frac{1}{2}$, since words in $X$ are of the form $a^n$ and such a word requires $\frac{n}{2}$ edit operations to transform into a closest word in $Y$.

In general, one expects such similarity notions to induce a metric or a pseudo-metric.
Some applications, such as repair, are inherently asymmetric; in such cases, one may only require adherence to the triangle inequality
$\metric{D}(X,Z)\leq \metric{D}(X,Y)+\metric{D}(Y,Z)$.
The triangle inequality is nevertheless central, as it ensures alignment with user intuitions, enables compositional reasoning (inferring distances between $X$ and $Z$ from distances between $X,Y$ and $Y,Z$), and is computationally beneficial in optimization and learning tasks by enabling pruning, efficient indexing, and incremental updates.

It turns out that none of the similarity notions used in the works mentioned above constitute a metric.
In this work, we ask whether it is possible to obtain a metric or a pseudo-metric between languages when the underlying similarity between words is based on the rationale of~\cite{BenediktPR14} and captures the \emph{percentage} of edit operations required to transform words in $X$ to words in $Y$ in the limit.
Since we are interested in percentages, we restrict attention to functions returning values in $[0,1]$.
That is, we seek a function {$\metric{D}:\powerset(\Sigma^*)\times\powerset(\Sigma^*)\to[0,1]$} that is a metric (or pseudo-metric) and captures the intuitions underlying~\cite{BenediktPR14}.

We next present several examples that formalize the desired intuitions and requirements  emerging from~\cite{BenediktPR14}.

\begin{requirement}[Percentage]\label{req:percentage}
As discussed above, one may expect $\metric{D}(a^*,(ab)^*)$ to be $\frac{1}{2}$.
Rather than insisting on a specific numerical value, we require a monotonicity property:
\[
\metric{D}(a^*,(a^jb)^*) < \metric{D}(a^*,(a^ib)^*) \quad \text{for all } j>i .
\]
{Intuitively, this reflects that as the percentage of symbols requiring editing increases, so does the induced distance.}
\end{requirement}

\begin{requirement}[Outlier-insensitive, finite-subset indifferent]
Consider repairing $X=b^*$ to $Y=a^*$. Since any word $x\in X$ is of the form $b^n$ and such a word requires at least $n$ edits (i.e. one edit per character), we expect $\metric{D}(b^*,a^*)=1$. 
Consider now repairing $X'=a^*\cup b$ to $Y$.
The percentage of edit operations for words of the form $a^n$ is zero, while the percentage for the word $b$ is $1$. 
Consequently, the supremum of the edit-percentage from words in $X'$
 to 
$Y$ is $1$. This is undesired as it gives the impression that $a^*\cup b$ is as farthest as possible from $a^*$, while we expect a value reflecting they are quite similar. Similarly, we expect repairing $a^*\cup b\cup bb$ to $(acc)^*$ to return $\frac{2}{3}$ and not $1$, since for all but finitely many words, the percentage is $\frac{2}{3}$. That is, we expect $\metric{D}$ to be outlier-insensitive, which we formulate as follows: 
\(
\metric{D}(X,Y)=\metric{D}(X{\cup}F,Y)
\)
for all infinite languages $X,Y$ and finite languages $F$.
Note that this requirement necessarily violates identity of indiscernibles on the full powerset, as it entails $\metric{D}(X{\cup}F,X)=\metric{D}(X,X)=0$. Thus, we seek for a pseudo-metric rather than a metric.\footnote{The \emph{identity of indiscernibles} prescribes that $d(x,y)=0$ iff $x=y$ and a \emph{pseudo-metric} relaxes this condition to require only $d(x,x)=0$.}
\end{requirement}

\begin{requirement}[Bounded-edits insensitivity]
Consider now repairing $X''=a^*ba^*$ to $Y=a^*$.
In contrast to the previous example, \emph{every} word $x\in X''$ requires at least one edit operation to reach a closest word in $Y$ (in fact, exactly one).
However, to capture the percentage nature, we note that the percentage of the number of edits required to transform {$a^nba^\ell$ to $a^{n+\ell}$ diminishes as $n$ or $\ell$} grows. 
Note that this is true also if we consider repairing $a^*b^ka^*$ to $a^*$.
While now every word requires $k$ edits, still as $n$ {or $\ell$} grows to infinity, the $k$ edits required are negligible compared to the length of the word. 
Thus, the following formal requirement emerges from~\cite{BenediktPR14}:
If there exists a bound $k$ such that every $x\in X$ can be transformed into a word in $Y$ with at most $k$ edits {and vice versa}, then
$\metric{D}(X,Y)=0$.
\end{requirement}

\begin{remark}
One may argue that it is desirable to distinguish the language $a^*$ from $a^*\cup b$ or from $a^*ba^*$.
Indeed, there exist language similarity measures that make such distinctions, in particular~\cite{BenediktPR11-icalp}.
In applications where such sensitivity is required, one may combine such a measure with the distance developed here (for example, via a product construction).
In this work, however, we deliberately impose invariance under finite subsets of words and a finite number of edits, as these capture the asymptotic notion of similarity that we and \cite{BenediktPR14} aim to capture.
\end{remark}

To summarize, we seek a language similarity function with a {bounded codomain},
specifically {$\metric{D}:\powerset(\Sigma^*)\times\powerset(\Sigma^*)\to[0,1]$}, that is a \textbf{pseudo-metric}  and satisfies the following  \textbf{requirements}:
\begin{enumerate}
\item \emph{Percentage nature}:
$\metric{D}(a^*,(a^jb)^*) < \metric{D}(a^*,(a^ib)^*)$ for all $j>i.$
\item \emph{Outlier insensitivity}: 
$\metric{D}(X,Y)=\metric{D}(X{\cup}F,Y)$ for all infinite $X,Y$ and finite $F$.

\item \emph{Bounded-edits insensitivity}:
If there exists a bound $k$ such that every $x\in X$ can be transformed into a word in $Y$ with at most $k$ edits {and vice versa}, then
$\metric{D}(X,Y)=0$.
\end{enumerate}
\vspace{1mm}

Since languages are sets of words, and we seek a similarity notion induced by edit distance between words, a natural question is whether there exists a general method to lift a metric $\elementmetric{d}$ on a universe $M$ to a metric $\metric{D}$ on subsets of $M$ that meets these requirements.
A classical lifting scheme is given by the Hausdorff distance.
Let us first recall the common way to define a distance between an element $x$ and a set $Y$. This measure, $\wordlang{d}:M\times \powerset(M)\to\nonnegR$ is defined as $\wordlang{d}(x,Y)=\inf_{y\in Y}\elementmetric{d}(x,y)$, i.e. it measures the distance of $x$ to the closest element $y$ in $Y$.\footnote{Its use in the context of formal languages goes back to~\cite{Wagner74} where $\elementmetric{d}$ is the Levenshtein Edit Distance~\cite{Levenshtein66}.} Next, the {directional (asymmetric)} distance from set $X$ to set $Y$ is defined as $\dr(X,Y)=\sup_{x\in X}\wordlang{d}(x,Y)$, namely it is the distance of the farthest element in $X$ to $Y$. Finally, 
given $\elementmetric{d}:M\times M\to\nonnegR$, the Hausdorff distance with respect to $\elementmetric{d}$ is the function
$\lift{H}_{\elementmetric{d}}:\powerset(M)\times\powerset(M)\to\nonnegR$ defined by 
\  $\lift{H}_{\elementmetric{d}}(X,Y)
=\max\left\{ \dr(X,Y), \dr(Y,X) \right\}$ \ 

i.e. it takes the maximum of the {directional distance from $X$ to $Y$ and in the other direction. 

While the Hausdorff lifting yields a metric and can be used to lift edit distances between words to distances between languages, applying it to the standard Levenshtein edit distance $\ed$~\cite{Levenshtein66} fails to produce a similarity measure with a percentage nature, and the resulting values are unbounded.
For example, $\hed(a^*,a^*\cup bb)=2$ and $\hed(a^*,(ab)^*)=\infty$.

Replacing $\ed$ by a normalized edit distance such as $\ned$~\cite{MarzalV93}, $\ged$~\cite{LiL07}, or $\ced$~\cite{HigueraM08} restores the percentage nature, but does not address a more fundamental limitation:
the Hausdorff distance is inherently sensitive to outliers.\footnote{Formal definitions of $\ed$, $\ned$, $\ged$, and $\ced$ appear in \autoref{subsec:word-metrics}. At a high level, $\ed$ denotes the minimum number of edit operations, whereas $\ned$, $\ged$, and $\ced$ correspond to different normalized variants.}
Indeed, while $\hned(a^*,(ab)^*)=\tfrac{1}{2}$, we have $\hned(a^*,a^*\cup bb)=1$.
This shortcoming is not specific to edit distance, but arises from the supremum-based definition of the Hausdorff construction itself.

The Hausdorff distance is well suited to lifting distances on finite universes or universes where elements are conceived as having the same size.
In contrast, it is size-oblivious and ill suited to infinite universes whose elements admit an unbounded notion of size, as is the case in formal languages, where languages of interest are necessarily infinite and contain words of unbounded length.
Similar limitations of Hausdorff-type constructions have been observed in prior work on language similarity measures, leading to the proposal of various alternative notions of language similarity; we review these in~\autoref{sec:related-work}.

Our main contribution is a new lifting scheme
$\lift{AH}_{\elementmetric{d}}:\powerset(M)\times\powerset(M)\to\nonnegR$, which we term \emph{Asymptotic Hausdorff}.
This construction lifts an element-level (pseudo) metric $\elementmetric{d}$ to a set-level pseudo-metric while remaining insensitive to finite outliers and respecting unbounded growth in element size.
Unlike the classical Hausdorff distance, it is specifically suited to infinite universes equipped with a natural size notion and yields metric-valued distances whose behavior is asymptotically aligned with the underlying element-level distance.

The lifting is defined for element-level metrics satisfying a property we call the \emph{asymptotic separation property}.
We show that normalized edit distances such as $\ned$, $\ged$, and $\ced$ satisfy this property, and that the resulting language distance meets all of our stated requirements.
We further show that this property is satisfied by other common distance functions such as the Euclidean distance, the $L_p$ metric, and in fact  every norm-induced metric. 
This generality suggests applications beyond the formal-languages setting, particularly in domains that employ finite representations or generators of infinitely many objects, equipped with a natural notion of size.

One such application arises in the study of graph spanner constructions.
A spanner is a subgraph of a given weighted graph and thus shares its vertex set with the input.
Spanner constructions are thus naturally viewed as set of objects indexed by graph size.
Any metric $\elementmetric{dist}$ comparing spanners over the same graph—for example, based on stretch, distortion, or sparsity—can be lifted via $\lift{AH}_{\elementmetric{dist}}$ to obtain an asymptotic comparison between spanner constructions, focusing on large-scale behavior while abstracting away finite-size effects.

Another example comes from procedural texture generators in computer graphics.
Such generators are commonly modeled as functions $f : \mathbb{R}^k \to [0,1]$, for $k \in \{2,3\}$, producing a continuous scalar field that is mapped to concrete values such as colors, materials, or block types.
Well-known instances include Perlin noise, simplex noise, Worley noise, and their fractal variants.
Similarity between generators is typically assessed by comparing the finite structures they induce over bounded spatial regions (often called \emph{patches}),
independently of spatial location.
By equipping such patches with an appropriate distance measure that captures their similarity, and using their spatial extent as the size parameter, the lifting $\lift{AH}$ naturally captures asymptotic similarity between texture generators.

After presenting the abstract framework of the Asymptotic Hausdorff lifting, we return to edit-operation-based language similarity measures, {study their properties}, and establish complexity results for computing our primary instance, $\lift{AH}_{\ned}$.
For regular languages, we prove \pspace-hardness and give an approximation algorithm in \coNExp.
For bounded context-free languages (BCFLs), we present an exact algorithm running in \Exp.

All proofs are deferred to the appendix.

\section{Language Similarity Notions in the Literature}\label{sec:related-work}

A variety of notions for measuring similarity between formal languages have been explored in the literature, arising from different motivations.
In the following, we survey these measures through the lens of the requirements identified in the introduction, and highlight limitations that motivate our Asymptotic Hausdorff metric.

\paragraph*{Complexity-based similarity measures}

Since our focus is on formal languages, where (possibly infinite) languages are finitely represented by some computational model, one natural approach is to define a distance measure between languages via a distance between their representations. To ensure that such a measure is insensitive to the particular choice of representation, one may appeal to a canonical representation, when one exists for the class of languages under consideration.  

\subparagraph*{Kolmogorov and automata-size based measures}
In this spirit,~\cite{Kudlek08} proposes using Kolmogorov complexity and defines the Kolmogorov distance between languages $X$ and $Y$ as $\metric{K}(X,Y)=|K(X)-K(Y)|$, where $K(L)$ denotes the Kolmogorov complexity of the language $L$. For regular languages, the minimal DFA can serve as a canonical representation, allowing $K$ to be replaced by the number of DFA states. Alternatively,~\cite{Kudlek08} suggests using the size of a minimal NFA, defined as the sum of its states, initial and final states, and transitions.

It is straightforward to see that these notions induce a pseudo-metric. (Indeed, any function of the form $\metric{S}(X,Y)=|S(X)-S(Y)|$, where $S(L)$ maps languages to $\nonnegR$, induces a pseudo-metric on the space of languages.)
However, $\metric{K}$ and its variants measure differences in the complexity of languages rather than differences between the languages themselves. For example, $\metric{K}(a^*,b^*)=0$, since the two languages are equally simple, even though 
every word in $a^*$ needs a complete rewrite to transform into a word in  $b^*$. In contrast, our goal is to define a pseudo-metric under which $a^*$ and $b^*$ are as far apart as possible.

\paragraph*{Set-theoretic similarity measures}

As mentioned in the introduction, since languages are sets, any metric on sets can be used to induce a metric on languages. 

\subparagraph*{Jaccard}
One of the earliest notions of set similarity, dating back to the 19th century, is the \emph{Jaccard index}, $\jaccardind$, along with its dual notion, the \emph{Jaccard distance} $\jaccard$, which serves as a dissimilarity measure. These are defined as follows:
\begin{equation}
\jaccardind(X,Y)=\tfrac{|X\cap Y|}{|X \cup Y|}
\qquad\qquad
\jaccard(X,Y)=\tfrac{|X\triangle Y|}{|X \cup Y|}
\end{equation}
where $\triangle$ denotes symmetric set difference.
The Jaccard index and distance are undefined when $|X \cup Y|$ is $0$ or $\infty$, and thus are  inapplicable for infinite languages. 

\subparagraph*{Ces\'aro-Jaccard}
For infinite languages, one approach is to consider one of the limits 
\begin{align}\label{eq:non-conv-jaccard-n}
\lim_{n\to\infty}\jaccard(X^{(n)}, Y^{(n)}) & \qquad\text{ or } \qquad
    \lim_{n\to\infty} \jaccard(X^{(\leq n)}, Y^{(\leq n)})
\end{align}
where $L^{(n)}$ 
(resp. $L^{(\leq n)}$) denotes the set of words in $L$ of length $n$ (resp. at most $n$).
However, as shown in~\cite{ParkerYY16}, these limits need not exist. For example, considering the left limit, if $X = a^*$ and $Y = (aa)^*$, then the fraction evaluates to $0$ for even $n$ and to $1$ for odd $n$.

To address this issue,~\cite{ParkerYY16} propose smoothing the sequence using the Ces\`aro average, yielding the following distance measure between languages:
\begin{equation*}\label{eq:cessaro-jaccard}
    \cj(X,Y)=\lim_{n\to\infty}\tfrac{1}{n}\sum_{i=1}^n \jaccard(X^{(\leq i)},Y^{(\leq i)})
    =\lim_{n\to\infty}\tfrac{1}{n}\sum_{i=1}^n \tfrac{|(X\triangle Y)^{(\leq i)}|}{|(X\cup Y)^{(\leq i)}|}
\end{equation*}
They show that $\cj$ is a pseudo-metric.

However, while the Ces\`aro-Jaccard distance successfully addresses the fact that formal-language theory is primarily concerned with infinite languages, and is finite-subset indifferent, it does not incorporate any notion of similarity between individual words.
As a result, both $\cj(ca^*,ba^*)=1$ and $\cj(c^*,b^*)=1$, despite the fact that in the former pair the edit cost per word is uniformly bounded (one edit per word), while in the latter it grows unboundedly with word length.
Our goal, by contrast, is a metric that reflects this asymptotic discrepancy, identifying the first pair as close and the second as far.

\subparagraph*{Discounted-sum Jaccard}
Another approach to addressing the potential non-convergence of the limits in \autoref{eq:non-conv-jaccard-n}
is to employ a discounted-sum construction. This idea was recently proposed in~\cite{BruseHL22}. In particular, they show that 
\begin{equation*}\label{eq:disc-sum-jaccard}
    \dsj^\lambda(X,Y)=(1-\lambda)\sum_{n=0}^{\infty}\lambda^n\jaccard(X^{(n)},Y^{(n)})
    =(1-\lambda)\sum_{n=0}^{\infty}\lambda^n \tfrac{|(X\triangle Y)^{(n)}|}{|(X\cup Y)^{(n)}|}
\end{equation*}
is a pseudo-metric for every $\lambda\in(0,1)$.

However, this construction places greater weight on discrepancies at smaller word lengths, which limits its ability to capture asymptotic percentage behavior. In particular, differences among short words dominate the value of the distance, even when they become negligible relative to word length.
For example, for $\lambda=\tfrac{1}{2}$ we obtain {$\dsj^\lambda(a^{> 2},a^*)= \tfrac{7}{8}$} and $\dsj^\lambda(a^{\leq 2}\cup b^{>2},a^*)=\tfrac{1}{8}$, despite the fact that the former pair differs only on finitely many short words, while the latter exhibits a persistent asymptotic discrepancy.

\subparagraph*{Shortlex vector approach}
Another approach explored in~\cite{Kudlek08} begins by ordering all words over the alphabet using the shortlex order (first by length and then lexicographically). A language $L$ is then represented by an infinite binary vector $v_L \in \{0,1\}^{\omega}$, where $v_L(i)=1$ if and only if the $i$-th word belongs to $L$. The distance between two languages $X$ and $Y$ is defined by applying a chosen distance measure between binary vectors to $v_X$ and $v_Y$.

Since $v_X(i)=v_Y(i)$ precisely when the $i$-th word either belongs to both $X$ and $Y$ or to neither, this construction effectively accounts for words in the symmetric difference (and intersection) of the two languages. The precise behavior of the resulting language distance depends on the specific choice of vector distance. Nevertheless, this approach does not incorporate any notion of similarity between words themselves, making it difficult to see how it could {determine} $ba^*$ and $ca^*$ as being closer than $b^*$ and $c^*$.

\paragraph*{Word-level lifted similarity measures}
\subparagraph*{Predicate-lifted Jaccard}
In~\cite{CeweiZTI13}, it is observed that language similarity notions can benefit from enriching set-based similarity measures, such as the Jaccard distance, with an explicit notion of distance between individual words. 
Rather than using the strict symmetric difference
the approach of~\cite{CeweiZTI13} introduces a predicate $\predicate$ that captures when two words are considered sufficiently close. The comparison between $X$ and $Y$ 
considers only pairs of words deemed sufficiently close according to $\predicate$.
For example, the predicate $\predicate(x,y)$ may be defined as $\hamming(x,y)\leq c_0$, where $\hamming$ denotes the Hamming distance between words and $c_0 \in \mathbb{N}$ is a fixed constant. Another example from~\cite{CeweiZTI13} is $\predicate(x,y)=\lcs(x,y)\leq c_0$, where $\lcs$ denotes the length of the \emph{longest common subsequence}.

Given such a predicate $\predicate$, they define 
\  
\(
    \jaccardind^\predicate(X,Y)=\tfrac{| \predicate(X,Y)|}{|X \cup Y|}
\)
\  
where 
    $\predicate(X,Y) = \{ x \in X \mid \exists y \in Y,\ \predicate(x,y) \} \ \cup \ \{ y \in Y \mid \exists x \in X,\ \predicate(x,y) \}$.
This construction essentially replaces the strict intersection in the standard Jaccard index with the set of all words that are sufficiently close according to the predicate $\predicate$.

\subparagraph*{Predicate-lifted information-rate}
The focus in~\cite{CeweiZTI13} is on extending the notion of information rate introduced by Shannon and Weaver~\cite{Shannon1949} and applied to formal languages by Chomsky and Miller~\cite{ChomskyM58}. 
The information rate of a language $L$ is defined as 
\ 
\(
\ir(L)=\lim_{n\to\infty} \tfrac{\log |L^{(n)}|}{n}.
\) \ 
This notion pertains to a single language rather than a pair of languages and is intended to capture the \emph{density} of a language. 
In~\cite{CeweiZTI13}, an extension for two languages, which additionally incorporates a predicate, is suggested:\ 
\(
\ir^\predicate(X,Y)= \tfrac{\ir(\predicate(X,Y))}{\ir(X \cup Y)}.
\) \ 

 Since the information rate is not a (pseudo-)metric (see \cref{clm:ir-fails-triangle}), it does not serve as a candidate for our purposes.

\subparagraph*{Predicate-lifted Cesàro-Jaccard}
While not suggested in the literature, we note that the use of a predicate can also be applied to the Ces\`aro--Jaccard distance. For example, one can define $\cj^\predicate(X,Y)$ analogously to $\cj(X,Y)$, by replacing $(X \triangle Y)^{(\leq n)}$ with
{$(\overline{\predicate}(X,Y))^{(\leq n)}$ where $\overline{\predicate}(X,Y)=\{(x,y)\mid (x,y)\notin \predicate(X,Y)\}$}. 
Taking the predicate $\predicate$ to be 
{$\ed(x,y) \leq 1$}, we obtain $\cj^\predicate(ca^*,ba^*)=0$, which is desirable and resolves the issue that $\cj(ca^*,ba^*)=1$. Similarly, defining $\predicate$ as 
{$\ned(x,y) \leq \frac{1}{2}$} gives $\cj^\predicate(a^*,(ba)^*)=0$ while $\cj^\predicate(a^*,(bba)^*)=1$.\footnote{The formal definition of $\ed$ and $\ned$ are deferred to \autoref{subsec:word-metrics}. Intuitively, $\ed$ counts minimal number of edit operations, and $\ned$ minimal percentage of edits. In particular, for every $k\in\bb{N}$ we have $\ed(ca^k,ba^k)=1$, $\ned(a^{2k},(ab)^k)=\frac{1}{2}$ and $\ned(a^{3k},(bba)^k)=\frac{2}{3}$. }

However, the use of a predicate necessitates choosing a fixed threshold, which prevents distances from degrading gradually and makes it impossible to satisfy the monotonicity property described in \autoref{req:percentage}.

\vspace{2mm}
We now turn to language similarity notions that lift a similarity measure between words.

\subparagraph*{Infinitum and Supremum based}
As mentioned in the introduction, the measure 
\  
$\iid(X,Y) = \inf_{x \in X} \inf_{y \in Y} \elementmetric{d}(x,y)$, \ 
where $\elementmetric{d}$ assigns a cost to string transformations, has been used in the literature~\cite{Mohri02,Samanta13,FiliotMRST20}. However, this measure is not a metric. For example, taking $\elementmetric{d}$ to be the normalized edit distance $\ned$, we have 
$\iined(a^+,b^*) = 1 > 0 + 0 = \iined(a^+,a^*) + \iined(a^*,b^*)$.  
Moreover, $\iined(a^*,b^*)=0$, illustrating the outlier sensitivity of $\iid$ (the outlier being $\varepsilon$).
By similar reasoning, the dual notion
\ $\ssd(X,Y) = \sup_{x \in X} \sup_{y \in Y} \elementmetric{d}(x,y)$ \
is also not a metric and suffers from outlier sensitivity; for instance, $\ssned(a^* \cup bb, a^*) = 1$.

\subparagraph*{Prefix-distance based}
Considering variations in word-level similarity, several works examine the prefix distance between words. The prefix distance between two strings $x$ and $y$, denoted $\elementmetric{prf}(x,y)$, is defined as the number of characters in $x$ and $y$ that do not belong to their longest common prefix. The Hausdorff lifting of $\elementmetric{prf}$, denoted $\hprf$, has been studied in various works~\cite{NgRS17}.  
We note that $\hprf$, being a Hausdorff lifting, is sensitive to outliers. In addition, it is unbounded: for example, $\hprf(b^*,a^*) = \infty$, and it is also sensitive to bounded edits, as seen from $\hprf(ba^*,a^*) = \infty$.

\subparagraph*{$\acost$, Benedikt et al.'s measure}
As mentioned in the introduction, Benedikt et al.~\cite{BenediktPR13} studied the number of edits required to repair a word in $X$ into a word in $Y$. Since this quantity is often unbounded, their subsequent work~\cite{BenediktPR14} focused on capturing the \emph{percentage} of edits required in the limit. Their work forms the basis for our approach.  
The exact formula they use is  
\ 
\(    
\acost(X,Y) = \displaystyle \lim_{n \to \infty} \sup_{\substack{x \in X \\ |x|\geq n}} \inf_{y \in Y} \tfrac{\ed(x,y)}{|x|}.
\) \\
It is straightforward to see that $\acost(X,Y) \in [0,1]$ for any pair of languages $X,Y$.  
We note that $\acost$ is asymmetric by nature. For applications requiring symmetry, one can define $
\acostsym(X,Y) = \max(\acost(X,Y), \acost(Y,X))$.
However, the critical issue is that $\acost$ (and thus $\acostsym$) violates the triangle inequality.

\begin{restatable}[]{claim}{claimacostisnotametric}\label{claim: acost is not a metric}
$\acost$ does not satisfy the triangle inequality.
\end{restatable}

In summary, while a variety of language similarity measures have been proposed, existing approaches face key limitations. Complexity-based measures capture overall representation size but ignore word-level structure; set- and vector-based measures fail to account for word similarity or are sensitive to infinite languages; and word-level or edit-distance–based liftings can be outlier-sensitive or fail to satisfy fundamental metric properties such as the triangle inequality. These observations motivate the need for a pseudo-metric that simultaneously accounts for word-level similarity, is robust to outliers, and preserves essential metric properties. In the next section, we introduce such a construction: the \emph{Asymptotic Hausdorff} metric, $\lift{AH}_{\elementmetric{d}}$, which generalizes the Hausdorff lifting to capture asymptotic behavior at the language level. We first present it at an abstract level, and later specialize to $\ahned$, the lifting of the normalized edit distance $\ned$~\cite{MarzalV93}.

\section{Asymptotic Hausdorff}\label{sec:AH}

We turn to define the central notion of the paper, the \emph{Asymptotic Hausdorff lifting}. In what follows we assume $M$ is a domain (set) and $d:M \times M \to \nonnegR$ is a metric or a pseudo-metric. That is, we assume $d$ satisfies the three pseudo-metric requirements:
\begin{enumerate}
    \item \emph{Reflexivity:} $d(x,x)=0$ for all $x\in M$. 
    \item \emph{Symmetry:} $d(x,y)=d(y,x)$
 for all $x,y\in M$.
    \item \emph{Triangle inequality:} 
    $d(x,z){\,\leq\,} d(x,y){+}d(y,z)$ for all $x,y,z{\in} M$.
 \end{enumerate}

\subsection{Defining the Asymptotic Hausdorff Lifting}

As discussed in the introduction, we are interested in domains whose elements are equipped with a natural notion of size.
This allows us to distinguish between bounded and unbounded behavior and to focus on asymptotic phenomena.
\begin{mydefinition}[Size notion]
    A \emph{size notion} for $M$ is a map $s:M\to \nonnegR$.
\end{mydefinition}

Henceforth, we assume $M$ is equipped with such a size notion $s$.
We are interested in sets with increasing size of elements. To capture this we introduce the following definition, which considers infinite sequences of elements (rather than sets).

\begin{mydefinition}[$s$-bounded sequence]
    A sequence $(x_i)^\infty_{i=1}$ is called \emph{$s$-bounded} if there exists some $N \in \bb{N}$ such that 
    \(
        \limsup_{i \to  \infty} s(x_i) \le N.
    \)
     Otherwise we say it is \emph{$s$-unbounded}.
\end{mydefinition}

We are interested in element-level distances $d$ that separate $s$-bounded sequences from $s$-unbounded sequences, in the sense that the asymptotic distance is as large as possible. 

\begin{mydefinition}[Asymptotic separation property]\label{def:asymptotic separation property}
    We say $d$ has the \emph{asymptotic separation property} if for every $s$-unbounded sequence $(x_k)^\infty_{k=1}$, and $s$-bounded sequence $ (y_k)^\infty_{k=1}$, 
    $d$ satisfies 
    \(\lim_{k \to \infty}d(x_k, y_k) = \sup d \)
    where $\sup d$ is the supremum of the image of $d$.
\end{mydefinition}

We now have all the ingredients needed to define the Asymptotic Hausdorff lifting. As in the Hausdorff lifting, we {first} define {a} directional distance, and then take the maximum of going from $X$ to $Y$ and in the other direction.

\begin{mydefinition}[Asymptotic Hausdorff lifting]\label{def:asymptotic relative distance}
    Given a pseudo-metric $d$ on $M$, the two functions $\dar$ and $\dAH$ of type
    $\mathcal{P}(M) \times \mathcal{P}(M) \to \nonnegR$ are defined as follows
    \begin{align*}\label{eq:Relative-AH}
        \dar(X,Y) \defeq \lim_{k \to \infty} \sup_{\substack{x \in X \\ s(x)\geq k}}\inf_{y \in Y} d(x,y) &  \qquad  \dAH(X,Y) = \max\left\{ \dar(X,Y), \dar(Y,X) \right\}
    \end{align*}
    We refer to $\dar(X,Y)$ as the \emph{asymptotic directional distance} from $X$ to $Y$ and to $\dAH$ as the \emph{asymptotic Hausdorff lifting} of $d$.
\end{mydefinition}

\subsection{Properties of the Asymptotic Hausdorff Lifting}
We first establish that the asymptotic directional distance is well defined.

\begin{restatable}[]{claim}{claimasymptoticrelativedistancewelldefined}\label{claim: asymptotic relative distance well defined}
    $\dar(X,Y)$ is well defined for all $X,Y \subseteq M$.
\end{restatable}

Equivalently, the asymptotic directional distance $\dar(X,Y)$ can be defined as the supremum over all sequences
$(x_n)_{n\geq 1} \subseteq X$ with $s(x_n)\to\infty$, of the limit
{$\limsup_{n\to\infty} \inf_{y\in Y} d(x_n,y)$}.
That is, as the worst-case asymptotic distance to $Y$ attained along sequences
of elements of $X$ whose size {grows unboundedly}.

\begin{restatable}[Equivalent definition to asymptotic directional distance]{claim}{claimEquivalentdefinitiontoasymptoticrelativedistance}\label{claim: Equivalent definition to asymptotic relative distance}
    Let $X$ and $Y$ be sets. Then
    \(
        \dar(X,Y) = \sup_{\substack{(x_n)^\infty_{n=1}\subseteq X \\ s(x_n) \ge n ,\ \forall n \in \bb{N}}} \limsup_{n \to \infty} \inf_{y\in Y} d(x_n,y)
    \)
\end{restatable}

We next show that the asymptotic separation property is sufficient to lift the triangle inequality to the asymptotic setting.

\begin{restatable}[$\dar$ satisfies the triangle inequality]{theorem}{thmasymptoticrelativedistancesatisfiesstriangleinequality}\label{thm:asymptotic relative distance satisfiess triangle inequality}
    Let $d$ be a pseudo-metric that has the asymptotic separation property. Then $\dar$ satisfies the triangle inequality. 
\end{restatable}


The proof makes use of the following claim.

\begin{restatable}[Point-wise triangle inequality]{claim}{claimpointwisetrinagleinequality}\label{claim:pointwise trinagle inequality}
    Let $M$ be a set, $d$ a pseudo-metric on $M$ and $Z$ a subset of $M$.
    Then for every $x,y \in M$
    \(
        \inf_{z\in Z}d(x,z) \le d(x,y) + \inf_{z \in Z}d(y,z)
    \)
\end{restatable}
   \begin{proof}
        Let $Z \subseteq M$ and $x,y \in M$. As $d$ is a pseudo-metric 
        $d(x,z) \le d(x,y) +d(y,z)$
        for every $z \in Z$. Thus
        $\inf_{z\in Z}d(x,z) \le \inf_{z \in Z} \left\{d(x,y) + d(y,z) \right\}
                                 = d(x,y) + \inf_{z \in Z}d(y,z).$
    \end{proof}

    \begin{proof}[Proof of \autoref{thm:asymptotic relative distance satisfiess triangle inequality}]
        Let $X,Y,Z$ be subsets of $M$ and let $\e > 0$. Let $N \in \bb{N}$ such that 
        \begin{equation}\label{eq:YZ-sup}
            \sup_{\substack{y \in Y\,  s(y)\ge N}}\ \inf_{z \in Z} d(y,z)
            \le \dar(Y,Z) + \tfrac{\e}{2}.
        \end{equation}

        \noindent
        If all infinite sequences $(x_k)^\infty_{k=1} \subseteq X$ are $s$-bounded, then $X$ is an $s$-bounded set and thus
           \[ \dar(X,Z) = \lim_{k \to \infty} \sup_{\substack{x \in X \\
           s(x)\geq k}}\inf_{z \in Z} d(x,z)
                     = \lim_{k \to \infty} \sup_{x \in \emptyset}\inf_{z \in Z} d(x,z)\\
                     = 0 \le \dar(X,Y) + \dar(Y,Z). 
        \]
        
        \noindent
        Otherwise, let $(x_k)^\infty_{k=1} \subseteq X$ be some arbitrary sequence that has an infinite subsequence $(x_{k_i})^\infty_{i=1}$ such that $(s(x_{k_i}))^\infty_{i=1}$ is a non decreasing sequence that tends to $\infty$. We consider two cases:
         
        \smallskip \noindent
        \textbf{Case 1:}  There exists a subsequence $(k_j)^\infty_{j=1}$ such that for every $y \in Y $ that satisfies $ s(y) \ge N$, we have that $d(x_{k_j}, y) > \inf_{y \in Y} d(x_{k_j}, y) + \tfrac{\e}{2}$.  

        Let $Y_{x_{k_j}}$ be the set of elements that satisfy $d(x_{k_j},y) \le \inf_{y \in Y} d(x_{k_j}, y) + \tfrac{\e}{2}$. 
        We get that $Y_{x_{k_j}} \subseteq Y^{\le N}$ where $Y^{\le N} = \{y \in Y: s(y)\le N\}$. Therefore as $d$ has the asymptotic separation property
        \begin{align*}
            \dar(X,Y) & = \lim_{k \to \infty} \sup_{\substack{x \in X \\ s(x)\geq k}}\inf_{y \in Y} d(x,y)
                    \ge \lim_{j \to \infty} \inf_{y \in Y} d(x_{k_j},y)\\
                     &= \lim_{j \to \infty} \inf_{y \in Y^{\le N}} d(x_{k_j},y)
                     =^{\textrm{(\autoref{def:asymptotic separation property})}} \sup d.
        \end{align*}
        \noindent
        Since  $\sup d$ bounds $\dar(A,B)$ for any $A,B$, together we get that $\dar(X,Y)=\sup d$ and
        \[
            \dar(X,Z)  \le \sup d = \dar(X,Y) \le \dar(X,Y) + \dar(Y,Z).  
        \]

        \smallskip \noindent
        \textbf{Case 2:}  There exists $k'$ such that for every $k \ge k'$ there exists $y_k \in Y$ that satisfies $s(y_k) \ge N$ and also satisfies
        \begin{equation}\label{eq:yk-choice}
            d(x_k, y_k) \le \inf_{y \in Y} d(x_k, y) + \tfrac{\e}{2}
        \end{equation}
        \noindent
        As $d$ is a pseudo-metric by Claim~\ref{claim:pointwise trinagle inequality} we know that
        \begin{equation}\label{eq:pointwise-triangle}
        \inf_{z\in Z}d(x_k,z)\le  d(x_k,y_k) + \inf_{z \in Z}d(y_k,z) 
        \end{equation}
        \noindent
        And because $s(y_k) \ge N$ and $a \leq \sup A$ for any $A \supseteq \{a\}$  we get
        \begin{equation}\label{eq:inf_z y_k to dar(Y,Z)}
            \inf_{z \in Z}d(y_k, z) \le \sup_{\substack{y \in Y \\ s(y)\geq N}} \inf_{z \in Z}d(y, z) \le^\eqref{eq:YZ-sup} \dar(Y,Z) + \tfrac{\e}{2}
        \end{equation}
        \noindent
        Putting it all together we get
        \begin{equation}\label{eq:d(x_k,z) to inf y}
            \begin{aligned}
            \inf_{z\in Z}d(x_k,z)
            &\le^\eqref{eq:pointwise-triangle} d(x_k,y_k) + \inf_{z \in Z}d(y_k,z)\\
            &\le^\eqref{eq:yk-choice} \inf_{y \in Y}d(x_k, y) + \tfrac{\e}{2} + \inf_{z \in Z}d(y_k,z)\\
            &\le^\eqref{eq:inf_z y_k to dar(Y,Z)} \inf_{y \in Y}d(x_k, y) + \dar (Y,Z) +\e.\\
            \end{aligned}
        \end{equation}
        Note that this inequality is satisfied for an arbitrary sequence in $X$ with size that tends to $\infty$. 
        Using Claim~\ref{claim: Equivalent definition to asymptotic relative distance} there exists $(p_k)^\infty_{k=1}$ a sequence that satisfies
        \begin{equation}\label{eq:p_k to dar(X,Z)}
            \e + \limsup_{k\to \infty}\inf_{z \in Z}d(p_k,z) \ge \dar(X,Z).  
        \end{equation}
        Thus, we finally get
        \begin{align*}
            \dar(X,Z) &\le^\eqref{eq:p_k to dar(X,Z)} \limsup_{k\to \infty}\inf_{z \in Z}d(p_k,z) +\e\\ 
                     &\le^\eqref{eq:d(x_k,z) to inf y} \limsup_{k\to \infty} \left\{\inf_{y \in Y}d(p_k, y) + \dar(Y,Z) + 2\e \right\}\\
                     &= \limsup_{k\to \infty} \left\{\inf_{y \in Y}d(p_k, y)\right\} + \dar(Y,Z) + 2\e \\
                     &\le \sup_{\substack{(x_k)^\infty_{k=1}\subseteq X \\ s(x_k) \ge k ,\ \forall k \in \bb{N}}} \limsup_{k\to \infty} \left\{\inf_{y \in Y}d(x_k, y)\right\} + \dar(Y,Z) + 2\e \\
                     &=^{\autoref{claim: Equivalent definition to asymptotic relative distance}} \nedar(X,Y) + \dar(Y,Z) + 2\e. 
        \end{align*}
        \noindent
        As $\e$ is arbitrary small we get that $  \dar(X,Z) \le \dar(X,Y) + \dar(Y,Z)$.   
    \end{proof}

Since $\dAH$ obviously satisfies reflexivity and symmetry, an immediate corollary of \autoref{thm:asymptotic relative distance satisfiess triangle inequality} is that $\dAH$ is a pseudo-metric.

\begin{theorem}[Asymptotic Hausdorff is a pseudo-metric]\label{theorem:Asymptotic Hausdorff is metric}
    $\dAH$ is a pseudo-metric when $d$ has the asymptotic separation property.
\end{theorem}

{
Returning to our primary motivating setting of words and language, in the next section, we show that many word similarity measures satisfy the asymptotic separation property, and therefore can be lifted to pseudo-metrics between languages.
We expect that similar results hold for metrics on trees and graphs, but we do not pursue this direction here. 
}

More broadly, this naturally raises a stronger question: how restrictive is the asymptotic separation property in general?
We therefore turn to identifying broad classes of distance functions for which this property holds.
The following theorem shows that the asymptotic separation property is in fact quite common: every metric induced by a norm satisfies it.

\begin{restatable}[Norm-induced metrics satisfy asymptotic separation]{theorem}{thmnormedinducedmetricshavetheasymptoticseparationproperty}\label{thm: normed induced metrics have the asymptotic separation property}
Let $(M,\norm{\cdot})$ be a normed space with $d(x,y) = \norm{x-y}$ and $s(x)=\norm{x}$. 
Then $d$ has the asymptotic separation property. 
\end{restatable}

\begin{corollary}
     Let $(M,\norm{\cdot})$ be a normed space and $d(x,y) = \norm{x-y}$. Then $\dAH$ is a pseudo-metric on $\mathcal{P}(M)$.
\end{corollary}

As an immediate consequence of the corollary, the Asymptotic Hausdorff lifting applies to all metrics induced by norms.
This includes, in particular, the Euclidean distance and, more generally, $L_p$ distances for any $p \ge 1$.
Hence, $\lift{AH}_d$ can be meaningfully applied in continuous settings alongside the discrete ones considered earlier.

While we arrived at this notion from the perspective of formal languages, as discussed in the introduction, we believe that it is applicable in a much broader range of settings.
In particular, it is well suited to contexts in which infinite sets arise from finite representations or generators, are equipped with a natural notion of size, and contain elements whose size is unbounded.
This situation commonly occurs, for example, when considering functions defined over infinite domains, where each input $x$ induces an element $f(x)$ in the set generated by the function.

\section{Formal Languages and Distances}\label{sec:FLdistances}

We now focus on language similarity notions obtained via the Asymptotic Hausdorff lifting.
This construction applies to a variety of word-level metrics, including both classical and normalized edit distances.
While all are suitable for lifting, only normalized metrics are compatible with the percentage-based requirements discussed in the introduction.
Among these, the normalized edit distance $\ned$ will play a central role and serve as our main point of reference.
We begin by briefly reviewing standard distances between words.

\vspace{-2mm}
\subsection{Metrics on words (Preliminaries)}\label{subsec:word-metrics}
\vspace{-1mm}
Let $\Sigma$ be a finite alphabet. For a word $x \in \Sigma^*$, we write $|x|$ for its length and $x[i]$ for its $i$-th letter. The empty word is denoted by $\varepsilon$.

One of the oldest metrics on words is the \emph{Hamming distance}. It measures the distance between $x$ and $y$ as the number of letters on which they differ plus the difference between their lengths. Formally,
    if $\ell_x=|x|$ and $\ell_y=|y|$ then $\hamming(x,y)=\big| \{ i \colon x[i]\neq y[i], i \leq \min(\ell_x,\ell_y)\}\big|+\big|\ell_x-\ell_y\big|$.
Another common simple metric between words is the \emph{prefix distance}~\cite{NgRS17}. It measures the number of letters in $x$ and $y$ that are not in the longest common prefix of $x$ and $y$. Formally, $\prf(x,y)=|x|+|y|-2\cdot\max\{|z|\colon x,y\in z\cdot\Sigma^*\}$. Beyond these position-based notions, many widely used distances between words are defined in terms of edit operations, which we review next.

\subparagraph*{Edit operations and edit paths.}
We work with the standard edit operations: insertion, deletion, substitution, and no-op.
Let $\hat{\Gamma} = (\Sigma \cup \{\varepsilon\})^2$, and write $\swap{a}{b}$ for the pair $(a,b)$.
The set of \emph{edit operations} is $\Gamma = \hat{\Gamma} \setminus \{\swap{\varepsilon}{\varepsilon}\}$, where
$\swap{a}{b}$ denotes substitution of $a$ by $b$,
$\swap{a}{a}$ a no-op,
$\swap{a}{\varepsilon}$ deletion of $a$,
and $\swap{\varepsilon}{a}$ insertion of $a$.

An \emph{edit path} from $x$ to $y$ is a finite sequence $\ep$ over $\Gamma$ such that, writing
$\ep=\swap{a_1}{b_1}\swap{a_2}{b_2}\cdots\swap{a_n}{b_n}$, we have $a_1\cdots a_n = x$ and $b_1\cdots b_n = y$.
We denote by $|\ep|=n$ the length of $\ep$.
Let $\dgamma : \Gamma \to [0,1]$ assign weights to edit operations.
The \emph{weight} of an edit path $\ep$ is
$\wgt(\ep)=\sum_{i=1}^n \dgamma(a_i,b_i)$.
Its \emph{cost} is defined as {$\cost(\ep)=\wgt(\ep)/|\ep|$}. 
We write $p:x \leadsto y$ to denote that $p$ is an edit path from $x$ to $y$.

\subparagraph*{Edit-distance notions.}
Several notions of edit distance have been proposed in the literature, differing mainly in how edit operations are aggregated and normalized.
The most basic notion is the \emph{Levenshtein (edit) distance}~\cite{Levenshtein66}, denoted $\ed$, which returns the minimum total weight of an edit path transforming $x$ into $y$.

For applications involving words of substantially different lengths, normalization becomes essential. When words have equal length, normalizing by the word length is straightforward; however, for unequal lengths, naïve normalizations—such as dividing by the maximum, minimum, or sum of the lengths—generally fail to preserve the metric properties (cf.~\cite{LiL07}).
To address this issue, several normalized variants of edit distance have been proposed, including the \emph{normalized edit distance} $\ned$~\cite{MarzalV93}, the \emph{generalized edit distance} $\ged$~\cite{LiL07}, and the \emph{contextual edit distance} $\ced$~\cite{HigueraM08}.
Despite differing in their normalization schemes, all these notions yield bounded distances that satisfy the metric axioms, and are therefore well suited for comparing words of varying lengths.

A common choice of weights is the \emph{uniform weight}, in which no-op operations have cost $0$ and all other edit operations have {weight} $1$. 
All notions but $\ced$ allow non-uniform weights.
When non-uniform weights are considered, they must satisfy additional conditions to ensure that the induced distances $\ed_d$, $\ned_d$, and $\ged_d$ are metrics~\cite{LiL07,FismanGMW22,FismanT24}.
Unless stated otherwise, we henceforth assume the uniform weight.

\begin{mydefinition}[Edit-distance notions]
Let $x,y \in \Sigma^*$.
\begin{itemize}[nosep]
    \item \emph{Levenshtein (edit) distance}~\cite{Levenshtein66}:
    $\ed_d$ minimizes the weight of an edit path:\\[1mm]
    \(
     \phantom{------} \ed_\dgamma(x,y)
    = \min \{ \wgt(\ep) \mid \ep \text{ is an edit path from } x \text{ to } y \}.
    \)
    \vspace{2mm}
    \item \emph{Normalized edit distance}~\cite{MarzalV93,FismanGMW22,FismanT24}:
    $\ned_d$ minimizes the average cost per operation, by dividing by the edit path length:\\[1mm]
    \(
    \phantom{------}  \ned_\dgamma(x,y)
    = \min \{ \cost(\ep) \mid \ep \text{ is an edit path from } x \text{ to } y \}.
    \)
    \vspace{2mm}
    \item \emph{Generalized edit distance}~\cite{LiL07}:
    $\ged$ is another way to obtain an averaged cost:\\[1mm]
    \( \phantom{------}
    \ged_d(x,y)
    = \tfrac{2 \cdot \ed_d(x,y)}{|x| + |y| + \ed_d(x,y)}.
    \)
    \vspace{2mm}    
    \item \emph{Contextual edit distance}~\cite{HigueraM08}:
    Last, $\ced$ provides an averaged cost by considering the context of the edits. Formally, for strings $s,s'$ for which $\ed(s,s')=1$, one defines
    $\ced(s,s') = 1/\max(|s|,|s'|)$.
    For a sequence $\rho=(s_0,\ldots,s_k)$ satisfying $\ed(s_i,s_{i+1})=1$ for every $i<k$,
    let $\ced(\rho)=\smash\sum_{i=1}^k \ced(s_{i-1},s_i)$. 
    Then \\[1mm]
    \( \phantom{------} \ced(x,y)
    = \min \{ \ced(\rho) \mid \rho{=}(s_0,\ldots,s_k),\ s_0{=}x,\ s_k{=}y \}.
    \)
\end{itemize}
\end{mydefinition}

 \begin{myexample}
Consider $x=aab$ and $y=abac$.
One edit path from $x$ to $y$ is
$\ep_1=\swap{a}{a}\swap{a}{b}\swap{b}{a}\swap{\e}{c}$.
Another edit path is
$\ep_2=\swap{a}{\e}\swap{a}{a}\swap{b}{b}\swap{\e}{a}\swap{\e}{c}$.
We have $\wgt(\ep_1)=\wgt(\ep_2)=3$.
Since no edit path has smaller weight, it follows that $\ed(x,y)=3$.
Applying this value in the definition of $\ged$, we obtain
$\ged(x,y)=\frac{2\cdot 3}{3+4+3}=\frac{3}{5}$.
Since $|\ep_1|=4$ and $|\ep_2|=5$, we have
$\cost(\ep_1)=\frac{3}{4}$ and $\cost(\ep_2)=\frac{3}{5}$.
As no edit path has smaller cost, we conclude that $\ned(x,y)=\frac{3}{5}$.
For $\ced$, consider the sequence of strings
$s_0=aab$, $s_1=abb$, $s_2=aba$, and $s_3=abac$.
Note that $\ed(s_{i-1},s_i)=1$ for all $1\le i\le 3$.
Therefore,
$\ced(s_0,s_1,s_2,s_3)=\frac{1}{3}+\frac{1}{3}+\frac{1}{4}=\frac{11}{12}$.
However, a different sequence yields a smaller value.
In particular,
$\ced(aab,aabc,abbc,abac)=\frac{1}{4}+\frac{1}{4}+\frac{1}{4}=\frac{3}{4}$.
Thus, $\ced(x,y)\le \frac{3}{4}$.
\end{myexample}

The values of $\ed$ are clearly unbounded.
In contrast, the values of $\ned$ and $\ged$ are bounded by~$1$ and may attain this bound.
The values of $\ced$ are unbounded; however, they can be made bounded by considering the variant $\ced\,'(x,y)=\max\{1,\ced(x,y)\}$~\cite{FismanGMW22,HigueraM08}.
Finally, we note that $\ged(x,y)\leq \ned(x,y)$ for all $x,y$; see \autoref{claim: ged le ned}.

We establish that each of these notions satisfies the asymptotic separation property, and thus is amenable to lifting via the Asymptotic Hausdorff construction.

\begin{restatable}[]{claim}{claimwordmetricswithasp}\label{claim: word metrics with asp}
The asymptotic separation property holds for $\ed$, $\ned$, $\ged$, $\ced$, and $\prf$.
\end{restatable}

\begin{corollary}
    $\lift{AH}_{\ed}$, $\lift{AH}_{\ned}$,
    $\lift{AH}_{\ged}$, $\lift{AH}_{\ced}$, and $\lift{AH}_{\prf}$, are all pseudo-metrics on the set of languages.
\end{corollary}

\subsection{$\lift{AH}$ for Normalized Edit Distances}
Recall the three requirements given in the introduction. We claim that the three notions of normalized edit distance $\ned$, $\ged$ and $\ced$, all satisfy these properties.

\begin{restatable}[$\nedAH, \gedAH, \cedAH$ are finite-subset indifferent]{claim}{clmfinitesubsetindif}\label{clm:finite-subset-indif} 
 Let $X,Y$ be infinite language and $F$ a finite languages. 
Then $\metric{D}(X,Y)=\metric{D}(X{\cup}F,Y)$ for every $\metric{D}\in\{ \nedAH, \gedAH, \cedAH\}$. 
\end{restatable}

\begin{restatable}[$\nedAH, \gedAH, \cedAH$ are bounded-edits insensitive]{claim}{clmboundededits}\label{clm:bounded-edits} Let $X,Y{\subseteq} \Sigma^*$.
If $\exists k{\in}\bb{N}$ such that for every $x{\in}X$ there exists $y{\in}Y$ such that $\ed(x,y)<k$ and vice versa then 
$\metric{D}(X,Y)=0$ for every $\metric{D}\in\{ \nedAH, \gedAH, \cedAH\}$. 
\end{restatable}

\begin{restatable}[$\nedAH, \gedAH, \cedAH$ have a percentage nature]{claim}{clmprecentagenature}\label{clm:precentage-nature} 
For every $\metric{D}\in\{ \nedAH, \gedAH, \cedAH\}$ we have that
$\metric{D}(a^*,(a^jb)^*) < \metric{D}(a^*,(a^ib)^*)$ for all $j > i.$ 
\end{restatable}

\begin{corollary}\label{cor:all-normalized-satisfy-all-reqs}
    $\nedAH, \gedAH, \cedAH$ satisfy all of our requirements. 
\end{corollary}

Recall (cf.\ \autoref{req:percentage}) that one may expect
$\metric{D}(a^*,(ab)^*)=\frac{1}{2}$ and, more generally,
$\metric{D}(a^*,(a^kb)^*)=\frac{1}{k+1}$.
We deliberately adopted a more relaxed requirement, as enforcing these equalities exactly would be overly restrictive.
Nevertheless, we show that $\nedAH$ does satisfy this stricter behavior, whereas $\gedAH$ and $\cedAH$ do not.

\begin{requirement}[Strict percentage property]
We say that a language metric $\metric{D}$ satisfies the \emph{strict percentage property} if the following hold:
\begin{itemize}[nosep]
\item $\metric{D}\left(a^*,(a^kb^l)^*\right)=\frac{l}{k+l}$ for all $k,l\geq 0$ with $k+l\geq 1$;
\item $\metric{D}\left((a^ib)^*,(a^kb)^*\right)=\frac{k-i}{k+1}$ for all $i<k$.
\end{itemize}
\end{requirement}

\begin{restatable}[Satisfaction of the strict percentage property]{claim}{claimahnedsatisfiesstrictprecentage}\label{claim:ahned satisfies strict precentage}
The strict percentage property is satisfied by $\nedAH$, but not by $\gedAH$ or $\cedAH$.
\end{restatable}

For this reason, we prefer $\nedAH$ over $\gedAH$ and $\cedAH$.

\begin{restatable}[Relations]{remark}{lemrelations}\label{lem:relations} 
    The following relations between the metrics hold for every pair of languages $X,Y$:
    \(
       \gedAH(X,Y) \le \ \nedAH(X,Y) \le \nedH(X,Y) \le \edH(X,Y)
    \) 
\end{restatable}

While $\nedAH$ satisfies all of our stated requirements, its insensitivity to finite outliers and bounded local edits implies that it is a pseudo-metric rather than a metric.
In scenarios where one wishes to additionally distinguish, for example, $a^*\cup b$ from $a^*$ (i.e., to forgo outlier insensitivity), or $a^*b$ from $a^*$ (i.e., to forgo bounded-edit insensitivity), this can be achieved by combining $\nedAH$ with an additional language metric $\metric{D}$ (e.g. $\hed$, $\hprf$, etc.).
Specifically, one may consider the refined generalized metric
\(
\lift{D}_{\times}(X,Y)=(\nedAH(X,Y),\metric{D}(X,Y)),
\)
ordered lexicographically so that the $\nedAH$ component is dominant.\footnote{Here we use the term \emph{generalized metric} to refer to a distance function whose codomain is a totally ordered set (rather than $\nonnegR$), and which satisfies the triangle inequality with respect to that order; see, e.g.,~\cite{Lawvere73,Flagg97}.}
If $\metric{D}$ takes values in a domain that is bounded above by a constant $C$ (as is the case, for example, for $\hned$), then this generalized metric can be collapsed into a genuine metric by defining
\(
\lift{D}_{\otimes}(X,Y)=(C{+}1)\cdot\nedAH(X,Y)+\metric{D}(X,Y),
\)
which preserves the dominance of the asymptotic distance.

\subsection{Asymptotic Essence of Regular Languages}\label{sec:essense}

A pseudo-metric $\metric{D}$ naturally induces an equivalence relation $\equivD$, where
$X \equivD Y$ if and only if $\metric{D}(X,Y)=0$.
In our context, two languages are considered equivalent with respect to language-metric $\metric{D}$ if their $\metric{D}$-distance is zero.

Recalling the discussion in the introduction, which was illustrated using regular expressions, it appears that removing all parts of a regular expression that are not under a Kleene star yields a language that is equivalent to the original one under the desired pseudo-metric $\metric{D}$.
Intuitively, the \emph{asymptotic essence} of a regular expression is captured by the subexpressions occurring under Kleene closure, whereas all other subexpressions are asymptotically negligible.
Accordingly, we denote by $\essenceRegex(r)$ the regular expression obtained by retaining only these Kleene-starred subexpressions.

Note that applying this procedure to different regular expressions defining the same language may result in different languages.
For example, consider $r_1 = a(aa)^* \cup (aa)^*$ and $r_2 = a^*$.
Although $\sema{r_1} = \sema{r_2}$, we have $\sema{\essenceRegex(r_1)} \neq \sema{\essenceRegex(r_2)}$, since
$\essenceRegex(r_1) = (aa)^*\cup(aa)^*$ whereas $\essenceRegex(r_2) = a^*$.

Since regular languages do not admit a canonical regular expression, we seek to define asymptotic essence at the level of automata, so that it can later be applied to the minimal DFA, which is canonical up to isomorphism.
Given a DFA {or an NFA} $A$ recognizing a language $L$, we seek an intuitively simpler language $L'$ that is equivalent to $L$ under $\equivD$.
To this end, we decompose $A$ into its strongly connected components (SCCs) and replace all transitions that are not contained in {some} SCC by $\varepsilon$-transitions. 
The resulting automaton is {a simpler} NFA $A'$ that recognizes a language $L'$ which, intuitively, preserves exactly the asymptotically significant behavior of $A$.\footnote{Here, we use the strict notion of SCC: a singleton state forms an SCC only if it has a self-loop.}

We denote by $\essenceFa(A)$ the NFA obtained by this construction.
The \emph{asymptotic essence} of a regular language $L$, denoted $\essence(L)$, is then defined as $\essenceFa(A_L)$, where $A_L$ is the minimal DFA recognizing $L$.

The following claim shows that these constructions—whether applied to regular expressions or to automata—yield languages that are equivalent under $\equivD$ when $\metric{D}$ is instantiated as one of the Asymptotic Hausdorff distances using a normalized word metric.

\begin{restatable}[Asymptotic Essence Equivalences]{claim}{claimessenceequivalences}\label{claim:essence-equivalences}
Let $L$ be a regular language, and $r$ and $A$ a regular expression and a {NFA}  recognizing $L$, resp.
Then, for every $\metric{D} \in \{\nedAH,\gedAH,\cedAH\}$
\[
L \equivD \sema{\essence(L)} \equivD \sema{\essenceRegex(r)} \equivD \sema{\essenceFa(A)}
\]
\end{restatable}
\begin{proof}[Sketch proof]
    The proof follows from
    \autoref{clm:finite-subset-indif}
    since $\essenceRegex$ and $\essenceFa$ (and thus $\essence$) yield languages that differ from the original language by finitely many words 
    and require finally many edits.
\end{proof}

We note that this claim does not hold in general for $\lift{AH}_{\elementmetric{d}}$. For example, it fails for $\lift{AH}_{\ed}$, since $\essenceRegex(a^*bb)=a^*$ while $\lift{AH}_{\ed}(a^*bb,a^*)=2$ implying $a^*bb \not\equiv_{\lift{AH}_{\ed}} a^*$.

\section{On the Computation of \nedAH }\label{sec:computation}

We now turn to the computation of $\ahned$, our primary notion of interest (cf.\ \autoref{cor:all-normalized-satisfy-all-reqs} and \autoref{claim:ahned satisfies strict precentage}).
We begin, in \autoref{sec:nedAH-hardness}, by considering regular languages, where we establish a \pspace-hardness result and present an approximation algorithm in \coNExp.
We then move to \autoref{sec:nedAH-bounded-CFL}, where we develop a detailed algorithm for bounded context-free languages.

\subsection{$\ahned$ for Regular Languages}\label{sec:nedAH-hardness}
{For regular languages, we establish both a hardness result and an approximation bound for computing $\ahned$.}

\subparagraph*{{Hardness}}

We begin by showing that, as in the case of $\acost$, which serves as our starting point, the problem is already \pspace-hard when languages are regular and $X=\Sigma^*$.

\begin{restatable}[\pspace-hardness]{theorem}{thmpspacehardness}\label{thm:pspace-hardness}
Let $\Sigma$ be an alphabet, $T$ an NFA over $\Sigma$ and $\nu$ a rational number.  
    The problem of deciding whether $\nedar\big(\Sigma^*, \sema{T}\big) \le \nu$ is \pspace-hard.
\end{restatable}

\begin{proof}[Sketch proof]
    The proof proceeds by a reduction from the universality problem for NFAs, adapting ideas from the \pspace-hardness construction for $\acost$~\cite{BenediktPR14} to account for normalization by edit-path length rather than word length.
    Given an NFA $A$ over $\Sigma$, we construct an NFA $T$ over the extended alphabet $\Gamma=\Sigma\cup\{\#\}$ that recognizes the language $(\#\cdot\sema{A})^*$.
    We then show that if $A$ is universal, i.e.\ $\sema{A}=\Sigma^*$, then
    $\nedar(\Gamma^*,\sema{T})=0$, whereas if $A$ is not universal, then
    $\nedar(\Gamma^*,\sema{T})>0$.
    
    To establish the latter case, we consider the sequence $(x_k)_{k\ge 1}$ defined by
    $x_k = (\#x)^k$ for some $x\in\Sigma^*\setminus\sema{A}$.
    We first analyze the cost of an optimal edit path for $x_k$ and then study the limit behavior arising in the computation of $\nedar$.
\end{proof}

The complete proof makes use of the following claim which we also use for the approximation result.

\begin{restatable}{claim}{claimboundonnededitpath}\label{claim:bound-on-ned-edit-path}
Let $x$ be a word and let $Y$ be a language recognized by an NFA with $n$ states.
Then there exists a word $y_x\in Y$ attaining the value $\inf_{y\in Y}\ned(x,y)$.
Moreover, there exists an optimal edit path from $x$ to $y_x$ of length at most $n(|x|{+}1)$.
\end{restatable}

\subparagraph*{{Approximation}}
Next, we establish bounds relating $\nedar$ and $\acost$.
In particular, $\nedar(X,Y)$ is sandwiched between $\acost(X,Y)$ and $\tfrac{1}{d}\cdot \acost(X,Y)$, where $d$ is the number of states in the minimal DFA for $Y$.

\begin{restatable}[]{claim}{claimnedarsmalleracost}\label{claim: nedar smaller acost}
        Let $X$ and $Y$ be regular languages. Then 
$\nedar(X,Y) \le \acost(
X,Y)$.
\end{restatable} 
\vspace{-2mm}
\begin{restatable}[]{claim}{claimnedarlargeracostdividenY}\label{claim: nedar larger acost divide nY}
    Let $X$ and $Y$ be regular languages. Then $\nedar(X,Y) \ge \tfrac{1}{d} \cdot \acost(X,Y)$
    where $d$ is the number of states in the minimal DFA of $Y$. 
\end{restatable}

Using these bounds and the \coNExp-time algorithm for $\acost$ developed in~\cite{BenediktPR14}, which relies on substantial technical machinery, we obtain a \coNExp-time approximation algorithm for $\nedar$.

\begin{restatable}[]{lemma}{lemmaahrsandwitch}\label{lemma:ahr-sandwitch}
        Let $X$ and $Y$ be regular languages. Then there is a \coNExp\ algorithm finding $c$ such that $\nedar(X,Y) \in [\tfrac{1}{d},1]c$
        where $d$ is the number of states in the minimal DFA of $Y$.
 \end{restatable}

\subsection{$\ahned$ for Bounded Context-Free Languages}\label{sec:nedAH-bounded-CFL}

We now turn to \emph{bounded context-free languages} (BCFLs), a well-studied subclass of context-free languages with rich structural properties and strong decidability results. We assume familiarity with standard definitions and basic properties of context-free languages, and provide the necessary definitions for \emph{bounded} CFLs.

\subparagraph*{Bounded CFL (Preliminaries)}
A language $L$ is said to be \emph{bounded} if there exist fixed words $w_1, w_2, \ldots, w_n \in \Sigma^*$ such that
\(
L \subseteq w_1^* w_2^* \cdots w_n^*.
\)
A language is a \emph{bounded context-free language} if it is both bounded and context-free.

Let $\Sigma=\{\sigma_1,\ldots,\sigma_k\}$. The \emph{Parikh vector} of a word $w\in\Sigma^*$ is the $k$-tuple
\(
\Parikh(w)=(|w|_{\sigma_1},\ldots,|w|_{\sigma_k})
\),
where $|w|_{\sigma}$ is the number of occurrences of $\sigma$ in $w$.
For bounded languages, it is often convenient to work with a Parikh representation relative to the bounding words.
Specifically, given $w_1,\ldots,w_n$ and a word
\(
w = w_1^{n_1} w_2^{n_2} \cdots w_n^{n_n},
\)
we define the Parikh vector of $w$ with respect to these words to be $(n_1,n_2,\ldots,n_n)$.

Parikh’s Theorem states that for every context-free language $L$, the set of Parikh vectors
$\Parikh(L)=\{\Parikh(w)\mid w\in L\}$
is a \emph{semi-linear set}~\cite{Parikh66}.
A set $S\subseteq\mathbb{N}^n$ is \emph{linear} if there exist vectors
$U=\{\vec{u}_0,\vec{u}_1,\ldots,\vec{u}_m\}\subseteq\mathbb{N}^n$
such that
\(
S=\{\vec{u}_0+t_1\vec{u}_1+\cdots+t_m\vec{u}_m \mid t_1,\ldots,t_m\in\mathbb{N}\}.
\)
The vector $\vec{u}_0$ is called the \emph{constant vector}, and the remaining vectors $\vec{u}_1,\ldots,\vec{u}_m$ are called \emph{period vectors}. The matrix in $\mathbb{N}^{n\times(m+1)}$ whose columns are the vectors
$\vec{u}_0,\vec{u}_1,\ldots,\vec{u}_m$ is called the \emph{generator matrix} of $S$.
A set is \emph{semi-linear} if it is a finite union of linear sets.

\begin{myexample}\label{ex:parikh}
As an example, consider the language
\(
L=\{(aba)^n (bcc)^{3n} d^k e^m c^{2m+1} \mid n,k,m\ge 0\}.
\)
This language is bounded, since
$L\subseteq w_1^* w_2^* w_3^* w_4^* w_5^*$
for $w_1{=}aba$, $w_2{=}bcc$, $w_3{=}d$, $w_4{=}e$, and $w_5{=}c$.
For the word $w=(aba)(bcc)^3d^2 e c^3$, the corresponding Parikh vector is $(1,3,2,1,3)$.
The Parikh image of $L$ is generated by the constant vector
$\vec{u}_0=(0,0,0,0,1)$
together with the period vectors $\vec{u}_1=(1,3,0,0,0)$,  $\vec{u}_2=(0,0,1,0,0)$, $\vec{u}_3=(0,0,0,1,2)$, that correspond to the powers $n$, $k$, and $m$, respectively. 
\end{myexample}

Ginsburg and Spanier~\cite{GinsburgSpanier66} characterized bounded context-free languages by showing that a bounded language is context-free if and only if its Parikh image (with respect to the bounding words) is a \emph{stratified semi-linear set}.
A linear set is \emph{stratified} if
(i) each of its period vectors has at most two non-zero coordinates, and
(ii) the non-zero coordinates of distinct period vectors do not interleave.
Formally, if $\vec{u}$ has non-zero coordinates at indices $i_1{<}i_2$ and $\vec{v}$ has non-zero coordinates at $j_1{<}j_2$, then it is not the case that $i_1{<}j_1{<}i_2{<}j_2$.
In the example above, the period vectors satisfy these conditions: each has at most two non-zero entries, and the index sets $\{1,2\}$, $\{3\}$, and $\{4,5\}$ are pairwise non-interleaving.

\subsubsection*{Towards solving $\ahned$ for BCFLs}
For simplicity, we focus on BCFLs whose Parikh image is a stratified linear (rather than semi-linear) set.
Henceforth, let $X \subseteq x_1^*\cdots x_n^*$ and $Y \subseteq y_1^*\cdots y_m^*$ be  BCFLs with generators {$U \in \bb{N}^{n \times (k+1)}$} and {$V \in \bb{N}^{m \times (r+1)}$}, respectively. Let $X_0$ be the BCFL with the same generating vector set $U$ as $X$ but where $\vecInd{u}{0}  = \Vec{0}$. Thus with $\sum^n_{i=1} \vecIndEnt{u}{0}{i}\cdot |w_i|$ edit operations we can move between $X$ and $X_0$ and hence by~\autoref{clm:bounded-edits} $\nedar(X_0,X) = 0$. From this point onward we thus assume $\vecInd{u}{0}=\vecInd{v}{0}  =\Vec{0}$ and hence {$U \in \bb{N}^{n \times k}$ and {$V \in \bb{N}^{m \times r}$}}.

\begin{myexample}\label{ex:generators}
    We use $X=\{(a)^n(bba)^{2n}\mid n\in\bb{N}\}$ and $Y=\{(abbb)^k(ab)^{2m}(c)^{m}\mid k,m\in\bb{N}\}$ as a running example.
    The generators of $X$ and $Y$ are $U\in\bb{N}^{2\times 1}$ and $V\in\bb{N}^{3\times 2}$ where $\vec{u}_1=(1,2)$  $\vec{v}_1=(1,0,0)$ and $\vec{v}_2=(0,2,1)$. 
\end{myexample}

We compute $\ahned$ for BCFLs using linear programs.
This requires several auxiliary notions, which we introduce next.

\subparagraph*{The edit graph} 
Let $x=a_1 a_2\ldots a_n$ and $y=b_1 b_2\ldots b_m$.
The edit graph $\edgraph{x,y}$ is the directed graph whose vertices are the grid points
$\{0,\ldots,n\}\times\{0,\ldots,m\}$,
with edges from $(i,j)$ to $(i{+}1,j)$, $(i,j{+}1)$, and $(i{+}1,j{+}1)$ whenever the target vertex exists.
An edge to $(i{+}1,j)$ corresponds to deleting $a_{i}$, and has weight $1$;
an edge to $(i,j{+}1)$ corresponds to inserting $b_{j}$, and has weight $1$;
and an edge to $(i{+}1,j{+}1)$ corresponds either to a no-op (with weight $0$) if $a_{i}=b_{j}$,
or to a substitution of $a_i$ by $b_j$ (with weight $1$) otherwise.

\begin{myexample}\label{ex:edit-graph-x-y}   
    The edit graph of $x=abbabba$ and $y=abbbababc$ is given in \autoref{fig:all}  (left). Edges that weigh $0$ are dashed whereas edges that weigh $1$ are solid. A minimum cost-path is marked on the graph in black.  It corresponds to the edit path $\swap{a}{a}\swap{b}{b}\swap{b}{b}\swap{\e}{b}\swap{a}{a}\swap{b}{b}\swap{b}{\e}\swap{a}{a}\swap{\e}{b}\swap{\e}{c}$. It has 10 edges, out of which 4 weigh $1$ so its cost is $\frac{4}{10}$. Accordingly $\ned(x,y)=\frac{4}{10}$.
\end{myexample}

\subparagraph*{Blocks}
Note that words in $X$ (resp. $Y$) are parameterized by vectors in $\bb{N}^n$ (resp. $\bb{N}^m$).
Let $\vec{n}=(n_1,\ldots,n_n)\in\bb{N}^n$ and $\vec{m}=(m_1,\ldots,m_m)\in\bb{N}^m$.
These vectors induce the words
$w^{\setX}_{\vec{n}} = x_1^{n_1}\cdots x_n^{n_n}$ and
$w^{\setY}_{\vec{m}} = y_1^{m_1}\cdots y_m^{m_m}$, resp.
We denote their lengths by
$\ell^{\setX}_{\vec{n}} = |w^{\setX}_{\vec{n}}|$ and $\ell^{\setY}_{\vec{m}} = |w^{\setY}_{\vec{m}}|$.
Consider now $x = w^{\setX}_{\vec{n}}$ and $y = w^{\setY}_{\vec{m}}$.
The edit graph of $x$ and $y$ can be partitioned into $n\times m$ rectangular subgrids, which we henceforth call \emph{blocks}.
The $(i,j)$-block corresponds to the edit graph of $x_i^{n_i}$ and $y_j^{m_j}$.

\begin{myexample}\label{ex:edit-graph-xn-ym}
Let $X=\{(a)^n(bba)^{2n}\mid n\in\bb{N}\}$ and $Y=\{(abbb)^k(ab)^{2m}(c)^{m}\mid k,m\in\bb{N}\}$.
Let $\vec{n}=(1,2)$ and $\vec{m}=(1,2,1)$
\autoref{fig:all} (middle) shows the edit graph of
$w^{\setX}_{\vec{n}}=a \cdot bba \cdot bba$ 
and $w^{\setY}_{\vec{m}}=aabb \cdot  ab \cdot ab \cdot c$ 
partitioned into its six blocks by the bold gray lines.
\end{myexample}

\begin{wrapfigure}[4]{r}{0.1\textwidth}
\vspace{-3mm}
\scalebox{0.4}{
\begin{tikzpicture}[
  font=\LARGE
]


\fill[pattern=north east lines, pattern color=gray] (0,1) rectangle (1,2);
\fill[pattern=north east lines, pattern color=gray] (0,0) rectangle (1,1);
\fill[pattern=north east lines, pattern color=gray] (2,0) rectangle (3,1);
\fill[pattern=north east lines, pattern color=gray] (1,0) rectangle (2,1);

\draw[step=1cm, black] (0,0) grid (3,2);

\node[above] at (0.5,2) {$\textcolor{\yonecol}{y_1}$};
\node[above] at (1.5,2) {$\textcolor{\ytwocol}{y_2}$};
\node[above] at (2.5,2) {$\textcolor{\ythreecol}{y_3}$};

\node[left] at (0,1.5) {$\textcolor{\xonecol}{x_1}$};
\node[left] at (0,0.5) {$\textcolor{\xtwocol}{x_2}$};

\end{tikzpicture}\label{fig:interleaving}
}
\end{wrapfigure}

\subparagraph*{Interleavings}
Let $p$ be an edit path from $x$ to $y$.
We consider the $n\times m$ grid, and say that a cell {$(i,j)$} is \emph{lit} if the edit path $p$ passes through the {$(i,j)$}-block.
We write $\pi(p)$ for the set of cells lit by $p$, and refer to $\pi(p)$ as the \emph{interleaving} imposed by $p$.
Observe that $\pi(p)$ forms a monotone path from $(1,1)$ to $(n,m)$ in the $n\times m$ grid. We use $\Pi$ to denote the set of all interleavings from $(1,1)$ to $(n,m)$.

\begin{myexample}\label{ex:grid-rectangles-path}
Continuing Ex.\ref{ex:edit-graph-xn-ym}, the $2\times 3$ grid on the right
shows that the marked edit path passes through blocks $(1,1)$, $(2,1)$, $(2,2)$, and $(2,3)$; the cells corresponding to these blocks are diagonally hatched.
\end{myexample}

\begin{remark}\label{rem:interleavin-set-size}
Note that $|\Pi| = \binom{n+m-2}{n-1}$,
since any such interleaving consists of a total of $n{-}1$ rightward steps and $m{-}1$ downward steps, and is therefore determined by the choice of which $n{-}1$ of the {$n{+}m{-}2$} steps are rightward.
\end{remark}

\begin{figure}
\vspace{-15mm}
    \scalebox{0.5}{
    \makebox[0pt][l]{\hspace{-1.5cm}
        \begin{tabular}{lll}
            \begin{tikzpicture}[
  font=\Large
]

\tikzset{
  path thin/.style={line width=3pt, dashed},
  path thick/.style={line width=3pt},
}

\def\rowletter#1{%
  \ifcase#1 a\or b\or b\or a\or b\or b\or a\fi%
}
\def\colletter#1{%
  \ifcase#1 a\or b\or b\or b\or a\or b\or a\or b\or c\fi%
}

\draw[step=1cm, gray, line width=0.15pt] (0,0) grid (9,7);

\foreach \x in {0,...,8} {
  \foreach \y in {0,...,6} {
    \pgfmathtruncatemacro{\r}{6-\y}
    \edef\rl{\rowletter{\r}}
    \edef\cl{\colletter{\x}}
    \ifx\rl\cl
      \draw[gray, dashed] (\x,\y+1) -- (\x+1,\y);
    \else
      \draw[gray] (\x,\y+1) -- (\x+1,\y);
    \fi
  }
}

\node[above] at (0.5,7) {$a$};
\node[above] at (1.5,7) {$b$};
\node[above] at (2.5,7) {$b$};
\node[above] at (3.5,7) {$b$};
\node[above] at (4.5,7) {$a$};
\node[above] at (5.5,7) {$b$};
\node[above] at (6.5,7) {$a$};
\node[above] at (7.5,7) {$b$};
\node[above] at (8.5,7) {$c$};

\node[left] at (0,6.5) {$a$};
\node[left] at (0,5.5) {$b$};
\node[left] at (0,4.5) {$b$};
\node[left] at (0,3.5) {$a$};
\node[left] at (0,2.5) {$b$};
\node[left] at (0,1.5) {$b$};
\node[left] at (0,0.5) {$a$};


\draw[path thin]
  (0,7) --
    node[above right] {} (1,6) -- 
    node[above right] {} (2,5) -- 
    node[above right] {} (3,4) -- 
    node[above]       {} (4,4) -- 
    node[above right] {} (5,3) -- 
    node[above right] {} (6,2) -- 
    node[right]       {} (6,1) -- 
    node[right]       {} (7,0) -- 
    node[below]       {} (8,0) -- 
    node[below]       {} (9,0); 

\draw[path thick] (3,4) -- (4,4); 
\draw[path thick] (6,2) -- (6,1); 
\draw[path thick] (7,0) -- (8,0); 
\draw[path thick] (8,0) -- (9,0); 

\end{tikzpicture}\label{fig:edit-graph-xy}
            &
            \begin{tikzpicture}[
  font=\Large
]

\tikzset{
  path thin/.style={line width=3pt, dashed},
  path thick/.style={line width=3pt},
}

\def\rowletter#1{%
  \ifcase#1 a\or b\or b\or a\or b\or b\or a\fi%
}
\def\colletter#1{%
  \ifcase#1 a\or b\or b\or b\or a\or b\or a\or b\or c\fi%
}

\draw[step=1cm, gray] (0,0) grid (9,7);

\node[above] at (0.5,7) {$\textcolor{\yonecol}{a}$};
\node[above] at (1.5,7) {$\textcolor{\yonecol}{b}$};
\node[above] at (2.5,7) {$\textcolor{\yonecol}{b}$};
\node[above] at (3.5,7) {$\textcolor{\yonecol}{b}$};
\node[above] at (4.5,7) {$\textcolor{\ytwocol}{a}$};
\node[above] at (5.5,7) {$\textcolor{\ytwocol}{b}$};
\node[above] at (6.5,7) {$\textcolor{\ytwocol}{a}$};
\node[above] at (7.5,7) {$\textcolor{\ytwocol}{b}$};
\node[above] at (8.5,7) {$\textcolor{\ythreecol}{c}$};

\node[left] at (0,6.5) {$\textcolor{\xonecol}{a}$};
\node[left] at (0,5.5) {$\textcolor{\xtwocol}{b}$};
\node[left] at (0,4.5) {$\textcolor{\xtwocol}{b}$};
\node[left] at (0,3.5) {$\textcolor{\xtwocol}{a}$};
\node[left] at (0,2.5) {$\textcolor{\xtwocol}{b}$};
\node[left] at (0,1.5) {$\textcolor{\xtwocol}{b}$};
\node[left] at (0,0.5) {$\textcolor{\xtwocol}{a}$};


\draw[line width=3pt,  gray] (0,7) -- (9,7);
\draw[line width=3pt,  gray] (0,6) -- (9,6);
\draw[line width=3pt, gray] (0,0) -- (9,0);

\draw[line width=3pt, gray] (0,0) -- (0,7);
\draw[line width=3pt, gray] (4,0) -- (4,7);
\draw[line width=3pt, gray] (8,0) -- (8,7);
\draw[line width=3pt, gray] (9,0) -- (9,7);


\draw[decorate,decoration={brace,amplitude=6pt},\yonecol]
  (0,7.4) -- (4,7.4) node[midway,above=6pt] {$\textcolor{\yonecol}{y_1}$};
\draw[decorate,decoration={brace,amplitude=6pt},\ytwocol]
  (4,7.4) -- (6,7.4) node[midway,above=6pt] {$\textcolor{\ytwocol}{y_2}$};
\draw[decorate,decoration={brace,amplitude=6pt},\ytwocol]
  (6,7.4) -- (8,7.4) node[midway,above=6pt] {$\textcolor{\ytwocol}{y_2}$};
\draw[decorate,decoration={brace,amplitude=6pt},\ythreecol]
  (8,7.4) -- (9,7.4) node[midway,above=6pt] {$\textcolor{\ythreecol}{y_3}$};

\draw[decorate,decoration={brace,mirror,amplitude=6pt},\xonecol]
  (-0.7,7) -- (-0.7,6) node[midway,left=6pt] {$\textcolor{\xonecol}{x_1}$};
\draw[decorate,decoration={brace,mirror,amplitude=6pt},\xtwocol]
  (-0.7,6) -- (-0.7,3) node[midway,left=6pt] {$\textcolor{\xtwocol}{x_2}$};
\draw[decorate,decoration={brace,mirror,amplitude=6pt},\xtwocol]
  (-0.7,3) -- (-0.7,0) node[midway,left=6pt] {$\textcolor{\xtwocol}{x_2}$};


\draw[path thin]
  (0,7) --
    node[above right] {} (1,6) -- 
    node[above right] {} (2,5) -- 
    node[above right] {} (3,4) -- 
    node[above]       {} (4,4) -- 
    node[above right] {} (5,3) -- 
    node[above right] {} (6,2) -- 
    node[right]       {} (6,1) -- 
    node[right]       {} (7,0) -- 
    node[below]       {} (8,0) -- 
    node[below]       {} (9,0); 

\draw[path thick] (3,4) -- (4,4); 
\draw[path thick] (6,2) -- (6,1); 
\draw[path thick] (7,0) -- (8,0); 
\draw[path thick] (8,0) -- (9,0); 

\end{tikzpicture}\label{fig:edit-graph-xy-blocks}
            &
            \begin{tikzpicture}[scale=1.5]

  \def\W{4}
  \def\H{6}

  \def\w{2}
  \def\h{3}

  \def\cA{blue}
  \def\cB{magenta}
  \def\cC{cyan}
  \def\cD{orange}

  \draw[step=1cm, gray] (0,0) grid (\W,\H);

  \newcommand{\ColorRegionGrid}[3]{
    \begin{scope}
      \clip #1 rectangle #2;
      \draw[step=1cm, #3] (0,0) grid (\W,\H);
    \end{scope}
  }

  \ColorRegionGrid{(0,0)}{(\w,\h)}{\cA}
  \ColorRegionGrid{(\w,0)}{(2*\w,\h)}{\cB}
  \ColorRegionGrid{(0,\h)}{(\w,2*\h)}{\cC}
  \ColorRegionGrid{(\w,\h)}{(2*\w,2*\h)}{\cD}

  \newcommand{\LabelBlock}[3]{
    \begin{scope}[shift={(#1,#2)}]
      \foreach \X in {0,1,2} {
        \foreach \Y in {0,1,2,3} {
          \pgfmathtruncatemacro{\Ytop}{3-\Y}
          \node[
            anchor=north west,
            font=\small,
            text=#3
          ]
          at (\X+0.01,\Y+0.01)
          {\hspace{-1px}$({\X},{\Ytop})$};
        }
      }
    \end{scope}
  }

  \LabelBlock{-0.6}{0}{\cA}
  \LabelBlock{\w}{0}{\cB}
  \LabelBlock{-0.6}{\h+0.3}{\cC}
  \LabelBlock{\w}{\h+0.3}{\cD}




  \def\rowletter#1{%
    \ifcase#1 b\or b\or a\or b\or b\or a\fi%
  }
  \def\colletter#1{%
    \ifcase#1 a\or b\or a\or b\fi%
  }

  \def\cellcolor#1#2{%
    \ifnum#1<2
      \ifnum#2<3 \cA \else \cC \fi
    \else
      \ifnum#2<3 \cB \else \cD \fi
    \fi
  }



  \def\rowletter#1{%
    \ifcase#1 b\or b\or a\or b\or b\or a\fi%
  }
  \def\colletter#1{%
    \ifcase#1 a\or b\or a\or b\fi%
  }

  \def\cellcolor#1#2{%
    \ifnum#1<2
      \ifnum#2<3 \cA \else \cC \fi
    \else
      \ifnum#2<3 \cB \else \cD \fi
    \fi
  }

  \foreach \i in {0,...,3} {
    \foreach \j in {0,...,5} {

      \pgfmathtruncatemacro{\r}{5-\j}

      \edef\rl{\rowletter{\r}}
      \edef\cl{\colletter{\i}}
      \edef\cc{\cellcolor{\i}{\j}}

      \ifx\rl\cl
        \draw[dashed, \cc] (\i,\j+1) -- (\i+1,\j);
      \else
        \draw[\cc] (\i,\j+1) -- (\i+1,\j);
      \fi
    }
  }

\draw[decorate,decoration={brace,mirror,amplitude=6pt}]
  (-0.7,6) -- (-0.7,3) node[midway,left=6pt, font=\Large] {${x_2}$};
\draw[decorate,decoration={brace,mirror,amplitude=6pt}]
  (-0.7,3) -- (-0.7,0) node[midway,left=6pt, font=\Large] {${x_2}$};

\draw[decorate,decoration={brace,amplitude=6pt}]
  (0,6.4) -- (2,6.4) node[midway,above=6pt, font=\Large] {${y_2}$};
\draw[decorate,decoration={brace,amplitude=6pt}]
  (2,6.4) -- (4,6.4) node[midway,above=6pt, font=\Large] {${y_2}$};  

\tikzset{
  path thin/.style={line width=3pt, dashed},
  path thick/.style={line width=3pt},
}

\draw[path thin]
  (0,4) --
    node[right, font=\Large] {$\mathbf{e_1}$} (1,3) -- 
    node[right, font=\Large] {$\mathbf{e_2}$} (2,2) -- 
    node[right, font=\Large]       {$\mathbf{e_3}$} (2,1) -- 
    node[right, font=\Large]       {$\mathbf{e_1}$} (3,0) -- 
    node[above, font=\Large]       {$\mathbf{e_4}$} (4,0);  

\draw[path thick] (2,2) -- (2,1); 
\draw[path thick] (3,0) -- (4,0); 

\end{tikzpicture}
        \end{tabular}
    }}
    \caption{\textbf{Left:} $\edgraph{x,y}$, \textbf{Middle:} the blocks of $\edgraph{x,y}$, \textbf{Right:} the sub-blocks of the $(2,2)$-block.}\label{fig:all}
\end{figure}

\subparagraph*{Sub-blocks and the cyclic-edit graph}
In our linear program, we want to reason simultaneously about many words $x\in X$ and $y\in Y$.
To this end, we observe that each $(i,j)$-block is composed of $n_i\times m_j$ identical copies of the edit graph of $x_i$ and $y_j$, which we call \emph{sub-blocks}.
\autoref{fig:all} (right) illustrates the partition of the $(2,2)$-block into four identical sub-blocks corresponding to 
the edit graph of $x_2$ to $y_2$.

To reason uniformly about powers of $x_i$ and $y_j$, we introduce the \emph{cyclic edit graph} of $x_i$ and $y_j$, denoted $\cyclicedgraph{i,j}$.
Let $x_i = a_1\ldots a_{|x_i|}$ and $y_j = b_1\ldots b_{|y_j|}$.
The vertices of $\cyclicedgraph{i,j}$ form a quotient of the set
$V=\{(l,k)\mid 0\leq l\leq |x_i|,\ 0\leq k\leq |y_j|\}$
where $(0,k)$ is identified with $(|x_i|,k)$, and $(l,0)$ is identified with $(l,|y_j|)$.

These identifications reflect the adjacency of sub-blocks:
for instance, the vertex $(0,k)$ of a given sub-block (e.g., the orange block in \autoref{fig:all} (right)) coincides with $(|x_i|,k)$ of the sub-block to its left (the cyan one), and analogously for the vertical direction.

As a result, there is a bijection between edit paths from {$\wSuff{{x_i}}{x_i}^*\wPref{{x_i}}$} to {$\wSuff{{y_j}}{y_j}^*\wPref{{y_j}}$} and walks in $\cyclicedgraph{i,j}$, where {$\wSuff{{x_i}}$} and {$\wSuff{{y_j}}$} are suffixes and {$\wPref{{x_i}}$} and {$\wPref{{y_j}}$} are prefixes of {$x_i$ and $y_j$}, respectively.
This bijection preserves both weight and length.
Moreover, every walk starting at $(0,0)$ corresponds to an edit path from some word in {$x_i^*\wPref{{x_i}}$} to some word in {$y_j^*\wPref{{y_j}}$}, and similarly for walks ending at $(0,0)$ and words in {$\wSuff{{x_i}}x_i^*$} and {$\wSuff{{y_j}}y_j^*$}.

\begin{myexample}\label{ex:walk-in-cyclic-edgraph}
Consider the edit path in the $(2,2)$-block of \autoref{fig:all} (right).
It induces the edit path $\swap{a}{a}\swap{b}{b}\swap{b}{\e}\swap{a}{a}\swap{\e}{b}$
which corresponds to the walk
\(
(0,2)\to(1,3){\equiv}(1,0)\to(2,1)\to(2,2){\equiv}(0,2)\to(1,3)\to(2,3)
\)
in $\cyclicedgraph{2,2}$.
This walk represents the alignment of the words $a\cdot bba$ and $ab\cdot ab$, that is,
$a x_2$ and $y_2 y_2$, respectively, where $a$ is indeed a suffix of $x_2$.
\end{myexample}

\subparagraph*{Directions}
Another means to simultaneously reason on many $x\in X$ and $y\in Y$ is the notion of \emph{direction}.
Consider again $\vec{n}=(n_1,\ldots,n_n)\in\bb{N}^n$, 
$w^{\setX}_{\vec{n}} = x_1^{n_1} \cdots x_n^{n_n}$, and $\ell^{\setX}_{\vec{n}} = |w^{\setX}_{\vec{n}}|$. 
Note that $\ell^{\setX}_{\vec{n}} = \smash\sum_{i=1}^n n_i \cdot |x_i|$.
The \emph{direction} induced by $\vec{n}$ is the vector
$\dir(\vec{n}) = (d_1,\ldots,d_n)$,  where
$d_i = {n_i}/{\ell^{\setX}_{\vec{n}}}$. 

\begin{myexample}\label{ex:x-direction}
For instance, for $X=\{(a)^n(bba)^{2n}\mid n\in\bb{N}\}$, $\vec{n}=(1,2)$, we get that $w^{\setX}_{\vec{n}}=abbabba$, $\ell^{\setX}_{\vec{n}}=7$ and $\dir(\vec{n})=(\frac{1}{7},\frac{2}{7})$.
\end{myexample}

For every $\vec{n}\in\bb{N}^n$, the induced direction $\dir(\vec{n})$ lies in the set
        \(
            \Delta_{\setX} \defeq \{\vec{d} \in \bb{R}^n_{\ge0} : \sum^n_{i=1}\vec{d}[i]\cdot|x_i| = 1\}.
        \)
Note that each direction $\vec{a}\in\Delta_{\setX}$ is induced by infinitely many vectors in $\bb{N}^n$: if $\vec{n}$ induces $\vec{a}$, then so does $c\vec{n}$ for any $c\in\bb{N}$.

Not every vector $\vec{n}\in\bb{N}^n$ induces a word $w^{\setX}_{\vec{n}}$ that belongs to $X$.
We therefore restrict attention to directions that are realizable by words in $X$, at least asymptotically.
Let $S_{\setX}\subseteq\bb{N}^n$ denote the Parikh image of $X$ with respect to $x_1,\ldots,x_n$.
We define the \emph{direction set} of $X$ as
\(
D_{\setX} \defeq\, \overline{\{\dir(\vec{n}) \mid \vec{n}\in S_{\setX}\}} \;\subseteq\; \Delta_{\setX},
\)
where the closure is taken in the standard topology on $\bb{R}^n$.
Since $S_{\setX}$ is a linear set, $D_{\setX}$ is a polytope.

\paragraph*{The Linear Program}
Fix a direction {$\vec{a}\in D_{\setX}$ and an interleaving $\pi\in\Pi$.
Our goal is to describe an optimal edit path from a word in $X$ inducing the direction {$\vec{a}$} to a word in $Y$, under the restriction that the path respects the interleaving $\pi$. We use a linear program to solve the infimum for any $\vec{b} \in D_{\setY}$ and $\lambda \in (0,1]$, where $\vec{b}$ is the "closest" direction to $\ve{a}$ in $S_\setY$ and if $\ep$ is a corresponding edit path then $\lambda$ represents the fraction $\tfrac{|x|}{|\ep|}$.
To this aim we define a linear program over the following set of variables:
\[\begin{array}{@{\bullet\quad}l@{\qquad}l}
    f^{i,j}_e\ge 0 
    & \text{for each }(i,j)\in\pi\text{ and }e\in E(\cyclicedgraph{i,j})
   \\
   \tau_{t}\ge 0
   & 
   \text{for each }1\le t \le \ell
   \\
   \lambda \in (0,1]
\end{array}\]
Intuitively the variable $f_e^{i,j}$ holds the fraction corresponding to the number of times edge $e$ is used in $\cyclicedgraph{i,j}$ along $\pi$ divided by the length of the edit path $p$. Thus, the sum overall these variables should be $1$, and this is Constraint~\ref{cons:mass} below. Working with edit paths whose length is normalized to $1$ simplifies the reasoning for $\ned$: under this condition the cost of the path is the same as its weight. 

Not every assignment to the variables $f^{i,j}_e$ corresponds to a valid edit path.
In a path, for every vertex (except for the source and target) every visited vertex has a corresponding entry and exit edges. This gives rise to Constraint~\ref{cons:flw-cons} below.\footnote{The discrepancy for the source and target is taken care of in the proof details.}

We further require that the resulting edit path corresponds to a word in $X$ with direction $\vec{a}$, and transforms it into an optimal word in $Y$. 
We use the indicators $\indicator_{x_i}(e)$ and $\indicator_{y_j}(e)$ where $\indicator_{x_i}(e)$ equals $1$ if the edge $e$ consumes a letter of $x_i$—that is, if it corresponds to a substitution, deletion, or no-op—and equals $0$ otherwise (and analogously for $\indicator_{y_j}(e)$).
Constraint~\ref{cons:src-cons} ensures that the total use of edit operations reading symbols from the source block $x_i$ matches the contribution prescribed by the direction $\vec{a}\in D_{\setX}$.
Recall that parameter $\lambda$ corresponds to the ratio $\frac{|x|}{|\ep|}$.
It is introduced to account for the fact that the expression
\quad
\(
    \sum_{(i,j)\in\Pi,\ e\in E(\cyclicedgraph{i,j})} f^{i,j}_e \,\indicator_{x_i}(e)
\)
\quad
naturally yields $n_i |x_i|$ normalized by $|\ep|$, whereas the direction vector $\vec{a}\in D_{\setX}$ is defined in terms of normalization by $|x|$.
Multiplication by $\lambda$ therefore corrects the denominator, aligning the normalization with that of the source direction. See \autoref{ex:constraint-3}.

For the target word, we do not fix a direction in advance.
Instead, we allow the linear program to choose an optimal target direction.
This is achieved using the variables $\tau_1,\ldots,\tau_r$, which induce a normalized direction vector $\vec{b}$ for $Y$ via
\quad \label{eq:bj}
\(
    b_j \;=\; \sum_{s=1}^{r} \tau^{_s} \cdot \vecIndEnt{v}{s}{j}
    \quad\text{for each } 1 \leq j \leq m.
\)
Intuitively, the variables $\tau_1,\ldots,\tau_{r}$ 
normalize by $|\ep|$ a natural combination of the generators $\vec{v}_1,\ldots,\vec{v}_r$ of $Y$.
Accordingly, $b_j$ represents the $|\ep|$-normalized contribution of the target block $y_j$.
Constraint~\ref{cons:tgt-cons} ensures that the edit path consumes symbols from the target blocks in accordance with this induced direction. See \autoref{ex:constraint-4}.

\subparagraph*{Constraints}
The constraints of the linear program $\LP_{\vec{a},\pi}$ are as follows:
\begin{enumerate}
        \item\label{cons:mass} \emph{Normalization constraint}: \\[1mm]
        \(\displaystyle \phantom{--------------}
        \sum_{(i,j)\in\pi}\sum_{e\in E(\cyclicedgraph{i,j})} f^{i,j}_e = 1
        \).
        \vspace{2mm}
        \item \label{cons:flw-cons} \emph{Flow conservation in each block}:
        for every $(i,j)\in\pi$ and every vertex $v\in V(\cyclicedgraph{i,j})$ \\[1mm]
        \(\displaystyle \phantom{--------------}
            \sum_{e\in\mathrm{Out}(v)} f^{i,j}_e \;=\; \sum_{e\in\mathrm{In}(v)} f^{i,j}_e.
        \)
        \vspace{2mm}
        \item \label{cons:src-cons} \emph{Source-consumption constraints}: for all {$i{\in \{1,\dots,n\}}$} \\[1mm]
        \(\displaystyle \phantom{--------------}
             \sum_{(i,j)\in\pi}\sum_{e\in E(\cyclicedgraph{i,j})}  f^{i,j}_e\indicator_{x_i} (e)\;=\; |x_i|\cdot a_i \cdot \lambda.
        \)
        \vspace{2mm}
        \item \label{cons:tgt-cons}\emph{Target-consumption constraints}: for all {$j{\in \{1,\dots,m\}}$} \\[1mm]
        \(\displaystyle\phantom{--------------}
            \sum_{(i,j)\in\pi}\sum_{e\in E(\cyclicedgraph{i,j})} f^{i,j}_e\,\indicator_{y_j}(e)\;=\;|y_j| \cdot b_j 
        \)
\end{enumerate}

\subparagraph*{Objective:}
        \( \phantom{--}\displaystyle
            \min\ \sum_{(i,j)\in\pi}\sum_{e\in E(\cyclicedgraph{i,j})}
            f^{i,j}_e\,c(e).
        \)\\[1mm]
where $c(e)$ is the weight of $e$. The objective thus looks for an edit path with minimum cost. \\[1mm]

For a fixed direction $\vec{a}\in D_{\setX}$, each interleaving $\pi\in\Pi$ induces a linear program $\LP_{\vec{a},\pi}$.
We denote by $\OPT_{\vec{a},\pi}$ the optimal value of $\LP_{\vec{a},\pi}$.
We aggregate these by taking the best value over all interleavings, and define
\(
F(\vec{a}) \defeq \min_{\pi\in\Pi}\OPT_{\vec{a},\pi}.
\)
Continuity of $F$ is used in the proof of the following theorem.

\begin{restatable}[]{claim}{claimLPcontinuous}\label{claim:LP-continuous}
The function $F : D_{\setX}\to\bb{R}$ is continuous.
\end{restatable}

We can finally state the main theorem for this section.

\begin{restatable}[]{theorem}{thmLP}\label{thm:LP}
     Let $X$ and $Y$ be bounded CFLs, and let $\Pi$ and $D_{\setX}$ be the corresponding interleaving and directions set.
     Then
    \(\displaystyle
    \nedar(X,Y) = \sup_{\vec{a} \in D_{\setX}} \min_{\pi \in   \Pi} \ \OPT_{\vec{a}, \pi} = \sup_{\vec{a} \in D_{\setX}}F(\vec{a}).
    \) 
\end{restatable}

Putting the above ingredients together, we obtain the following complexity bound for computing $\lift{AH}_{\ned}$ on bounded context-free languages.

\begin{restatable}[]{lemma}{lemLPexp}\label{lem:LP-exp}
    There exists an \Exp algorithm that computes $\nedar(X,Y)$ for bounded context-free languages $X$ and $Y$.
\end{restatable}

\begin{proof}[Proof sketch]

By the LP characterization, $\nedar(X,Y)=\sup_{\vec a\in D_X}F(\vec a)$, where $F(\vec a)=\min_{\pi\in\Pi}\OPT_{\vec a,\pi}$ and $\OPT_{\vec a,\pi}$ is the optimum of a linear program whose constraints are independent of $\vec{a}$ and whose right-hand side depends affinely on it. By LP duality, for each $\pi$, $\OPT_{\vec {a},\pi}$ is a convex, piecewise-linear function of $\vec{a}$. Hence, $F$ is also piecewise linear.

Since $D_X$ is a compact polytope, the supremum of $F$ is attained at a vertex of the arrangement induced by these linear pieces. The number of defining hyperplanes is exponential, as it depends on the number of interleavings and dual vertices. Standard bounds on hyperplane arrangements imply that the number of candidate vertices is exponential.

The algorithm enumerates these candidates, evaluates $F$ at each by solving the corresponding  linear programs, and returns the maximum. This yields an exponential algorithm.
    
\end{proof}

\section{Discussion}
The Asymptotic Hausdorff lifting captures a notion of similarity that is inherently asymptotic, comparing infinite sets by their long-run behavior while deliberately abstracting away from finite deviations.
This choice is intentional: in formal verification, two languages are considered similar when similarity is witnessed on increasingly long words.
In particular, in applications such as repair, robustness, and automata learning, asymptotic behavior with respect to a normalized edit distance between words reflects the semantic core of the compared languages.

Our results highlight an inherent trade-off between robustness and sensitivity.
Classical Hausdorff-style constructions are highly sensitive to outliers, whereas asymptotic notions necessarily collapse certain local distinctions.
The Asymptotic Hausdorff lifting makes this trade-off explicit and allows additional sensitivity to be reintroduced in a controlled manner by combining it with complementary distances.

Although formal languages serve as a central motivating domain, the lifting scheme itself is more general and applies to any infinite domain equipped with a meaningful notion of size and an element-level metric.

From an algorithmic perspective, our results indicate that computing $\ahned$ is already challenging for regular and bounded context-free languages, and extending the proposed techniques to richer classes or more efficient exact algorithms is an interesting direction for future work.

\bibliographystyle{plainurl}
\bibliography{gm}

\appendix
\section{Additional Examples}

\begin{example}\label{ex:constraint-3}
Continuing Ex.\ref{ex:grid-rectangles-path}, we obtain $\lambda=\frac{7}{10}$, since $|x|=7$ and $|\ep|=10$.
Recall from Ex.\ref{ex:x-direction} that $\dir(\vec{n})=(\frac{1}{7},\frac{2}{7})$.
We now verify Constraint~\ref{cons:src-cons}  for each $i\in\{1,2\}$.
\begin{itemize}
    \item For $i=1$:
    \(
    \sum_{\substack{(1,j)\in\pi \\ e\in E(\cyclicedgraph{1,j})}} f^{1,j}_e \indicator_{x_1}(e)
    = \tfrac{1}{10}.
    \) \\[1mm]
    Indeed, among the rectangles $(1,j)$, only rectangle $(1,1)$ is traversed by the path, and within it there is a single edge that consumes a letter of $x_1$. For the right hand side we have,
    \(
    |x_1| \cdot a_1 \cdot \lambda
    = 1 \cdot \tfrac{1}{7} \cdot \tfrac{7}{10}
    = \tfrac{1}{10}.
    \)
    so equality holds.
    \vspace{2mm}
    \item For $i=2$:
    \(
    \sum_{\substack{(2,j)\in\pi \\ e\in E(\cyclicedgraph{2,j})}} f^{2,j}_e \indicator_{x_2}(e)
    = \tfrac{6}{10}.
    \) \\[1mm]
    This holds since 
    in the $(2,1)$-block there are 2 edges that consume a letter of $x_2$, in the $(2,2)$-block there 4 such edges, and in the $(2,3)$-block there are none.
    For the right hand side we have,
    \(
    |x_2| \cdot a_2 \cdot \lambda
    = 3 \cdot \tfrac{2}{7} \cdot \tfrac{7}{10}
    = \tfrac{6}{10}.
    \)
    satisfying the requested equality.
\end{itemize}
\end{example}

\begin{example}\label{ex:constraint-4}
Continuing Ex.\ref{ex:grid-rectangles-path}, recall that the generators of $Y$
are $\vec{v}_1=(1,0,0)$ and $\vec{v}_2=(0,2,1)$. The vector $\vec{m}=(1,2,1)$ is obtained by the natural linear combination $1\cdot\vec{v}_1+1\cdot\vec{v}_2$. Since $|\ep|=10$ these yield $\tau_1=\tfrac{1}{10}$, 
    $\tau_2=\tfrac{1}{10}$. 
Accordingly we get,
    $$\begin{array}{l@{=}l}
    b_1 & \tau_1\cdot \vecIndEnt{v}{1}{1} +
                \tau_2 \cdot \vecIndEnt{v}{2}{1}=
                \tau_1\cdot 1 = \frac{1}{10} \\ 
    b_2& \tau_1\cdot \vecIndEnt{v}{1}{2} + \tau_2\cdot \vecIndEnt{v}{2}{2} =
                \tau_2\cdot 2 = \frac{2}{10} 
                 \\ 
    b_3 & \tau_3\cdot \vecIndEnt{v}{1}{3}  + \tau_2 \cdot \vecIndEnt{v}{2}{3} =
    \tau_2\cdot 1 = \frac{1}{10}\end{array}.$$
We now verify Constraint~\ref{cons:tgt-cons}  for each $j\in\{1,2,3\}$.
\begin{itemize}
    \item For $j=1$ : 
    $\sum_{\substack{(i,1)\in\pi \\ e\in E(\cyclicedgraph{i,1})}} 
    f^{i,1}_e \indicator_{y_1}(e)
    = \tfrac{4}{10} = 
    |y_1|\cdot b_1=4 \cdot \tfrac{1}{10}$
        \vspace{2mm}
    \item For $j=2$ : 
    $\sum_{\substack{(i,2)\in\pi \\ e\in E(\cyclicedgraph{i,2})}} 
    f^{i,2}_e \indicator_{y_2}(e)
    = \tfrac{4}{10}= |y_2|\cdot b_2
    =2 \cdot \tfrac{2}{10}$
        \vspace{2mm}
    \item For $j=3$ : 
    $\sum_{\substack{(i,3)\in\pi \\ e\in E(\cyclicedgraph{i,3})}} 
    f^{i,3}_e \indicator_{y_3}(e)
    = \tfrac{1}{10}= 
    |y_3|\cdot b_3=1 \cdot \tfrac{1}{10}$    
\end{itemize}
\end{example}

\section{Proofs of \autoref{sec:related-work}}
\claimacostisnotametric*
\begin{proof}
        Let $X,Y,Z$ be the following languages
        \[
            \begin{array}{ccc}
            X = \{ a^{18\cdot 10^k} \mid k \in \bb{N}\} &
            Y = \{ a^{19\cdot 10^k} \mid k \in \bb{N} \} &
            Z = \{ a^{20\cdot 10^k} \mid k \in \bb{N}\}       
            \end{array}
        \]
        {We show that $\acost(X,Y){=}\frac{1}{18}$,   $\acost(Y,Z){=}\frac{1}{19}$, and $\acost(X,Z){=}\frac{1}{9}$. Thus, the triangle inequality fails, since  $\frac{1}{9}>\frac{1}{18} + \frac{1}{19}$.}
        
        As this is an unary alphabet $\ed(a^m,a^n) = \left|m-n\right|$.
        We denote the words in the languages as follows 
        \(
            \begin{array}{ccc}
            x_k =  a^{18\cdot 10^k} &
            y_k = a^{19\cdot 10^k} &
            z_k = a^{20\cdot 10^k}       
            \end{array}
        \).
        For every $x_k$ the following holds
        \begin{align*}
            &\ed(x_k, z_k) = |18\cdot 10^k - 20\cdot 10^k | = 2\cdot 10^k\\
            &\ed(x_k, z_{k-1}) = |18\cdot 10^k - 20\cdot 10^{k-1} | = 16\cdot 10^k\\
            &\ed(x_k, z_{k+1}) = |18\cdot 10^k - 20\cdot 10^{k+1} | = 182\cdot 10^k\\
        \end{align*}
        \noindent
        Thus we can see that $\arg \inf_{z \in Z}\ed(x_k,z) = z_k$.\\
        Furthermore for every $x_k$
        \begin{align*}
            &\ed(x_k, y_k) = |18\cdot 10^k - 19\cdot 10^k | = 1\cdot 10^k\\
            &\ed(x_k, z_{k-1}) = |18\cdot 10^k - 19\cdot 10^{k-1} | \ge 16\cdot 10^k\\
            &\ed(x_k, z_{k+1}) = |18\cdot 10^k - 19\cdot 10^{k+1} | = 172\cdot 10^k
        \end{align*} 
        Thus we can see that $\arg \inf_{y \in Y}\ed(x_k,y) = y_k$.\\
        Finally for every $y_k$ the following holds
        \begin{align*}
            &\ed(y_k, z_k) = |19\cdot 10^k - 20\cdot 10^k | = 1\cdot 10^k\\
            &\ed(y_k, z_{k-1}) = |19\cdot 10^k - 20\cdot 10^{k-1} | = 17\cdot 10^k\\
            &\ed(y_k, z_{k+1}) = |19\cdot 10^k - 20\cdot 10^{k+1} | = 181\cdot 10^k
        \end{align*}
        Thus we can see that $\arg \inf_{z \in Z}\ed(y_k,z) = z_k$.\\
        We can finally see that
        \[
        \begin{array}{c}
           \inf_{y \in Y} \frac{\ed(x_k,y)}{|x_k|} =  \frac{\ed(x_k,y_k)}{|x_k|} = \frac{10^k}{18\cdot 10^k} = \frac{1}{18}\\
           \inf_{z \in Z} \frac{\ed(x_k,z)}{|x_k|} =  \frac{\ed(x_k,z_k)}{|x_k|} = \frac{2\cdot 10^k}{18\cdot 10^k} = \frac{1}{9}\\
           \inf_{z \in Z} \frac{\ed(y_k,z)}{|y_k|} =  \frac{\ed(y_k,z_k)}{|y_k|} = \frac{10^k}{19\cdot 10^k} = \frac{1}{19}\\
        \end{array}
        \]
        Putting it all together we get 
        \begin{align*}
            \acost(X,Z) &= \lim_{n \to \infty} \sup_{\substack{x \in X  |x|\geq n}} \inf_{z \in Z} \tfrac{\ed(x,z)}{|x|}
            = \lim _{i \to \infty} \inf_{z \in Z} \tfrac{\ed(x_{k_i},z)}{|x_{k_i}|}
             = \frac{1}{9}.
\end{align*}

    \noindent
    Similarly $\acost(X,Y)=\frac{1}{18}$ and $\acost(Y,Z)=\frac{1}{19}$.
    Yet $\frac{1}{9}>\frac{1}{18} + \frac{1}{19}$.

    \end{proof}

\section{Proofs of \autoref{sec:AH}}

\claimasymptoticrelativedistancewelldefined*
 \begin{proof}
        Let $X,Y \subseteq M$ and let $b_k = \sup_{\substack{x \in X \\ s(x)\geq k}}\inf_{y \in Y} d(x,y)$.  As  $d(x,y) \le \sup d$ for every $x,y \in M$, we know that $0 \le b_k \le \sup d$ for every $k$. Moreover for every $k$ 
        \[
            b_k = \sup_{\substack{x \in X \\ s(x)\geq k}}\inf_{y \in Y} d(x,y) \geq \sup_{\substack{x \in X \\ s(x)\geq k+1}}\inf_{y \in Y} d(x,y) = b_{k+1}
        \]
        Hence $b_k$ is a bounded non-increasing sequence and thus converges. 
\end{proof}

\claimEquivalentdefinitiontoasymptoticrelativedistance*
\begin{proof}
        Let $(x_n)^\infty_{n=1}$ be a sequence in $X$ that satisfies $s(x_n) \ge n$ for every $n\in\bb{N}$. Then
        \[
            \inf_{y \in Y}d(x_n,y) \le \sup_{\substack{x \in X \\ s(x) \ge n}} \inf_{y \in Y} d(x,y)
        \]
        and
        \[
            \limsup_{n \to\infty}  \inf_{y \in Y}d(x_n,y) \le \limsup_{n \to\infty}  \sup_{\substack{x \in X \\ s(x) \ge n}} \inf_{y \in Y} d(x,y) = \dar(X,Y)
        \]
        and therefore
        \[
            \sup_{\substack{(x_n)^\infty_{n=1} \subseteq X\\ s(x_n) \ge n \,\forall n\in\bb{N}}} \lim_{n \to \infty} \inf_{y\in Y} d(x_n,y) \le \dar(X,Y).
        \]
        \smallskip
        As for the other direction let $\e >0$ and for every $n$ let $x_n$ be an element that satisfies
        \[
            s(x_n) \ge n \qquad \text{and} \qquad \inf_{y\in Y} d(x_n,y) \ge \sup_{\substack{x \in X \\ s(x) \ge n}} \inf_{y \in Y} d(x,y) - \e.
        \]
        Then
        \[
            \limsup_{n \to \infty } \inf_{y\in Y} d(x_n,y) \ge \limsup_{n \to \infty} \sup_{\substack{x \in X \\ s(x) \ge n}} \inf_{y \in Y} d(x,y) - \e = \dar(X,Y) - \e
        \]
        and thus
        \[
            \sup_{\substack{(x_n)^\infty_{n=1} \subseteq X\\ s(x_n) \ge n \,\forall n\in\bb{N}}} \lim_{n \to \infty} \inf_{y\in Y} d(x'_n,y) \ge \lim_{n \to \infty } \inf_{y\in Y} d(x_n,y) \ge \dar(X,Y).
        \]
        Finally
        \[
            \dar(X,Y) = \sup_{\substack{(x_n)^\infty_{n=1} \subseteq X\\ s(x_n) \ge n \,\forall n\in\bb{N}}} \lim_{n \to \infty} \inf_{y\in Y} d(x_n,y).
        \]
\end{proof}

\thmnormedinducedmetricshavetheasymptoticseparationproperty*

\begin{proof}
        Let $(x_k)^\infty_{k=1}$ be an $s$-unbounded sequence of elements, and $ (y_k)^\infty_{k=1}$ be  an $s$-bounded sequence of elements. Denote by $B$ a bound on the norm of the sequence $y_k$, there exists $n$ such that every $k \ge n$ satisfies
        \begin{equation}\label{eq: norm diff}
            \norm{x_k} \ge B \ge  \norm{y_k}
        \end{equation}
        Fixing $k\ge n$ we get
        \begin{align*}
            d(x_k, y_k) &= \norm{x_k - y_k} \\
                        &\ge \norm{\norm{x_k} - \norm{y_k}}\\
                        &= \norm{x_k} - \norm{y_k}\\
                        &\ge^{\eqref{eq: norm diff}} \norm{x_k} - B.
        \end{align*}
        Thus as $\left(\norm{x_k} - B\right)_{k} \underset{k \to \infty}{\to} \sup d$. Hence, we proved that $d$ has the asymptotic separation property.
\end{proof}

\section{Proofs of \autoref{sec:FLdistances}}

\claimwordmetricswithasp*
\begin{proof}
        Let $(x_n)^\infty_{n=1}$ and $(y_n)^\infty_{n=1}$ be length-bounded and length-unbounded  sequences respectively.
        \begin{itemize}
            \item 
   
        For $\ed$ we know that from some point $|y_n| > |x_n|$ and thus every edit path from $x_n$ to $y_n$ will need at least $|y_n| - |x_n|$ insertions. But $\lim_{n \to \infty} |y_n| - |x_n| = \infty$ and thus $\ed$ has the asymptotic separation property.
        \item
        From the same reasoning as in $\ed$ from some point onward
        \begin{align*}
            \lim_{n \to \infty} \ged(x_n, y_n) &= \lim_{n \to \infty} \tfrac{2\ed(x_n,y_n)}{|x_n|+|y_n|+ \ed(x_n,y_n)}\\
                    &\ge \lim_{n \to \infty} \tfrac{2\cdot (|y_n|-|x_n|)}{|x_n|+|y_n|+ \max\{|y_n|,|x_n|\}} = 1. 
        \end{align*}
        \item 
        For $\prf$ we get
        \begin{align*}
            \lim_{n \to \infty} \prf(x_n,y_n) &=\lim_{n \to \infty}|x_n|+|y_n|-2\cdot\max\{|z|\colon x_n,y_n\in z\cdot\Sigma^*\} \\
            &\ge \lim_{n \to \infty}|x_n|+|y_n|-2|x_n| = \infty
        \end{align*}
        \item
        For $\ced$ we know that from some point onward we will need at least $|y_n| - k$ edit operations in the sequence, where $k$ is the length bound of $(x_n)^\infty_{n=1}$. Thus
        \[
            \lim_{n \to \infty} \ced(x_n,y_n) \ge \lim_{n \to \infty} \sum^{|y_n|}_{i = k+1} \frac{1}{i} = \infty 
        \]
        \item
        Finally for $\ned$, Let $k$ be a bound on the lengths of $ (x_n)^\infty_{n=1}$ and fix $n$. Every edit path $p$ from $x_n$ to $y_n$ satisfies 
        \[
            |p| \le |x_n|+|y_n| \qquad \wgt(p) \ge \big||x_n| - |y_n| \big|
        \]
        Therefore from some point onward
        \[
            \ned(x_n,y_n) \ge \frac{\big||x_n| - |y_n| \big|}{|x_n|+|y_n|} \ge \frac{|y_n| - k}{|y_n| + k}
        \]
        Thus
        \[
             \lim_{n \to \infty} \ned(x_n,y_n) \ge  \lim_{n \to \infty} \frac{|y_n| - k}{|y_n| + k} = 1.
        \]
             \end{itemize}
    \end{proof}
\clmfinitesubsetindif*
\begin{proof}
    We prove the claim on $\metric{D}\in\{ \nedAH, \gedAH, \cedAH\}$ using the asymptotic directional distances of the different metrics. Let $\elementmetric{d} \in \{\ned, \ced, \ged\}$. It is trivial that
    \begin{align*}
        \dar(X{\cup}F,Y) &= \lim_{k \to \infty} \sup_{\substack{x \in X\cup F \\ |x|\geq k}}\inf_{y \in Y} \elementmetric{d}(x,y) \\&= \lim_{k \to \infty} \sup_{\substack{x \in X \\ s(x)\geq k}}\inf_{y \in Y} \elementmetric{d}(x,y) =  
        \dar(X,Y).
    \end{align*}
    As for the other direction, 
    first we also notice that it is trivial that 
    \[
        \dar(Y,X \cup F) \le \dar(Y,X). 
    \]
    Let $\e > 0$ and let $(y_n)^\infty_{n=1} \subseteq Y$ be a sequence with length tends to $\infty$ that satisfy
    \[
        \sup_{\substack{y \in Y \\ |y|\geq k}}\inf_{x \in X} \elementmetric{d}(y,x) \le \inf_{x \in X} \elementmetric{d}(y_k, x) + \e.
    \]
    \textbf{Case 1} - There exists $N$ such that for every $n \ge N$ and every $x \in X$ the following is satisfied
    \[
        \elementmetric{d}(y_n,x) > \inf_{x \in X \cup F} \elementmetric{d}(y_n,x) + \e.
    \]
    Thus from some point onward the closest word to $y_n$ with respect to metric $d$ is in $F$. But as $\elementmetric{d}$ has the asymptotic separation property we get
    \begin{align*}
        \dar(Y,X\cup F) &= \sup_{\substack{(y'_n)^\infty_{n=1} \subseteq Y \\ |y\_n| \ge n}}\limsup _{n \to \infty} \inf_{x \in X \cup F}\elementmetric{d}(y'_n,x)\\
            &= \limsup _{n \to \infty} \inf_{x \in X \cup F}\elementmetric{d}(y_n,x)\\
            &= \limsup _{n \to \infty} \inf_{x \in F}\elementmetric{d}(y_n,x)\\
            &= \sup \elementmetric{d}.
    \end{align*}
    But then
    \[
        \sup \elementmetric{d} \le \dar(Y,X\cup F) \le \dar(Y,X) \le \sup \elementmetric{d}.
    \]
    \textbf{Case 2} - There exists infinite subsequence $(y_{n_i})^\infty_{i=1}$ such that there exists $x_i \in X$ that satisfy 
    \[
        \elementmetric{d}(y_{n_i}, x_i) \le \inf_{x \in X \cup F} \elementmetric{d}(y_{n_i}, x) + \e.
    \]
    Thus
    \begin{align*}
        \dar(Y,X\cup F) &= \sup_{\substack{(y'_n)^\infty_{n=1} \subseteq Y \\ |y\_n| \ge n}}\limsup _{n \to \infty} \inf_{x \in X \cup F}\elementmetric{d}(y'_n,x)\\
                        &\ge \limsup _{i \to \infty} \inf_{x \in X \cup F}\elementmetric{d}(y_{n_i},x)\\
                        &= \e +  \limsup _{i \to \infty} \elementmetric{d}(y_{n_i},x_i) \\
                        &\ge  \e +  \limsup _{i \to \infty} \inf_{x \in X} \elementmetric{d}(y_{n_i},x) \\
                        &\ge 2\e + \lim_{k \to \infty} \sup_{\substack{y \in Y \\ |y|\geq k}} \inf_{x \in X} \elementmetric{d}(y,x) \\
                        &= 2\e + \dar(Y,X).
    \end{align*}
    And as $\e$ is arbitrary small
    \[
        \dar(Y, X\cup F) \ge \dar(Y, X). 
    \]
    
\end{proof}

\clmboundededits*
\begin{proof}
    We prove for the asymptotic directional distance as it is sufficient. Assume that there exists $k \in \bb{N}$ such that for every $x \in X$ there exists $y \in Y$ that satisfy $\ed(x,y) \le k$ and for every $y\in Y$ there exists $x\in X$ such that $\ed(x,y)<k$. If $X$ is infinite if and only if $Y$ is infinite thus if one is finite  then  
    \[
        \nedar(X,Y) = 0, \qquad \gedar(X,Y) = 0, \qquad \cedar(X,Y) = 0.
    \]
    From this point onward both are infinite. For $\nedar$ we get
    \[
        \nedar(X,Y) = \lim_{n \to \infty} \sup_{\substack{x \in X\\ |x| \ge n}} \inf_{y \in Y} \ned(x,y) \le \lim_{n \to \infty} \sup_{\substack{x \in X\\ |x| \ge n}} \frac{k}{|x|} = 0.
    \]
    For $\gedar$
    \begin{align*}
        \gedar(X,Y) &= \lim_{n \to \infty} \sup_{\substack{x \in X\\ |x| \ge n}} \inf_{y \in Y} \ged(x,y)
        \\ &= \lim_{n \to \infty} \sup_{\substack{x \in X\\ |x| \ge n}} \inf_{y \in Y} \frac{2 \ed(x,y)}{|x|+|y|+\ed(x,y)}
        \le \lim_{n \to \infty} \sup_{\substack{x \in X\\ |x| \ge n}} \frac{2k}{|x|} = 0.
    \end{align*}
    And for $\cedar$
    \[
        \cedar(X,Y) = \lim_{n \to \infty} \sup_{\substack{x \in X\\ |x| \ge n}} \inf_{y \in Y} \ced(x,y) 
        \le \lim_{n \to \infty} \sup_{\substack{x \in X\\ |x| \ge n}} \frac{k}{|x|-k} = 0.
    \]
\end{proof}

\clmprecentagenature*
\begin{proof}
    For every $(a^ib)^k \in (a^ib)^*$ there need to be done at least $k$ operations of deletion/replace in order to get to $a^m$ for some $m$. As for $\ned$ replacing is better then deletion because it results in a longer path, and any other addition to that path will only increase its cost thus
    \begin{align*}
        \nedar((a^ib)^*, a^*) &= \lim_{n \to \infty} \sup_{\substack{x \in (a^ib)^*\\ |x| \ge n}} \inf_{y \in a^*} \ned(x,y)\\
                            &= \lim_{n \to \infty} \tfrac{n}{n\cdot(j+1)} = \tfrac{1}{j+1}
    \end{align*}
    As  for $\ged$ we know 
    \begin{align*}
        \gedar((a^ib)^*, a^*) &= \lim_{n \to \infty} \sup_{\substack{x \in (a^ib)^*\\ |x| \ge n}} \inf_{y \in a^*} \ged(x,y)\\
                            &= \lim_{n \to \infty} \sup_{\substack{x \in (a^ib)^*\\ |x| \ge n}} \inf_{y \in a^*} \tfrac{2 \ed(x,y)}{|x|+|y|+\ed(x,y)}\\
                            &=  \lim_{n \to \infty} \tfrac{2n}{n\cdot(i+1) + n\cdot(i+1) + n}\\
                            &= \tfrac{2}{2(i+1) + 1}
    \end{align*}
    Finally for $\ced$ we get
     \begin{align*}
        \cedar((a^ib)^*, a^*) &= \lim_{n \to \infty} \sup_{\substack{x \in (a^ib)^*\\ |x| \ge n}} \inf_{y \in a^*} \ced(x,y)\\
                            &=  \lim_{n \to \infty} \tfrac{n}{n\cdot (i +1)} = \tfrac{1}{1+i}
    \end{align*}
    As all terms decrease as $i$ increases, the directional case is proved. Note that the exact same arguments can be used for the other asymptotic directional distance.       
\end{proof}

\claimahnedsatisfiesstrictprecentage*
\begin{proof}
    We show below that the strict percentage property fails for $\gedAH$ and $\cedAH$.
For $\nedAH$, the property follows directly from the definition and is straightforward to verify.
    \begin{itemize}
        \item 
        We show that $\gedAH$ violates the property. Let $X = a^*$ and $Y = b^*$. Then
    \begin{align*}
        \gedar(a^*,b^*) &= \lim_{n \to \infty} \sup_{\substack{x \in a^* \\ |x| \ge n}} \inf_{y \in b^*}\ged(x,y)\\
                        &= \lim_{n \to \infty} \sup_{\substack{x \in a^* \\ |x| \ge n}} \inf_{y \in b^*}\tfrac{2 \cdot \ed(x,y)}{|x|+|y|+\ed(x,y)}\\
                        &= \lim_{n \to \infty} \tfrac{2 \cdot  n}{n +n +n}\\
                        &= \tfrac{2}{3}
    \end{align*}
         \item We show that $\cedAH$ violates the property.
        Let $X = (ab)^*$ and $Y = (aab)^*$. Then
    \begin{align*}
        \cedar((ab)^*,(aab)^*) &= \lim_{n \to \infty} \sup_{\substack{x \in (ab)^* \\ |x| \ge n}} \inf_{y \in (aab)^*}\ced(x,y)\\
                        &= \lim_{n \to \infty} \sum^n_{i=1} \tfrac{1}{2\cdot n+i}\\
                        &= \lim_{n \to \infty} \tfrac{1}{n} \sum^n_{i=1} \tfrac{1}{2+\tfrac{i}{n}}\\
                        &= \int^1_{0}\tfrac{1}{2 + x} dx \\
                        &= \Big[ \ln(2+x) \Big]^1_0 \\
                        &= \ln{3} - \ln{2} \\
                        &= \ln(\tfrac{3}{2}) \approx 0.4054.
    \end{align*}
    \end{itemize}
\end{proof}

\begin{restatable}[GED is lesser than NED]{claim}{claimgedlened}\label{claim: ged le ned}
    Let $x$ and $y$ be words. Then $\ged(x,y) \le \ned(x,y)$.
\end{restatable}
\begin{proof}
         Let $\ep$ be an edit path from $x$ to $y$ such that $\ned(x,y) = \frac{\wgt(\ep)}{|\ep|}$. Denote $I,D,S,N$ the number of insertions, deletions, substitutions and no-op respectively. Thus
            \[
                |x| = D + S+N \quad |y|  = I + S + N \quad \wgt(\ep) = I+D + S \quad |\ep| = I+D+S+N. 
            \]
            Hence we get
            \[
                |x|+|y|+\wgt(\ep) = 2\cdot(I+D+N+S) + S \ge 2\cdot |\ep|
            \]
            And thus
            \[
                \ned(x,y) =  \tfrac{\wgt(\ep)}{|\ep|} \ge  \tfrac{2\cdot \wgt(\ep)}{|x|+|y|+\wgt(\ep)} \ge \tfrac{2\cdot \ed(x,y)}{|x|+|y|+\ed(x,y)} = \ged(x,y).
            \]
\end{proof}

\lemrelations*
    \begin{proof}
    \begin{itemize}
        \item 

         We first prove the inequality $\nedAH(X,Y)\leq \nedH(X,Y)$.
         We have seen in Claim~\ref{claim: asymptotic relative distance well defined} that the sequence
        \[
            \beta_n = \sup_{\substack{x \in X\\ |x| \ge n}} \inf_{y \in Y} \ned(x,y). 
        \]
        is a non-increasing sequence, thus
        \[
            \nedar(X,Y) = \lim_{n \to \infty} \beta_n \le \beta_0. 
        \]
        Furthermore 
        \[
            \beta_0 = \sup_{\substack{x \in X\\ |x| \ge 0}} \inf_{y \in Y} \ned(x,y) = \sup_{\substack{x \in X}} \inf_{y \in Y} \ned(x,y).
        \]
        Without loss of generality we get 
        \begin{align*}
            \nedAH(X,Y) &= \max \left\{\nedar(X,Y),\nedar(Y,X)\right\}\\
                            &= \max \left\{\lim_{n \to \infty} \sup_{\substack{x \in X\\ |x| \ge n}} \inf_{y \in Y} \ned(x,y) ,\lim_{n \to \infty} \sup_{\substack{y \in Y\\ |y| \ge n}} \inf_{x \in X} \ned(y,x)\right\}\\
                           &\le \max \left\{\sup_{\substack{x \in X}} \inf_{y \in Y} \ned(x,y),\sup_{\substack{y \in Y}} \inf_{x \in X} \ned(y,x)\right\}\\
                           &= \nedH(X,Y).
        \end{align*}
        \item 
        Turning to the inequality 
        $\nedH(X,Y)\leq \edH(X,Y)$ we will partition into 2 cases:\\
        \textbf{Case 1} - If $\edH(X,Y) = 0$ then $X=Y$ and thus $\nedH(X,Y) = 0 \le \edH(X,Y)$.\\
        \textbf{Case 2} - Otherwise $\edH(X,Y) \ge 1$ but as $\ned$ is 1-bounded metric we get 
        \[
            \nedH(X,Y) \le 1 \le \edH(X,Y). 
        \]
        \item
        Last, for inequality 
        $\gedAH(X,Y)\leq \nedAH(X,Y)$ given $\e>0$ we 
        \begin{align*}
            \gedar(X,Y) &= \lim_{n \to \infty} \sup_{\substack{x \in X \\ |x| \ge n}} \inf_{y \in Y} \ged(x,y) \\
                        &= \lim_{n \to \infty} \sup_{\substack{x \in X \\ |x| \ge n}} \inf_{y \in Y} \frac{2\cdot \ed(x,y)}{|x|+|y|+ \ed(x,y)} \\
                        &= \sup_{\substack{(x_n)^\infty_{n=1}\subseteq X \\ |x_n| \ge n}}\lim_{n \to \infty} \inf_{y \in Y} \frac{2\cdot \ed(x_n,y)}{|x_n|+|y|+ \ed(x_n,y)} \\
                        &= \e + \lim_{n \to \infty} \inf_{y \in Y} \frac{2\cdot \ed(x_n,y)}{|x_n|+|y|+ \ed(x_n,y)} \\
                        &\le^{\eqref{claim: ged le ned}} \e + \lim_{n \to \infty} \inf_{y \in Y} \ned(x_n, y) \\
                        &\le \e + \sup_{\substack{(x_i)^\infty_{i=1}\subseteq X \\ |x_i| \ge n}}\lim_{n \to \infty} \inf_{y \in Y} \ned(x,y) \\
                        &= \e + \lim_{n \to \infty} \sup_{\substack{x \in X \\ |x| \ge n}} \inf_{y \in Y} \ned(x,y) \\
                        &= \e + \nedar(X,Y).
        \end{align*}
        Hence $\gedar(X,Y) \le \nedar(X,Y)$ and finally
        \begin{align*}
            \nedAH(X,Y) &= \max \left\{\nedar(X,Y),\nedar(Y,X)\right\}\\
            &\ge \max \left\{\gedar(X,Y),\gedar(Y,X)\right\} = \gedAH(X,Y).
        \end{align*}
            \end{itemize}
    \end{proof}

\claimessenceequivalences*

\begin{claim}\label{clm:ir-fails-triangle}
    $\ir^\varphi$ does not satisfy the triangle inequality.
\end{claim}
\begin{proof}
    Let $\varphi(x,y)$ be the predicate $x=y$. It follows that $$\ir^\varphi(X,Y)=\frac{\ir(X\cap Y)}{\ir(X\cup Y)}.$$
    Consider the languages
    \[X = a^* b^+ c^*, \qquad
    Y = c^+, \qquad
    Z = b^+ c^* d^*.\]
    We have $X\cap Y=\emptyset$ and $Y\cap Z=\emptyset$, hence $
    \ir^\varphi(X,Y)=0 $ and $ \ir^\varphi(Y,Z)=0$.
    However, $X\cap Z = b^+ c^*$, and the set of words of length $n$ in $X\cap Z$ is
    $(X\cap Z)^{(n)}=\{\,b^i c^j \mid i+j=n,\ i>0\,\},$
    which grows linearly with $n$.
    Consequently, $\ir^\varphi(X,Z)>0$, while
$\ir^\varphi(X,Y)+\ir^\varphi(Y,Z)=0$.
    Thus, the triangle inequality fails.
\end{proof}

\section{Proofs of \autoref{sec:computation}}

\subsection{Proofs of \autoref{sec:nedAH-hardness}}

\thmpspacehardness*
\begin{proof} 
        The proof is via a reduction from the universality problem of an NFA, which is a known $\pspace$-hard problem. Let $A$ be a fixed NFA. We construct an NFA $T$ that recognize $\big(\# \cdot \sema{A}\big)^*$ where $\# \notin \Sigma$. Let $\Gamma = \Sigma \cup \{\#\}$ we consider calculatation of $\nedar\big(\Gamma^*, \sema{T}\big)$.
        
        If $\sema{A} = \Sigma^*$ then $\sema{T}$ is the language of all words over $\Gamma$ that starts with $\#$. For every $w \in \Gamma^*$ there is at most $1$ edit operations to make it start with $\#$. Thus by Lemma~\ref{clm:bounded-edits}
        \[
            \nedar\big(\Gamma^*, \sema{T}\big) = 0.
        \]
        
        Otherwise $\sema{A} \neq \Sigma^*$. Then there exists $x \in \Sigma^* \setminus \sema{A}$. Consider the sequence  $(x_k)^\infty_{k=1}$ where $x_k = (\#x)^k$. This sequence of words with length tends to $\infty$ is used to prove that $\nedar\big(\Gamma^*, \sema{T}\big) >0$.
        
        Fixing $k$ use the partition of $x_k$ to $k$ copies of $\#x$, and denote the $i$'th copy as $\#x_i$. We claim that in consecutive pairs $\#x_i$ and $\#x_{i+1}$ there is at least 1 edit operation of cost $1$.
        
        Let $p$ be an edit path from $x_k$ to some word in $\sema{T}$. Looking at the edit operation $p$ is making in $\#x_i$ one of the following must occur:
        (i) If there is an edit path of cost 1 in $x$ inside $x_i$ or on $\#$ of $\#x_i$ then the claim holds.
        (ii) Otherwise as $x \notin \sema{A}$ there must be a delete/substitution operation on the $\#$ of $\#x_{i+1}$. Thus the claim holds. 
        
        A corollary from the claim is that every edit path from $x_k$ to $\sema{T}$ weighs at least $\frac{k}{2}$. Since the cost of a path is its weight divided by its length, to bound the cost of the path from below, we need to consider long optimal edit paths.

        Let $n$ be the size of the NFA $T$. Thus by this reasoning and \autoref{claim:bound-on-ned-edit-path} 
        for every $y \in \sema{T}$
        \begin{align*}
            \ned(x_k, y) &\ge \tfrac{k}{2 \cdot n \cdot (k \cdot (|x|+1) +1)}\\
                         & \ge  \tfrac{k}{2 \cdot n\cdot k \cdot (|x|+2)}\\
                         &=  \tfrac{1}{2 \cdot n \cdot (|x|+2)} 
        \end{align*}
        As every $y \in \sema{T}$ satisfies this inequality we get 
        \begin{align*}
            \nedar\big(\Delta^*,\sema{T}\big) &= \lim_{k \to \infty} \sup_{\substack{x \in X \\ |x| \ge n}} \inf_{y \in \sema{T}} \ned(x, y) \\
                        &\ge \lim_{k \to \infty} \inf_{y \in \sema{T}} \ned(x_k, y)\\
                        &\ge  \lim_{k \to \infty} \tfrac{1}{2 \cdot n \cdot (|x|+2)} \\
                        &= \tfrac{1}{2 \cdot n \cdot (|x|+2)} \\
                        &> 0.
        \end{align*}
\end{proof}

\claimboundonnededitpath*
\begin{proof}
It was shown in~\cite{FiliotMRST20,FismanGW23} that the quantity $\iined(X,Y)=\inf_{x\in X}\inf_{y\in Y}\ned(x,y)$ can be computed in the edit-distance graph
$\mathcal{G}_{\textsc{ed}}(A_1,A_2)$, where $A_1$ and $A_2$ are NFAs recognizing $X$ and $Y$, respectively.\footnote{This graph has $n_1{\cdot}n_2$ vertices, where $n_1$ and $n_2$ are the numbers of states of $A_1$ and $A_2$.
An edge $(\langle q_1,q_2\rangle,\langle q'_1,q'_2\rangle)$ labeled by $\swap{\sigma_1}{\sigma_2}\in\Gamma$ exists if 
$q'_1\in\delta(q_1,\sigma_1)$ and $q'_2\in\delta(q_2,\sigma_2)$. Its weight correspond to the weight of the edit operation  $\swap{\sigma_1}{\sigma_2}$. }
To compute $\inf_{y\in Y}\ned(x,y)$ we can thus take $X=\{x\}$ and let $A_1$ be the standard acyclic NFA with $|x|+1$ states recognizing $\{x\}$.
Hence $\mathcal{G}_{\textsc{ed}}(A_1,A_2)$ has $(|x|{+}1)\cdot n$ vertices, and every edge either advances the $A_1$-component or is an insertion edge of the form $\swap{\varepsilon}{\sigma}$ (a self-loop in the $A_1$-component).

We show that it suffices to consider target words of length at most $n(|x|+1)$.
Fix $y\in Y$ and an edit path $p$ from $x$ to $y$.
Let $I$ and $D$ denote the numbers of insertions and deletions in $p$, respectively.
Then $|y|=|x|+I-D$, and hence $I\ge |y|-|x|$.
In particular, if $|y|>n(|x|+1)$ then
\[
I > n(|x|+1)-|x| \ge (n-1)(|x|+1).
\]
Partition the insertions of $p$ into the $|x|+1$ insertion blocks occurring between successive letters of $x$ (including before the first and after the last letter).
By the pigeonhole principle, some insertion block has length at least $n$.
Along that block, the run of the $A_2$-component visits at least $n+1$ states, and therefore contains a nontrivial cycle $c$ whose projection on the $A_1$-component is $\varepsilon$ (that is, $c$ consists solely of insertion steps).

Removing $c$ from the $A_2$-run preserves acceptance and yields a new edit path $p'$ from $x$ to some $y'\in Y$ with $|y'|<|y|$.
Moreover, $\wgt(c)=|c|$ (since we assume uniform weights).
Thus,
\[
\ned(x,y)=\tfrac{\wgt(p)}{|p|}
=\tfrac{\wgt(p')+\wgt(c)}{|p'|+|c|}
=\tfrac{\wgt(p')+|c|}{|p'|+|c|}
\ge \tfrac{\wgt(p')}{|p'|}
\ge \ned(x,y').
\]
Iterating this trimming argument yields a word $y'\in Y$ with $|y'|\le n(|x|+1)$ attaining $\inf_{y\in Y}\ned(x,y)$.
Since there are only finitely many words of length at most $n(|x|+1)$, this infimum is a minimum, as claimed.
\end{proof}

\claimnedarsmalleracost*
   
\begin{proof}
    \begin{align*}
        \nedar(X,Y) &= \lim_{n \to \infty} \sup_{\substack{x \in X \\ |x| \ge n}} \inf_{y \in Y} \ned(x,y)\\
        &= \lim_{n \to \infty} \sup_{\substack{x \in X \\ |x| \ge n}} \inf_{y \in Y} \min_{p:x \leadsto y} \tfrac{\wgt(p)}{|p|} \\
        &\le \lim_{n \to \infty} \sup_{\substack{x \in X \\ |x| \ge n}} \inf_{y \in Y} \min_{p:x \leadsto y} \tfrac{\wgt(p)}{|x|} \\
        &\le \lim_{n \to \infty} \sup_{\substack{x \in X \\ |x| \ge n}} \inf_{y \in Y}  \tfrac{\min_{p:x \leadsto y} \wgt(p)}{|x|} \\        
        &= \lim_{n \to \infty} \sup_{\substack{x \in X \\ |x| \ge n}} \inf_{y \in Y} \tfrac{\ed(x,y)}{|x|} \\
        &= \acost(X,Y).
    \end{align*}
\end{proof}

\claimnedarlargeracostdividenY*
\begin{proof}
    Let $x$ be a fixed word in $X$. By \autoref{claim:bound-on-ned-edit-path} there exists $y_x{\in}Y$ that attains the $\inf_{y \in Y}\ned(x,y)$ and furthermore there is an edit path $p:x\leadsto y_x$ achieving $\ned(x,y_x)$ and satisfying $|p|\le d (|x|+1)$. Thus
    \begin{align*}
        \inf_{y \in Y} \tfrac{\ed(x,y)}{|x|} &\le \tfrac{\ed(x,y_x)}{|x|}\\
        &\le \tfrac{\wgt(p)}{|x|} \\
        &=  \tfrac{\ned(x,y_x) \cdot|p|}{|x|}\\
        &= \ned(x,y_x) \cdot \tfrac{|p|}{|x|}\\
        &\le  \inf_{y \in Y} \ned(x,y) \cdot d \cdot \left( 1 + \tfrac{1}{|x|}\right).
    \end{align*}
    Hence we get the following
    \begin{align*}
        \sup_{\substack{x \in X \\ |x| \ge n}} \inf_{y \in Y} \tfrac{\ed(x,y)}{|x|} & \le  \sup_{\substack{x \in X \\ |x| \ge n}} \inf_{y \in Y} \ned(x,y) \cdot d \cdot \left( 1 + \tfrac{1}{|x|}\right) \\
        & \le d \cdot \left( 1 + \tfrac{1}{n}\right) \sup_{\substack{x \in X \\ |x| \ge n}} \inf_{y \in Y} \ned(x,y).  
    \end{align*}
    Taking the limit on $n$ on both sides we get
    \begin{align*}
        \acost(X,Y)&=\lim_{n\to \infty} \sup_{\substack{x \in X \\ |x| \ge n}} \inf_{y \in Y} \tfrac{\ed(x,y)}{|x|} \\
        & \le d \cdot \lim_{n\to \infty} \sup_{\substack{x \in X \\ |x| \ge n}} \inf_{y \in Y} \ned(x,y) \\ &
        = d \cdot \nedar(X,Y).
    \end{align*}
\end{proof}

\lemmaahrsandwitch*
\begin{proof}
By \autoref{claim: nedar smaller acost} and \autoref{claim: nedar larger acost divide nY}, the value $\acost(X,Y)$ provides both an upper bound and, up to a factor of $d$, a lower bound on $\nedar(X,Y)$.
Consequently, it suffices to compute $\acost(X,Y)$ in order to obtain a constant-factor approximation of $\nedar(X,Y)$.

The lemma now follows directly from the fact that Benedikt et al.~\cite{BenediktPR14} provide a \coNExp-time algorithm for computing $\acost$ between regular languages.
\end{proof}

\subsection{Proofs of \autoref{sec:nedAH-bounded-CFL}}
\thmLP*

We prove each direction separately, first direction in \autoref{sec:LPA} and the second in \autoref{sec:LPB}

        \subsubsection{From edit-path to LP solution}\label{sec:LPA}
                \begin{claim}\label{claim: Approximating every path by LP}
                    There exists constants $C,M>0$ depending only on the finite family of graphs $\{\cyclicedgraph{i,j}\}, \Pi$ and $U,V$, such that for any $x = w^{\setX}_{\vec{\alpha}}$ with $\vec{a} = \dir(\vec{\alpha})$,  and $y = w^{\setY}_{\vec{\beta}} \in Y$ and a edit path $\ep:x \leadsto y$ there exists $\vec{a'} \in D_{\setX}$ that satisfy
                    \[
                        \tfrac{\wgt(\ep)}{|\ep|} + \tfrac{C}{|x|} \ge  \min_{\pi \in \Pi}\OPT_{\vec{a'}, \pi} = F(\vec{a'}) 
                    \]
                    and {$\|\vec{a'}- \vec{a}\|_{\infty} \le \tfrac{M}{|x|}$}.
                    \end{claim}
                    \begin{proof}
                        We start by fixing the constant $C_0 \defeq n \cdot \max_{1 \le i \le n} |x_i| + m \cdot \max_{1 \le j \le m} |y_j|$ that only depends on the family of graphs $\{\cyclicedgraph{i,j}\}$.
                        Let {$x = w^{\setX}_\vec{\alpha}$, $y = w^{\setY}_{\vec{\beta}}$,} $p$ be an edit path from $x$ to $y$ and $\pi = \pi(p)$ be the induced interleaving. Note that $p$ can be partitioned into $p_{i,j}$ continuous walks in $\cyclicedgraph{i,j}$ for all $(i,j) \in \pi$. We denote by $N^{i,j}_e$ the number of times edge $e$ is traversed in $p_{i,j}$.

                        In order to address the flow conservation constraint in each rectangle in $\pi$, we notice that for each $p_{i,j}$ the only vertices in $\cyclicedgraph{i,j}$ that might not satisfy flow conservation are the start and end vertices of $p_{i,j}$. Thus we can add to each $p_{i,j}$ a path in $\cyclicedgraph{i,j}$ that connects the end vertex to the start vertex of $p_{i,j}$, creating $\widehat p_{i,j}$ that is a cycle in $\cyclicedgraph{i,j}$  and also satisfies
  
                        \[
                            |\widehat p_{i,j}| \le |p_{i,j}| + \diam(\cyclicedgraph{i,j}) = |p_{i,j}| + |x_i|+|y_j|
                        \]
                        where $\diam(\cyclicedgraph{i,j})$ is the diameter of the graph $\cyclicedgraph{i,j}$.

                      \noindent
                        Let $\widehat p$ be the concatenation of $\widehat p _{i,j}$, then
                        \[
                            |\widehat p| \le |p| + n\cdot \max_{i}|x_i| + m \cdot \max_j|y_j| = |p| + C_0 
                        \]
                        and
                        \[
                            \wgt(\widehat p) \le \wgt(p) + n\cdot \max_{i}|x_i| + m \cdot \max_j|y_j| = \wgt(p) + C_0. 
                        \]
                        As before we define $\widehat N^{i,j}_e$ as the number of times edge $e$ is traversed in $\widehat p_{i,j}$. Because we started from {$x  = w^{\setX}_\vec{\alpha} \in X$ and $y = w^{\setY}_\vec{\beta} \in Y$}, we know that for every $i$ (and resp. every $j$)
                        \[
                            |x_i|\cdot  \vec{\alpha_i} = \sum_{\substack{(i,j) \in \pi \\ e \in E(\cyclicedgraph{i,j})}} N^{i,j}_e \cdot \indicator_{x_i}(e) 
                            \qquad 
                            |y_j|\cdot  \vec{\beta_j} = \sum_{\substack{(i,j) \in \pi \\ e \in E(\cyclicedgraph{i,j})}} N^{i,j}_e \cdot \indicator_{y_j}(e).
                        \]
                        And summing all $i$ (resp. all $j$) we get
                        \begin{align*}
                             \sum_{\substack{(i,j) \in \pi \\ e \in E(\cyclicedgraph{i,j})}} \widehat N^{i,j}_e \cdot \indicator_{x_i}(e) = |x_i|\cdot   \vec{\alpha_i} + C_{x_i}
                            \\
                            \sum_{\substack{(i,j) \in \pi \\ e \in E(\cyclicedgraph{i,j})}} \widehat N^{i,j}_e \cdot \indicator_{y_j}(e) = |y_j|\cdot \vec{\beta_j} + C_{y_j}
                        \end{align*}
                        where $C_{x_i}$ is the number of edges that consume symbol from origin that were added in the cycle-closing path in some $\cyclicedgraph{i,j}$ for all $(i,j) \in \pi$, and the same for $C_{y_j}$ about target consuming symbols. Hence $C_{x_i}, C_{y_j} \le C_0$.

                        Note that $\widehat p_{i,j}$ is a cycle in $\cyclicedgraph{i,j}$ for every $1 \le i \le n$ and $1\le j \le m$ such that $(i,j) \in \pi$. Thus its projection on symbols spelled from $x_i$ and  $y_j$ is divisible by $|x_i|$ and $|y_j|$ respectively. Hence for every $i$ and $j$ there exists some $T_{y_j}, T_{x_i}\in\mathbb{N}$ that satisfy
                        \[
                        \begin{split}
                        T_{x_i}\cdot |x_i| = \sum_{\substack{(i,j) \in \pi \\ e \in E(\cyclicedgraph{i,j})}} \widehat N^{i,j}_e \cdot \indicator_{x_i}(e) = |x_i|\cdot \vec{\alpha_i} + C_{x_i}
                        \\
                            T_{y_j}\cdot |y_j| = \sum_{\substack{(i,j) \in \pi \\ e \in E(\cyclicedgraph{i,j})}} \widehat N^{i,j}_e \cdot \indicator_{y_j}(e) = |y_j|\cdot \vec{\beta_j} + C_{y_j}.
                        \end{split}
                        \]
                        Thus, $C_{y_j}$ is divisible by $|y_j|$ and $C_{x_i}$ is divisible by $|x_i|$. Let $h^y_j \cdot |y_j| = C_{y_j}$ and $h^x_i \cdot |x_i| = C_{x_i}$ and consider the vectors

             \[
                                         \vec{h^x} = 
                            \begin{pmatrix}
                                \frac{C_{x_1}}{|x_1|}\\
                                \vdots \\
                                \frac{C_{y_n}}{|x_n|}
                            \end{pmatrix}  
                            \qquad
                            \vec{h^y} =
                            \begin{pmatrix}
                                \frac{C_{y_1}}{|y_1|}\\
                                \vdots \\
                                \frac{C_{y_m}}{|y_m|}
                            \end{pmatrix}
                        \]
                         
                        Then {$\| \vec{h^x}\|_{\infty},\| \vec{h^y}\|_{\infty} \le C_0$ and}
                        \[
                        \begin{split}
                        \sum_{\substack{(i,j) \in \pi \\ e \in E(\cyclicedgraph{i,j})}} \widehat N^{i,j}_e \cdot \indicator_{x_i}(e) &= |x_i|\cdot \vec{\alpha_i} + C_{x_i} = |x_i| \cdot \left(\vec{\alpha_i} + h^x_i\right)\\
                        \sum_{\substack{(i,j) \in \pi \\ e \in E(\cyclicedgraph{i,j})}} \widehat N^{i,j}_e \cdot \indicator_{y_j}(e) &= |y_j|\cdot \vec{\beta_j} + C_{y_j} = |y_j|\cdot \left(\vec{\beta_j} + h^y_j \right). 
                        \end{split}
                        \]
                        Let 
                        \[
                        \vec{\delta^x} \defeq - \vec{h^x} +  C_0 \cdot \sum^{k}_{q=1} \vecInd{u}{q} \in \bb{R}^n \qquad 
                        \vec{\delta^y} \defeq - \vec{h^y} +  C_0 \cdot \sum^{r}_{s=1} \vecInd{v}{s} \in \bb{R}^m. 
                        \]
                        Then for every $i,j$ we know $\vec{\delta^x_i}, \vec{\delta^y_j} \ge 0$, and we finally get new vectors in $S_\setX$ and $S_\setY$, respectively: 
                        \[
                        \begin{split}
                            \vec{\alpha} + \vec{h^x} + \vec{\delta^x} &= \vec{\alpha} + \vec{h^x} - \vec{h^x} +  C_0 \cdot \sum^{k}_{q=1} \vecInd{u}{q} = \vec{\alpha} +  C_0 \cdot \sum^{k}_{q=1} \vecInd{u}{q} \in S_X\\
                            \vec{\beta} + \vec{h^y} + \vec{\delta^y} &= \vec{\beta} + \vec{h^y} - \vec{h^y} +  C_0 \cdot \sum^{r}_{s=1} \vecInd{v}{s} = \vec{\beta} + C_0 \cdot \sum^{r}_{s=1} \vecInd{v}{s} \in S_Y.
                        \end{split}
                        \]
                        Hence by adding $\vec{\delta^x_i}$ independent cycles of $x_i$ in some $\cyclicedgraph{i,j}$ for each $1\le i \le n$, and adding $\vec{\delta^y_i}$ independent cycles of $y_j$ in some $\cyclicedgraph{i,j}$ for each $1\le j \le l$ we get new paths $\widetilde p_{i,j}$ and their concatenation $\widetilde p$ that satisfy
                        \[
                            |\widetilde p| \le |\widehat p| + |\pi|\cdot C_0^2\cdot  \left( \max_{s,j} \vecIndEnt{v}{s}{j} + \max_{q,i} \vecIndEnt{u}{q}{i} \right) \le |p| + C_1
                        \]
                        for some constant $C_1$ depending only on $\Pi$, the generators and $C_0$. The same can be said about the weight of $\widetilde p $ and thus
                        \[
                            \wgt(\widetilde p) \le \wgt(p) + C_1.
                        \]
                        Moreover by defining $\widetilde N^{i,j}_e$ as the number of times edge $e$ is traversed in $\widetilde p_{i,j}$ we get 
                        \begin{align}\label{eq:tilde-Nij-x}
                        \begin{pmatrix}
                            \sum_{\substack{(1,j) \in \pi \\ e \in E(\cyclicedgraph{1,j})}} \widetilde N^{1,j}_e \cdot \indicator_{x_1}(e) \\
                            \vdots \\
                            \sum_{\substack{(n,j) \in \pi \\ e \in E(\cyclicedgraph{n,j})}} \widetilde N^{n,j}_e \cdot \indicator_{x_n}(e)
                        \end{pmatrix}
                        =
                           \begin{pmatrix}
                            |x_1|\cdot \left( \vec{\alpha}_1 + C_0 \cdot \sum^{k}_{q=1} \vecIndEnt{u}{q}{1} \right) \\
                            \vdots \\
                            |x_n|\cdot \left( \vec{\alpha}_n + C_0 \cdot \sum^{k}_{q=1} \vecIndEnt{u}{q}{n} \right)
                        \end{pmatrix} 
                        \end{align}

                        As we "corrected" the cycle closing occurred in the transition to $\widehat{p}_{i,j}$ by adding exactly the amount of cycles of $x_i$ in one of the graphs that ensures that the number of symbols read in $\widetilde{p}$ from $x_i$ is $\left( \vec{\alpha}_1 + C_0 \cdot \sum^{k}_{q=1} \vecIndEnt{u}{q}{1} \right)$ (number of cycles of $x_i$) times the length of $x_i$. And the same for $y_j$ thus
                        
                          \begin{align}\label{eq:tilde-Nij-y}
                        \begin{pmatrix}
                            \sum_{\substack{(i,1) \in \pi \\ e \in E(\cyclicedgraph{i,1})}} \widetilde N^{i,1}_e \cdot \indicator_{y_1}(e) \\
                            \vdots \\
                            \sum_{\substack{(i,m) \in \pi \\ e \in E(\cyclicedgraph{i,m})}} \widetilde N^{i,m}_e \cdot \indicator_{y_m}(e)
                        \end{pmatrix}
                        =
                           \begin{pmatrix}
                            |y_1|\cdot \left( \vec{\beta}_1 + C_0 \cdot \sum^{r}_{s=1} \vecIndEnt{v}{s}{1} \right) \\
                            \vdots \\
                            |y_m|\cdot \left( \vec{\beta}_m + C_0 \cdot \sum^{r}_{s=1} \vecIndEnt{v}{s}{m} \right)
                        \end{pmatrix} 
                        \end{align}                      
                        \noindent
                        We turn to define $\vec{a'}$ as follows:
                        \begin{equation}\label{eq:apri}
                            \vec{a'_i} = \tfrac{\left(\vec{\alpha} + C_0 \cdot \sum^{k}_{z=1} \vecInd{u}{z} \right)_i}{\sum^n_{d=1} \left(|x_d|\cdot \left( \vec{\alpha} + C_0 \cdot \sum^{k}_{\mu=1} \vecInd{u}{\mu} \right)_d\right)}
                        \end{equation}

                        Checking that $\vec{a'} \in D_{\setX}$:
                        \begin{align*}
                            \sum^{n}_{i=1}|x_i| \cdot \vec{a'_i} &=  \sum^{n}_{i=1}|x_i| \cdot \tfrac{\left(\vec{\alpha} + C_0 \cdot \sum^{k}_{z=1} \vecInd{u}{z} \right)_i}{\sum^n_{d=1} \left(|x_d|\cdot \left( \vec{\alpha} + C_0 \cdot \sum^{k}_{\mu=1} \vecInd{u}{\mu} \right)_d\right)}  \\
                            &= \tfrac{ \sum^{n}_{i=1}|x_i| \cdot \left(\vec{\alpha} + C_0 \cdot \sum^{k}_{z=1} \vecInd{u}{z} \right)_i}{\sum^n_{d=1} \left(|x_d|\cdot \left( \vec{\alpha} + C_0 \cdot \sum^{k}_{\mu=1} \vecInd{u}{\mu} \right)_d\right)}  \\
                            &= 1.
                        \end{align*}
                        Thus $\vec{a'} \in \Delta_X$ and as $\vec{a'}$ satisfies
                        \[
                            \vec{a'} = \tfrac{1}{\sum^n_{d=1} \left(\vec{\alpha} + C_0 \cdot \sum^{k}_{\mu=1} \vecInd{u}{\mu}\right)_d} \cdot \left( \vec{\alpha} + C_0 \cdot \sum^k_{z=1}\vecInd{u}{z} \right)
                        \]
                        and $\vec{\alpha} + C_0 \cdot \sum^{k}_{z=1} \vecInd{u}{z} \in S_X$, we get that $\vec{a'} \in \Cone(U)$.\\
                        Thus $\vec{a'} \in \Delta_X \cap \Cone(U) = D_X$.

                        We finally define the assignment for the variables of $\LP_{\vec{a}, \pi}$:
                        \begin{align}\label{eq:fij}
                            f^{i,j}_e\gets \tfrac{\widetilde N^{i,j}_e}{|\widetilde p|} \quad\text{for each }(i,j)\in\pi\text{ and }e\in E(\cyclicedgraph{i,j})
                        \end{align}
                        And for $\tau$ we denote by $\boldsymbol{\gamma}$ 
                        the combination of $V$ that yields $\vec{\beta} + C_0 \cdot \sum^{r}_{s=1} \vecInd{v}{s} \in S_Y$, and define
                        \begin{equation}\label{eq:gamma}
                            \tau_{s}\gets \tfrac{\gamma_s}{|\widetilde p|} \ \ (1\le s \le r)
                        \end{equation}
                        And for $\lambda$ we define
                        \begin{equation}\label{eq:lambda}
                            \lambda \gets \tfrac{\sum^n_{d=1} \left(\vec{\alpha} + C_0 \cdot \sum^{k}_{\mu=1} \vecInd{u}{\mu}\right)_d}{|\widetilde p|}
                        \end{equation}                        \noindent
                         We check all the constraints hold:
                         \begin{enumerate}
                            \item \textbf{{Normalization}:}\quad $\displaystyle \sum_{\substack{(i,j)\in\pi\\e\in E(\cyclicedgraph{i,j})}} f^{i,j}_e = \tfrac{1}{|\widetilde p|}\sum_{\substack{(i,j)\in\pi\\e\in E(\cyclicedgraph{i,j})}} \widetilde N^{i,j}_e = 1$.
                            \item \textbf{Flow conservation in each block:} In the transition to $\widehat p$ we made every walk $\widehat p_{i,j}$ in each $\cyclicedgraph{i,j}$ a cycle, thus satisfying flow conservation. In the transition to $\widetilde p_{i,j}$ we only added cycles so flow conservation satisfaction remains.
                            \item \textbf{Source-consumption constraints:}\quad for each $i\in[n]$,
                            \begin{align*}
                                |x_i| \cdot \vec{a'_i} \cdot \lambda &=^{\eqref{eq:apri},\eqref{eq:lambda}} |x_i| \cdot \tfrac{\left( \vec{\alpha} + C_0 \cdot \sum^k_{z=1}\vecInd{u}{z} \right)_i}{\sum^n_{d=1} \left(\vec{\alpha} + C_0 \cdot \sum^{k}_{\mu=1} \vecInd{u}{\mu}\right)_d} \cdot  \\ &\phantom{=^{\eqref{eq:apri},\eqref{eq:lambda}}|x_i|}
                                \cdot
                                \tfrac{\sum^n_{d=1} \left(\vec{\alpha} + C_0 \cdot \sum^{k}_{\mu=1} \vecInd{u}{\mu}\right)_d}{|\widetilde p|} \\
                                &= \tfrac{|x_i|\cdot \left( \vec{\alpha} + C_0 \cdot \sum^k_{z=1}\vecInd{u}{z} \right)_i}{\len({\widetilde p})}\\
                                &=^{\eqref{eq:tilde-Nij-x}} \tfrac{1}{|{\widetilde p}|} \cdot \sum_{\substack{(i,j) \in \pi \\ e \in E(\cyclicedgraph{i,j})}} \widetilde N^{i,j}_e \cdot \indicator_{x_i}(e) \\
                                &=^{\eqref{eq:fij}} \sum_{\substack{(i,j) \in \pi \\ e \in E(\cyclicedgraph{i,j})}} f^{i,j}_e \cdot \indicator_{x_i}(e)
                            \end{align*}
        
                            \item \textbf{Target-consumption constraints:}\quad for each $j\in[m]$,
                            \begin{align*}
                                |y_j| \cdot b_j &=^{p.\pageref{eq:bj}} |y_j|\cdot  \sum^r_{s=1} \tau_s \cdot \vecIndEnt{v}{s}{j} \\
                                    &=^{\eqref{eq:gamma}} 
                                    \tfrac{|y_j|}{|\widetilde p|} \cdot \sum^r_{s=1} \gamma_{s} \cdot \vecIndEnt{v}{s}{j} \\
                                    &= \tfrac{|y_j|}{|\widetilde p|} \cdot \left(\vec{\beta} + C_0 \cdot \sum^{r}_{s=1} \vecInd{v}{s} \right)_j \\
                                    &=^{\eqref{eq:tilde-Nij-y}} \tfrac{1}{|\widetilde p|} \cdot \sum_{\substack{(i,j) \in \pi \\ e \in E(\cyclicedgraph{i,j})}} \widetilde N^{i,j}_e \cdot \indicator_{y_j}(e) \\
                                    &=^{\eqref{eq:fij}} \sum_{\substack{(i,j) \in \pi \\ e \in E(\cyclicedgraph{i,j})}} f^{i,j}_e \cdot \indicator_{y_j}(e).
                            \end{align*}
                        \end{enumerate}
                        Hence all constraints hold and the assignment is a feasible solution for $\LP_{a, \pi}$. Therefore   
                        \begin{align*}
                            \min_{\pi \in \Pi}\OPT_{\vec{a'},\pi} &\le \OPT_{\vec{a'},\pi}\\
                            &\le \sum_{\substack{(i,j) \in \pi \\ e \in E(\cyclicedgraph{i,j})}} f^{i,j}_e\,c(e)\\
                            &= \tfrac{1}{|\widetilde p|} \sum_{\substack{(i,j) \in \pi \\ e \in E(\cyclicedgraph{i,j})}} \widetilde N^{i,j}_e\,c(e)\\
                            &= \tfrac{\wgt|\widetilde p|}{|\widetilde p|}\\
                            &\le \tfrac{\wgt(p)}{|p|} + \tfrac{C_1}{|p|}\\
                            &\le \tfrac{\wgt(p)}{|p|} + \tfrac{C_1}{|x|}
                        \end{align*}

                    As for the distance between $\vec{a}$ and $\vec{a'}$ 
                    \begin{align*}
                        |a_i - a'_i| &\ =^{(i)}\ \left| \tfrac{\vec{\alpha_i}}{|x|} - \tfrac{\left(\vec{\alpha} + C_0 \cdot \sum^{k}_{z=1} \vecInd{u}{z} \right)_i}{\sum^n_{d=1} \left(|x_d|\cdot \left( \vec{\alpha} + C_0 \cdot \sum^{k}_{\mu=1} \vecInd{u}{\mu} \right)_d\right)}\right| 
                        \\
                            &\ \le^{(ii)}\ \vec{\alpha_i} \cdot \left| \tfrac{1}{|x|} - \tfrac{1}{\sum^n_{d=1} |x_d|\cdot\vec{\alpha}_d \ + \  \sum^n_{d=1} |x_d|\cdot \left(C_0 \cdot \sum^{k}_{\mu=1} \vecInd{u}{\mu} \right)_d} \right| + \tfrac{\left(C_0 \cdot \sum^{k}_{z=1} \vecInd{u}{z} \right)_i}{|x|} 
                            \\
                            &\ \le^{(iii)}\  \vec{\alpha_i} \cdot \left| \tfrac{1}{|x|} - \tfrac{1}{|x| + \sum^n_{d=1} |x_d|\cdot \left(C_0 \cdot \sum^{k}_{\mu=1} \vecInd{u}{\mu} \right)_d} \right| + \tfrac{\left(C_0 \cdot \sum^{k}_{z=1} \vecInd{u}{z} \right)_i}{|x|} 
                            \\
                            &\ \le^{(iv)}\ \vec{\alpha_i} \cdot \tfrac{\sum^n_{d=1} |x_d|\cdot \left(C_0 \cdot \sum^{k}_{\mu=1} \vecInd{u}{\mu} \right)_d}{|x|^2} + \tfrac{\left(C_0 \cdot \sum^{k}_{z=1} \vecInd{u}{z} \right)_i}{|x|} 
                            \\
                            &\ \le^{(v)}\ \tfrac{\sum^n_{d=1} |x_d|\cdot \left(C_0 \cdot \sum^{k}_{\mu=1} \vecInd{u}{\mu} \right)_d}{|x|} + \tfrac{\left(C_0 \cdot \sum^{k}_{z=1} \vecInd{u}{z} \right)_i}{|x|} 
                            \\
                            &\ \le^{(vi)}\ \tfrac{M}{|x|}
                    \end{align*}
                {where (i) holds by definition of $a$ and $a'$, 
                (ii) by the triangle inequality, (iii) by the definition of $x$, (iv) since
                \(
                \left(|x| \le |x| + \sum^n_{d=1} |x_d|\cdot \left(C_0 \cdot \sum^{k}_{\mu=1} \vecInd{u}{\mu} \right)_d\right)
                \),
                (v) since 
                \(
                \left(\vec{\alpha_i} \le |x| \right)
                \)
                and (vi) by taking $M$ to be} constant depending only on $C_0$, the family  of graphs $\{\cyclicedgraph{i,j}\}$ and $U$.
                    \end{proof}

            \begin{lemma}
            $\nedar(X,Y) \ge \sup_{\vec{a} \in D_{\setX}} \min_{\pi \in \Pi} \OPT_{\vec{a},\pi}$
            \end{lemma}
            \begin{proof}

            Fix $\vec{a} \in D_{\setX}$, we choose a sequence {$\vec{\psi_N}$ of elements in $S_{\setX}$} such that $|x_N|=l^{\setX}_\vec{\psi_N} \to \infty$ and $\vec{\alpha_N} = \dir(\vec{\psi_N})$ satisfy $\lim_{N \to \infty}\vec{\alpha_N} = \vec{a}$, where $x_N =w^{\setX}_{\vec{\psi_N}} \in X$. Note that this sequence exists by the definition of $D_{\setX}$.\\
            \noindent
            For each $N$ let $y_N \in Y$ such that
            \[
                \ned(x_N,y_N) \le \inf_{y \in Y} \ned(x_N, y) + \e
            \]
            and let $p_N$ be an edit path from $x_N$ to $y_N$ such that $\tfrac{\wgt(p_N)}{|p_N|} = \ned(x_N,y_N)$, thus
            \begin{equation}\label{eq: xn yn ned close to inf xn}
                \inf_{y\in Y} \ned(x_N,y) \ge \tfrac{\wgt(p_N)}{|p_N|} - \e
            \end{equation}
            
            Apply Claim~\ref{claim: Approximating every path by LP} to $p_N$ to get $\vec{a'_N}$ that satisfy
            \begin{equation}\label{eq: use of claim on ep  to LP}
                \tfrac{\wgt(p_N)}{|p_N|} +  \tfrac{C}{|x|} \ge \min_{\pi \in \Pi} \OPT_{\vec{\alpha'_N},\pi} = F(\vec{\alpha'_N}) \qquad \| \vec{\alpha'_N}- \vec{\alpha_N}\|_\infty \le \tfrac{M}{|x_N|} 
            \end{equation}
            By \autoref{claim:LP-continuous}, $\vec{a} \to \min_{\pi \in \Pi} \OPT_{\vec{a},\pi}$ is continuous.\\
            Let $\e >0$, then there exists $\delta>0$ that satisfy
            \[
                |\vec{a}_1 -\vec{a}_2| < \delta \implies |\min_{\pi \in \Pi} \OPT_{\vec{a}_2,\pi} - \min_{\pi \in \Pi} \OPT_{\vec{a}_1,\pi}| < \e.
            \]
            From some point onward $\tfrac{M}{|x_N|} < \delta$ and $\tfrac{C}{|x_N|} < \e$ thus
            \[
                \|\vec{\alpha'_{N}}- \vec{\alpha_{N}}\|_{\infty} \le \tfrac{M}{|x_N|} < \delta.
            \]
            Hence
            \begin{equation}\label{eq: OPT dN is epsilon close to OPT d'N}
                F(\vec{\alpha_N}) = \min_{\pi \in \Pi} \OPT_{\vec{\alpha_N},\pi} \le \min_{\pi \in \Pi} \OPT_{\vec{\alpha'_N},\pi} +\e = F(\vec{\alpha'_N}) + \e.
            \end{equation}
            
            Meaning that from some point onward
            \begin{equation}\label{eq: inf ned xN to OPT dN}
                \begin{split}
                 \inf_{y \in Y} \ned(x_N,y) &\ge^\eqref{eq: xn yn ned close to inf xn} \tfrac{\wgt(p_N)}{|p_N|} - \e\\ 
                        &\ge^\eqref{eq: use of claim on ep  to LP} F(\vec{\alpha'_N})  -2\e\\
                        &\ge^\eqref{eq: OPT dN is epsilon close to OPT d'N} F(\vec{\alpha_N}) -3\e. 
                \end{split}
            \end{equation}
            \noindent
            By taking $\limsup_{N\to \infty}$ on both sides and as $\vec{\alpha} \to F(\vec{\alpha})$ is continuous we get
            \begin{align*}
                \nedar(X,Y) &= \sup_{\substack{(x'_N)^\infty_{N=1} \subseteq X \\ |x'_N|\to \infty}} \limsup_{N \to \infty } \inf_{y \in Y} \ned(x'_N ,y)  \\
                &\ge \limsup_{N \to \infty } \inf_{y \in Y} \ned(x_N ,y) \\
                &\ge^\eqref{eq: inf ned xN to OPT dN} \limsup_{N \to \infty } \left\{ F(\vec{\alpha_N}) -3\e \right\}\\
                &\ge^{\ref{claim:LP-continuous}} F(\limsup_{N \to \infty}\vec{\alpha_N}) -3\e\\
                &=F(\vec{a}) -3\e.
            \end{align*}
            As this is true for every $\vec{a} \in D_{\setX}$ and $\e$ is arbitrarily small
            \[
                \nedar(X,Y) \ge F(\vec{a}) = \sup_{\vec{a} \in D_{\setX}} \min_{\pi \in \Pi} \OPT_{\vec{a},\pi}
            \]     
                
            \end{proof}

\subsubsection{From LP to edit paths}\label{sec:LPB}

{To transform a solution of the linear program into an edit path, we use the following claim.
Before stating it, we provide some intuition for the constants whose existence it guarantees.}

$D$ will be an integer that transforms the solution of the linear program to integers, so we will be able to interpret them as walks in the graph family $\{\cyclicedgraph{i,j}\}$. $L_0$ and $L$ will correspond to lengths of an edit path between {some $x\in X$ with direction} $\vec{a}$ and the direction chosen in the LP solution. $L_0$ will be  determined with respect to $\e$ to ensure that the edit path (and $x'$ as a consequence) are large enough to neglect the bounded edits needed to transform the word read in the walks to $x' \in X$. Lastly, $K_{\text{gap}}$ will be a fixed bound on the number of edit operations needed to transform the concatenations of the walks to $x'$.

\begin{claim}\label{claim: from LP to path}
   Fix $\vec{a}\in D_{\setX}$. Let $\pi^\bullet\in\Pi$ be the interleaving that attains $F(\vec{a}) = \min_{\pi\in \Pi} \OPT_{\vec{a},\pi}$. Let $(f^\bullet, \tau^\bullet, \lambda^\bullet)$ be an optimal solution to $\LP_{\vec{a},\pi^\bullet}$ with value $\OPT_{\vec{a},\pi^\bullet}$.
   There exist constants $D \in \mathbb{N}_{\ge 1}$ and $K_{\text{gap}} \ge 0$, depending only on the graph family $\{\cyclicedgraph{i,j}\}$, $\Pi$, and the generators $U,V$, such that for every $\varepsilon > 0$, there exists an integer $L_0$ where for all integers $L \ge L_0$ that are multiples of $D$, we can construct:
   \begin{itemize}
       \item A word $x'  = w^\setX_\vec{\psi'} \in X$ for some $\psi'\in S_\setX$,
       \item A word $y' \in Y$,
       \item An edit path $p': x' \leadsto y'$,
   \end{itemize}
   satisfying:
   \begin{equation}
       \tfrac{\wgt(p')}{|p'|} \le \OPT_{\vec{a},\pi^\bullet} + \e =  \min_{\pi\in \Pi} \OPT_{\vec{a},\pi} + \e = F(\vec{a}) + \e \label{eq:rigorous_ratio}
   \end{equation}
   and
   \begin{equation}
       \left| |x'| - L\lambda^\bullet \right| \le K_{\text{gap}}. \label{eq:rigorous_length}
   \end{equation}
   Furthermore, the Parikh vector of $x'$ satisfies the explicit bound:
   \begin{equation}
       \| \vec{\psi'} - L \lambda^\bullet \vec{a} \|_\infty \le K_{\text{gap}}. \label{eq:rigorous_parikh}
   \end{equation}
\end{claim}

\begin{proof}
    Since $\vec{a} \in D_{\setX} = \Cone(U) \cap \Delta_{\setX}$, there exist non-negative coefficients $\eta_1, \dots, \eta_k$ such that:
    \[
        \vec{a} = \sum_{q=1}^k \eta_q \vecInd{u}{q}.
    \]
    The linear program is rational. Thus, there exists a common denominator $D$ for the optimal variables $(f^\bullet, \tau^\bullet, \lambda^\bullet)$. We assume $L$ is a multiple of $D$.\\
    \noindent
    We define $x'$ and $y'$ explicitly using integer linear combinations of the generators. For $x' \in X$ 
    define integer coefficients $\zeta_q = \lfloor L \lambda^\bullet \eta_q \rfloor$ for $1\leq q \leq k$. Since $L$ is a multiple of $D$, $L\lambda^\bullet$ is rational, but not necessarily integer, so we floor the coefficients. Define the Parikh vector $\vec{\psi'} = \sum_{q=1}^k \zeta_q \vecInd{u}{q}$ and define $x' = w^{\setX}_\vec{\psi'}$ thus by definition, $x' \in X$.\\
    \noindent
    We bound the difference between $\vec{\psi'}$ and the LP {source consumption expression} $L \lambda^\bullet \vec{a}$ 
    \[
        \vec{\delta}_{\setX} \defeq L \lambda^\bullet \vec{a} - \vec{\psi'} = \sum_{q=1}^k (L \lambda^\bullet \eta_q - \zeta_q) \vecInd{u}{q}.
    \]
    Since $0 \le L \lambda^\bullet \eta_q - \zeta_q < 1$, we have:
    \[
        \| \vec{\delta}_{\setX} \|_\infty \le \sum_{q=1}^k \| \vecInd{u}{q} \|_\infty \defeq B_U.
    \]
    This proves \eqref{eq:rigorous_parikh} with $K_{\text{gap}} \ge B_U$.
    Furthermore, summing the components gives the length bound:
    \[
    \begin{split}
        \big| |x'| - L \lambda^\bullet \big| &= \left| \sum^n_{i=1} |x'_i| \cdot \vec{\psi_i}  - L \lambda^\bullet \cdot \sum^n_{i=1} |x'_i|\cdot  \vec{a}_i \right|\\
                    &\le \sum^n_{i=1}|x_i|\cdot \left|\vec{\psi_i}  - L \lambda^\bullet \vec{a}_i\right|   \\
                    &\le n \cdot \max_{1\le i \le n}|x_i| \cdot B_U.
    \end{split}
    \]
    This proves \eqref{eq:rigorous_length}.

    As for $y' \in Y$, define $T_s \defeq L \tau^\bullet_s$. Since $L$ is a multiple of $D$, $T_s \in \bb{N}$.
    Define $\vec{\phi}' \defeq \sum_{s=1}^r T_s \vecInd{v}{s}$ and $y' \defeq w^{\setY}_\vec{\phi'}$.
    Thus $y' \in Y$. Note that $\vec{\phi'}$ equals the LP target consumption expression
    $L \vec{b}$ exactly.

    We interpret $N^{i,j}_e = L f^{\bullet i,j}_e \in \bb{N}$ as edge multiplicities in $\cyclicedgraph{i,j}$.
    To ensure Eulerian connectivity, we add a set of correction edges $E_{\text{corr}}$ connecting each weak connected component to $(0,0)$. The size of this set is bounded by a constant $C_{\text{cyc}}$ depending only on the graphs. Thus every walk in the graphs turns into Eulerian circuit starting at $(0,0)$ and thus an edit path from $x_i$ of some power to $y_j$ of some power. 
    Let $p_{\text{flow}}$ be the concatenated Eulerian cycles according to $\pi^\bullet$ formed by these edges.
    Let $x_{\text{flow}}$ and $y_{\text{flow}}$ be the source/target strings read by $p_{\text{flow}}$.
    
    The Parikh vector of $x_{\text{flow}}$ is the sum of flow demands plus correction edges:
    \[
        x_{\text{flow}} = w^\setX_{L \lambda^\bullet \vec{a} + \vec{\Delta}_{\text{cyc}}^X}, \quad \text{where } \|\vec{\Delta}_{\text{cyc}}^X\|_\infty \le C_{\text{cyc}}.
    \]
    Similarly for $y_{\text{flow}}$:
    \[
        y_{\text{flow}} = w^\setY_{L \vec{b} + \vec{\Delta}_{\text{cyc}}^Y  }, \quad \text{where } \|\vec{\Delta}_{\text{cyc}}^Y\|_\infty \le C_{\text{cyc}}.
    \]

    We now construct $p'$ as the concatenation $x' \xrightarrow{p_1} x_{\text{flow}} \xrightarrow{p_{\text{flow}}} y_{\text{flow}} \xrightarrow{p_2} y'$.
    For $p_1$ we take pure edit path between $x'$ and $x_{\text{flow}}$.
    \[
        \wgt(p_1) = \ed(x', x_{\text{flow}}).
    \]
    \begin{align*}
        \ed(x', x_{\text{flow}}) &\le \sum_{i=1}^n |x_i| \cdot | (\vec{\psi'})_i - ( L \lambda^\bullet \vec{a} + \vec{\Delta}_{\text{cyc}}^X)_i |\\
                    &= \sum_{i=1}^n |x_i| \cdot | (L \lambda^\bullet \vec{a} - \vec{\delta}_{\setX})_i - ( L \lambda^\bullet \vec{a} + \vec{\Delta}_{\text{cyc}}^X)_i |\\
                    &= \sum_{i=1}^n |x_i| \cdot (\delta_X + \vec{\Delta}_{\text{cyc}}^X)_i \le C_1
    \end{align*}
    for some constant $C_1$.
    \noindent
    For $p_2$ we also take a pure edit path between $y_{\text{flow}}$ and $y'$. Difference in the Parikh vectors is 
    \[
         L \vec{b} + \vec{\Delta}_{\text{cyc}}^Y - \vec{\phi}' =  L \vec{b} + \vec{\Delta}_{\text{cyc}}^Y - L \vec{b} =\vec{\Delta}_{\text{cyc}}^Y
    \]
    Thus using the same arguments $\wgt(p_2) \le C_2$ for some constant $C_2$.
    \noindent
    As for our total path $p'$ 
    \begin{align*}
        \wgt(p') &= \wgt(p_1) + \wgt(p_{\text{flow}}) + \wgt(p_2) \\
                  &\le C_1 + (L \cdot \OPT_{\vec{a},\pi^\bullet} + C_{\text{cyc}}) + C_2 \\
                  &= L \cdot \OPT_{\vec{a},\pi^\bullet} + K_{\text{wgt}}.
    \end{align*}
    where $K_{\text{wgt}} = C_1 + C_2 + C_{\text{cyc}}$
    and
    \begin{align*}
        |p'| &= |p_1| + |p_{\text{flow}}| + |p_2| \\
                  &\ge |p_{\text{flow}}| \\
                  &= L + C_{\text{cyc}} \ge L.
    \end{align*}
    Therefore:
    \[
        \tfrac{\wgt(p')}{|p'|} \le \tfrac{L \cdot \OPT_{\vec{a},\pi^\bullet} + K_{\text{wgt}}}{L} = \OPT_{\vec{a},\pi^\bullet} + \tfrac{K_{\text{wgt}}}{L}.
    \]
    Choosing $L \ge \tfrac{K_{\text{wgt}}}{\varepsilon}$ ensures the ratio is 
    \[
        \tfrac{\wgt(p')}{|p'|}\le \OPT_{\vec{a},\pi^\bullet} + \e = \min_{\pi \in \Pi} \OPT_{\vec{a},\pi} + \e = F(\vec{a}) + \e.
    \]
\end{proof}

\begin{lemma}
    $\nedar(X,Y) \le \sup_{\vec{a} \in D_{\setX}} \min_{\pi} \OPT_{\vec{a},\pi} = \sup_{\vec{a} \in D_{\setX}} F(\vec{a})$.
\end{lemma}

\begin{proof}
    Let $\e>0$ and let $\{x_k = w^\setX_\vec{\psi_k}\}_{k=1}^\infty \subseteq X$ be a sequence such that
    \begin{equation}\label{eq: x_k gets nedar}
        \nedar(X,Y) \le \limsup_{k \to \infty} \inf_{y \in Y} \ned(x_k, y) + \e 
    \end{equation}
    and $|x_k| \to \infty$. Since $D_{\setX}$ is compact, we pass to a subsequence such that $\vec{\alpha_k} \defeq \dir(\vec{\psi_k})$ satisfy
    \begin{equation}\label{eq: converging dir seq}
        \lim_{k \to \infty } \vec{\alpha_k} = \vec{a^\bullet} \in D_{\setX}
    \end{equation}
    \noindent
    As $|x_k|\to \infty$ there exists some $N$ such that for each $k \ge N$, given the optimal LP solution $(f_k, \tau_k, \lambda_k)$ for $\vec{\alpha_k}$, $\frac{|x_k|}{\lambda_k} - D$ is large enough to satisfy all assumption on $L$ assumed in \autoref{claim: from LP to path} besides being a multiply of $D$.
    
    We apply \autoref{claim: from LP to path} with direction $\vec{\alpha_k}$ and parameter $\e$.
    We choose $L_k$ to be the unique multiple of $D$ satisfying:
    \begin{equation}\label{eq: Choice of L_k}
        \tfrac{|x_k|}{\lambda_k} - D < L_k \le \tfrac{|x_k|}{\lambda_k}.
    \end{equation}
    This choice ensures:
    \begin{equation}\label{eq: x_k closr to L_k lambda_k}
        0 \le |x_k| - L_k \lambda_k < D \lambda_k \le D.
    \end{equation}
    
    \autoref{claim: from LP to path} yields $x'_k \in X$, $y'_k \in Y$, and a path $p'_k: x'_k \leadsto y'_k$ satisfying
    \begin{equation}\label{eq: cost of p'k near OPT d_k}
        \tfrac{\wgt(p'_k)}{|p'_k|} \le \min_{\pi \in \Pi} \OPT_{\vec{\alpha_k},\pi} + \e = F(\vec{\alpha_k}) + \e.
    \end{equation}
    
    We construct a path $p_{\text{total}}$ from $x_k$ to $y'_k$ by concatenating the edit path $p_k: x_k \leadsto x'_k$ with $p'_k$.
    First, we bound $\wgt(p_k) = \ed(x_k, x'_k)$.
    Using the Parikh bound from \autoref{claim: from LP to path} (Eq \ref{eq:rigorous_parikh}):
    \begin{equation}\label{eq: claim bound on r'_k}
        \| \vec{\psi}_k' - L_k \lambda_k \vec{\alpha_k} \|_\infty \le K_{\text{gap}}. 
    \end{equation}
    
    Since $x_k$ is exactly on direction $\vec{\alpha_k}$, $\vec{\psi_k} = |x_k| \vec{\alpha_k}$, we get
    \begin{align*}
        \| \vec{\psi_k} - \vec{\psi'_k}\|_\infty 
        &\le
        \| |x_k| \cdot \vec{\alpha_k} - L_k \cdot \lambda_k \cdot \vec{\alpha_k} \|_\infty + \| L_k \cdot \lambda_k \cdot \vec{\alpha_k} - \vec{\psi'_k} \|_\infty \\
        &\le^\eqref{eq: claim bound on r'_k} (|x_k| - L_k \lambda_k) \|\vec{\alpha_k}\|_\infty + K_{\text{gap}} \\
        &\le^\eqref{eq: x_k closr to L_k lambda_k} D  + K_{\text{gap}}.
    \end{align*}
    Let $B = n \cdot \max_{i} |x_i| \cdot (D + K_{\text{gap}})$. Then $\wgt(p) = \ed(x_k, x'_k) \le B$. The total path $p_{\text{total}}$ has:
    \begin{equation} \label{eq: cost of p total}
        \begin{split}
        \wgt(p_{\text{total}}) &= \wgt(p_k) + \wgt(p'_k) \le B + \wgt(p'_k).\\
        |p_{\text{total}}| &= |p_k| + |p'_k| \ge |p'_k|.
        \end{split}
    \end{equation}
    Thus we get
    \begin{align*}
        \inf_{y \in Y} \ned(x_k, y) &\le \ned(x_k, y'_k) \\ 
            &\le\tfrac{\wgt(p_{\text{total}})}{|p_{\text{total}}|}\\
            &\le^\eqref{eq: cost of p total} \tfrac{B + \wgt(p'_k)}{|p'_k|} \\
            &=\tfrac{B}{|p'_k|} + \tfrac{\wgt(p'_k)}{|p'_k|} \\
            &\le^\eqref{eq: cost of p'k near OPT d_k} \tfrac{B}{|p'_k|} + F(\vec{\alpha_k}) + \e.
    \end{align*}
    
    As $k \to \infty$, we have $|x_k| \to \infty$, implying $L_k \to \infty$ by \eqref{eq: Choice of L_k}. By \autoref{claim: from LP to path}, $|p'_k| \ge |x'_k| \ge L_k$, so $|p'_k| \to \infty$.
    Thus:
    \begin{equation}\label{eq: limsup of xk is less the limsup of OPT}
        \limsup_{k \to \infty} \inf_{y \in Y} \ned(x_k, y) \le \e + \limsup_{k \to \infty} F(\vec{\alpha_k}) .
    \end{equation}
    And finally
    \begin{align*}
        \nedar(X,Y) &\le^\eqref{eq: x_k gets nedar} \limsup_{k \to \infty} \inf_{y \in Y} \ned(x_k, y) + \e \\
                    &\le^\eqref{eq: limsup of xk is less the limsup of OPT} 2\e + \limsup_{k \to \infty} F(\vec{\alpha_k}) \\
                    &\le^{\ref{claim:LP-continuous}} 2\e + F(\limsup_{k \to \infty}\vec{\alpha_k})\\
                    &=^\eqref{eq: converging dir seq} 2\e + F(\vec{a^\bullet})\\
                    &\le 2\e + \sup_{\vec{a}\in D_X} F(\vec{a}).
    \end{align*}
    And as $\e$ is arbitrarily small
    \[
        \nedar(X,Y) \le \sup_{\vec{a} \in D_X} F(\vec{a}).
    \]
\end{proof}

\claimLPcontinuous*
\begin{proof}
For each fixed $\pi$, the mapping $\vec{a}\mapsto \OPT_{\vec{a},\pi}$ is the value function of a linear program whose right-hand side depends linearly on $\vec{a}$, and is therefore piecewise linear and continuous.
Since $\Pi$ is finite, taking the pointwise minimum preserves continuity.
\end{proof}

\lemLPexp*
\begin{proof}
    
    By \autoref{thm:LP},
    \[
        \nedar(X,Y) = \sup_{\vec{a} \in D_\setX} F(\vec{a}) = \sup_{\vec{a} \in D_\setX} \min_{\pi \in \Pi} \OPT_{\vec{a}, \pi}
    \]
    where $D_\setX \subseteq \nonnegR^k$ is the compact polytope of valid directions for $X$, and $\OPT_{\vec{a}, \pi}$ is the optimal value of the linear program $\LP_{\vec{a}, \pi}$.

   The primal linear program $\LP_{\vec{a}, \pi}$ minimizes a cost subject to constraints where the variable vector $\vec{x}$ must satisfy $M \vec{x} = h(\vec{a})$ and $\vec{x} \ge 0$.
    Here, $M$ is a constraint matrix independent of $\vec{a}$, and $h(\vec{a})$ is an affine function of the direction vector $\vec{a}$.

    By the Strong Duality Theorem for linear programming, for a fixed interleaving $\pi$ and a fixed $\vec{a}$, we have:
    \[
        \OPT_{a, \pi} = \max \{ {h}(\vec{a})^\top \vec{y} \mid M^\top \vec{y} \le \vec{c} \}
    \]
    where $\vec{y}$ is the vector of dual variables and $\vec{c}$ is the vector of primal target coefficients. The dual feasible region $P^*_\pi = \{ \vec{y} \mid M^\top \vec{y} \le \vec{c} \}$ is a rational polyhedron independent of $\vec{a}$. The objective function is linear in $\vec{a}$. 
    
    The maximum of a linear function over a polyhedron is attained at one of its vertices. Let $\mathcal{V}_\pi$ be the set of vertices of the dual polyhedron $P^*_\pi$. We can express $\OPT_{\vec{a}, \pi}$ as the pointwise maximum of a finite set of affine functions of $\vec{a}$:
    \[
        \OPT_{a, \pi} = \max_{\vec{v} \in \mathcal{V}_\pi} \left( h(a)^\top \vec{v} \right)
    \]
    Consequently, the function $\vec{a} \mapsto \OPT_{\vec{a}, \pi}$ is convex and piecewise-linear. The function $F(\vec{a})$ is defined as the pointwise minimum of these convex functions:
    \[
        F(\vec{a}) = \min_{\pi \in \Pi} \max_{\vec{v} \in \mathcal{V}_\pi} \left( h(\vec{a})^\top \vec{v} \right)
    \]
    Therefore $F(\vec{a})$ is a piecewise-linear function. The supremum of a piecewise-linear function over a compact polytope $D_X$ is attained at a vertex of the arrangement formed by the boundaries of the linear regions.
    
    The linear regions of $F(\vec{a})$ are defined by the hyperplanes where two constituent affine functions intersect. Let $\mathcal{L}$ be the collection of all affine functions defining the pieces:
    \[
        \mathcal{L} = \bigcup_{\pi \in \Pi} \{ \ell_{\vec{v}}(a) = h(a)^\top \vec{v} \mid \vec{v} \in \mathcal{V}_\pi \}
    \]
    A vertex of the arrangement is the unique intersection of $k$ hyperplanes (where $k$ is the dimension of the direction space $D_\setX$), where each hyperplane is either:
    \begin{enumerate}
        \item An equality $\ell_1(a) = \ell_2(a)$ for some $\ell_1, \ell_2 \in \mathcal{L}$, or
        \item A boundary constraint of the polytope $D_X$.
    \end{enumerate}
    
    We turn to bound the complexity. Let $N$ be the input size.
    \begin{itemize}
    
        \item The number of interleavings $|\Pi|$ is exponential in $N$.
        
        \item The dimension of the dual variables and constraints  corresponds to the number of constraints and variables in the primal $\LP_{a, \pi}$, which are both polynomial in $N$. Consequently, the number of vertices $|\mathcal{V}_\pi|$ of the dual polyhedron is at most exponential in $N$.
        
        \item The total size of the set of linear functions is $|\mathcal{L}| \le |\Pi| \cdot \max_\pi |\mathcal{V}_\pi|$, which is also exponential in $N$.
    \end{itemize}
    
    The maximum number of vertices in an arrangement of $M$ hyperplanes in $\mathbb{R}^k$ is $O(M^k)$. Where $M = |\mathcal{L}|$ is exponential in $N$, and $k$ (the dimension of the generators of $X$) is linear in $N$. Thus, the number of candidate points for $\sup_{\vec{a} \in D_\setX}F(\vec{a})$ is also exponential in $N$.
    
    The algorithm proceeds by enumerating all such candidate points, evaluating $F(\vec{a})$ at each point (which requires solving $|\Pi|$ linear programs, also an exponential operation), and returning the maximum. The total time complexity is therefore exponential in the input size.
\end{proof}

\end{document}